\documentclass[prx,aps,twocolumn,superscriptaddress,nofootinbib,reprint]{revtex4-2}

\pdfoutput=1
\usepackage{dsfont,amsmath,amssymb,amsthm, bbm,graphicx,mathtools,bm,ulem,physics}
\usepackage[dvipsnames]{xcolor}
\usepackage[toc,page]{appendix}
\usepackage{enumitem,hyperref}
\usepackage{thmtools,thm-restate}
\usepackage[english]{babel}

\makeatletter
\renewcommand*\env@matrix[1][\arraystretch]{%
	\edef\arraystretch{#1}%
	\hskip -\arraycolsep
	\let\@ifnextchar\new@ifnextchar
	\array{*\c@MaxMatrixCols c}}
\makeatother

\hypersetup{
	colorlinks=true,  
	linkcolor=blue,   
	citecolor=blue,   
	urlcolor=blue     
}

\renewcommand{\eqref}[1]{Eq.~(\ref{#1})}
\newcommand{\figref}[1]{Fig.~(\ref{#1})}

\newcommand{\appref}[1]{Appendix~\ref{#1}}

\newcommand{\lemref}[1]{Lemma~\ref{#1}}
\newcommand{\thmref}[1]{Theorem~\ref{#1}}

\def\>{\rangle}
\def\<{\langle}

\def\E{ {\cal E} }
\def\P{ {\cal P} }
\def\H{ {\cal H} }
\def\M{ {\cal M} }
\def\S{ {\cal S} }

\def\U {{\cal U}}
\def\V {{\cal V}}

\def\G {{\cal G}}
\def\N{ {\cal N} }
\def\F{ {\cal F} }

\def\B{ {\cal B} }

\def\H{ {\cal H} }
\def\W{ {\cal W} }
\def\X{ {\cal X} }

\def\P{ {\cal P} }

\def\D{ {\cal D} }
\def\T{ {\cal T} }
\def\id{ \mathbbm{1} }
\def\tr{ \mbox{tr} }

\def\I{ \mathcal{I} }

\DeclareMathOperator{\vspan}{span}

\DeclareMathOperator{\supp}{supp}

\def\tr{ \mathrm{tr} }

\def\>{\rangle}
\def\<{\langle}

\renewcommand{\emph}{\textit}

\newtheorem{theorem}{Theorem}

\newtheorem{lemma}[theorem]{Lemma}
\newtheorem{definition}[theorem]{Definition}

\newtheorem{corollary}[theorem]{Corollary}

\DeclareMathOperator{\Span}{Span}

\begin{document}
	

	\title{Infinitesimal reference frames suffice to determine the asymmetry properties of a quantum system}
	\author{Rhea Alexander}
	\email{ra4518@ic.ac.uk}
	\affiliation{Department of Physics, Imperial College London, London SW7 2AZ, UK}	
	\affiliation{School of Physics and Astronomy, University of Leeds, Leeds, LS2 9JT, UK}
	
	\author{Si Gvirtz-Chen}
	\affiliation{School of Physics and Astronomy, University of Leeds, Leeds, LS2 9JT, UK}
	
	\author{David Jennings}
	\affiliation{School of Physics and Astronomy, University of Leeds, Leeds, LS2 9JT, UK}
	\affiliation{Department of Physics, Imperial College London, London SW7 2AZ, UK}
	
	\begin{abstract}
Symmetry principles are fundamental in physics, and while they are well understood within Lagrangian mechanics, their impact on quantum channels has a range of open questions. The theory of asymmetry grew out of information-theoretic work on entanglement and quantum reference frames, and allows us to quantify the degree to which a quantum system encodes coordinates of a symmetry group. Recently, a complete set of entropic conditions was found for asymmetry in terms of correlations relative to infinitely many quantum reference frames. However, these conditions are difficult to use in practice and their physical implications unclear. In the present theoretical work, we show that this set of conditions has extensive redundancy, and one can restrict to reference frames forming any closed surface in the state space that has the maximally mixed state in its interior. This in turn implies that asymmetry can be reduced to just a single entropic condition evaluated at the maximally mixed state. Contrary to intuition, this shows that we do not need macroscopic, classical reference frames to determine the asymmetry properties of a quantum system, but instead infinitesimally small frames suffice. Building on this analysis, we provide simple, closed conditions to estimate the minimal depolarization needed to make a given quantum state accessible under channels covariant with any given symmetry group.
\end{abstract}
	
	\maketitle
	
	\section{Introduction}
	
Symmetry principles have been extensively studied both in classical and quantum theory, and in particular for Lagrangian dynamics of a quantum system~\cite{noether1918invarianten}. However, such evolution is a strict subset of the most general kind of physical transformation that quantum theory permits -- quantum channels~\cite{watrous2018theory}. This more general setting not only includes unitary dynamics, open system dynamics, and the ability to vary the system dimension as it transforms, but it also interpolates between deterministic unitary dynamics and measurements that sharply collapse a quantum system. How symmetry principles constrain quantum channels is therefore a crucial question. 

The study of quantum entanglement~\cite{plenio2007entanglement,horodecki2009entanglement} lead to a much broader conception of physical properties in terms of `resources' relative to a set of quantum channels. This gave a precise way to quantify other fundamental features such as quantum coherence~\cite{aberg2006quantifying,marvian2013theory,baumgratz2014coherence,streltsov2017coherence}, thermodynamics~\cite{brandao2013athermality, horodecki2013fundamental,gour2015resource}, non-Gaussianity \cite{albarelli2018resource,takagi2018non-gaussianity},  magic states for quantum computing~\cite{veitch2014resource,howard2017application}, and many more \cite{Chitambar2019}. In particular, the theory of asymmetry provides an information-theoretic means to quantify the degree to which a quantum system breaks a symmetry \cite{marvian2013theory,Marvian2014Extending,takagi2019skew}. 

Asymmetry sits at the crossroads between abstract quantum information theory and physical laws, and quantifies what has been called `unspeakable quantum information' \cite{bartlett2007reference,Marvian2016Coherence}. This information cannot be transcribed into a data string on paper or in an email, but instead requires the transfer of a system that carries a non-trivial action of the symmetry group, via a symmetric (covariant) quantum channel. Given this, such concepts find application in quantum metrology \cite{Hall2012Nonlinear}, symmetry-constrained dynamics~\cite{Marvian2014Extending, Cirstoiu2020_Noether,Chiribella2021Symmetries}, quantum reference frames \cite{AharonovSusskind1967,Chiribella2004Efficient,Jones2006Entanglement,bartlett2007reference,gourspekkens2008resource,Vaccaro2008Tradeoff}, thermodynamics~\cite{lostaglio2015description,lostaglio2015quantum, marvian2020coherence}, measurement theory~\cite{wigner1952WAY,araki1960measurement,yanase1961optimal,marvian2012information, ahmadi2013WAY}, macroscopic coherence~\cite{yadin2016general}, and quantum speed-limits \cite{Marvian2016Speed}. More recent work has seen a renewed interest in quantum reference frames in a relativistic setting and the problem of time in quantum physics~\cite{rovelli1991quantum,rovelli1996relational,Vedral2017Evolution,Nikolova2018Relational,giacomini2019covariance,loveridge2019relative,smith2019quantizing,Martinelli2019Quantifying,mendes2019time,vanrietvelde2020change,Carmo2021Quantifying,Chataignier2021Relational}, as well as applications in quantum computing and covariant quantum error-correcting codes~\cite{EastinKnill2009Restrictions, Faist2020_Approx_QEC,Woods2020continuousgroupsof,yang2020covariant,almheiri2015bulk,pastawski2015holographic, gschwendtner2021programmability} where the Eastin-Knill theorem provides an obstacle to transversal gate-sets forming a continuous unitary sub-group~\cite{EastinKnill2009Restrictions}.

The central question considered in this paper concerns transforming from a quantum state  $\rho$ to another quantum state $\sigma$ under a symmetry constraint. More precisely, we address the following fundamental question:
\begin{center}
   \textbf{Core Question:} \textit{When is it possible to transform $\rho_A$ to $\sigma_B$ under a quantum channel between systems $A$ and $B$ that is covariant with respect to a symmetry group $G$? }
\end{center}
This question addresses the fundamental way in which symmetries constrain quantum theory, and turns out to be surprisingly non-trivial. One might initially conjecture that if we consider the generators $\{J_k\}$ of the group representation $U(g)= \exp[ i \sum_k g_k J_k]$ and compute their moments $\<J^n_k\> := \tr[ J_k^n \rho]$ for a given state that this provides an answer to the above question. However, while this intuition is correct in the case of unitary dynamics on pure states, it is false more generally \cite{Marvian2014Extending}. In the case of the rotation group, for example, it is possible for $\<J_k\>$ to both increase and decrease under rotationally symmetric operations~\cite{Marvian2014Extending,Cirstoiu2020_Noether}.

Given that a disconnect occurs between the symmetry principle and the generators of the group as observables for mixed quantum states, one might therefore conjecture that the problem requires an additional entropic accounting, and we must supplement our analysis with the von Neumann entropy $S(\rho)$ of the quantum state (or any general entropies that are a function of the spectrum of the state, such as the R\'{e}nyi entropies) to determine the solution. Again, this turns out to still be insufficient, and it has been shown that even if we consider all moments of the generators and the entire spectrum of the quantum state $\rho$, this still does not answer the above question \cite{Marvian2014Extending}. The missing asymmetry ingredient is instead a non-trivial combination of quantum information aspects and physics specific to the symmetry group. 

Recent work~\cite{gour2018quantum} has provided a complete set of necessary and sufficient conditions for asymmetry which fully determine the state interconversion structure with respect to a symmetry group $G$. However, this set of conditions turns out to not be particularly intuitive and moreover forms an infinite set of conditions that must be checked.  The present work unpacks these conditions, determines the minimal set of conditions needed and obtains conditions that could be used in practical situations.
\subsection{Main results of the paper}

The complete set of asymmetry conditions in~\cite{gour2018quantum} are framed in terms of correlations between the given quantum system $A$ and quantum frame systems $R$ that is locally in a state $\eta_R$ that transforms non-trivially under the group action. These correlations are measured via the single-shot conditional entropy measure $H_{\rm min} (R|A)$, which is a central quantity in quantum encryption~\cite{renner2005security}. However, the problem is that the complete set of conditions requires this entropy to be computed for \emph{all} possible reference frame states $\eta_R$, and so the question is whether one can reduce to a much simpler set of conditions for asymmetry theory.

The main results of this work are as follows:
\begin{enumerate}
    \item We prove that it suffices to consider any closed surface $\partial \D$ of reference frame states $\eta_R$ that contains the maximally mixed state in its interior.
    \item We prove that under $\varepsilon$--smoothing a finite number of reference states suffice to determine those states accessible under $G$--covariant channels.
    \item We prove that infinitesimal reference frames suffice to specify asymmetry, and interconversion under $G$--covariant channels is equivalent to a single entropic minimality condition at the maximally mixed state.
    \item We derive $O(d^2)$ closed conditions to estimate the minimal depolarization noise needed to make any given output state $\sigma$ accessible from a state $\rho$ under $G$--covariant channels. These essentially take the form
\begin{align}
\log ||\sigma^{\lambda}_j||_1 -  D_2 ( \rho^\lambda_j || \G(\rho)) \le f(d,p),
\end{align}
where $\sigma^\lambda_j, \rho^\lambda_j$ are asymmetry modes with respect to the group~\cite{marvian2014modes}, and $f(d,p) $ is a function depending on the output system dimension, its irrep structure, and the level of depolarization. The function $D_2 (X || Y)$ is a generalization of the Sandwiched $\alpha=2$ R\'{e}nyi divergence~\cite{muller2013quantum,wilde2014strong}.
\end{enumerate}

Results (1--3) show that the structure of reference frame states that determine the asymmetry properties of a system has a range of freedoms. In particular, result (3) is surprising because it is contrary to what is expected from previous work on this topic. Previously, it was natural to conjecture that in order to specify the asymmetry of a state one should make use of reference frame states that encode a group element as distinguishably as possible. 
Finally, result (4) exploits the general structure analysed to provide conditions, and could find application in describing symmetry-constrained quantum information in concrete settings.

\section{Symmetry constraints and relational physics}
Quantum entanglement \cite{horodecki2009entanglement} is usually understood as associated with a pre-order $\succ_e$ on quantum states, defined by a class of quantum channels called Local Operations and Classical Communications (LOCC). The set of states that can be generated under LOCC are called separable states, and any other state is then said to have non-trivial entanglement. The pre-order is defined as $\rho \succ_e \sigma$ if and only if we can transform from $\rho$ into $\sigma$ via an LOCC channel. This provides the resource-theoretic formulation of entanglement.

This general perspective on properties of quantum systems can be used in the above problem on transforming between quantum states under a symmetry constraint. Specifically, we can identify a symmetry ordering $\succ$ on quantum states, defined now by $\rho \succ \sigma$ whenever it is possible to transform from $\rho$ into $\sigma$ via a quantum channel that respects a given symmetry group $G$. The symmetry pre-order then \emph{defines} what it means for one quantum state to be more asymmetric than another with respect to the group $G$.

We can make this precise in the following way. Given a quantum system $A$, with associated Hilbert space $\H_A$, we denote by $\B(\H_A$) the space of bounded linear operators on $\H_A$. A symmetry group $G$ acts on the system via a unitary representation $U(g)$ on $\H_A$. States of $A$ are positive, trace-one operators $\rho \in \B(\H_A)$, and at the level of the density operator $\rho$ the symmetry group acts as $\U_g(\rho) := U(g) \rho U(g)^\dagger$. A quantum channel $\E : \B(\H_A) \rightarrow \B(\H_B)$ is a completely-positive, trace-preserving map~\cite{watrous2018theory} that sends states of an input system $A$ to states of some output system $B$. A quantum channel $\E$ is then said to be symmetric, or $G$--covariant, with respect to a group action if $\E(U(g) \rho U(g)^\dagger) =U(g) \E(\rho )U(g)^\dagger $ for all $g \in G$ and all states $\rho$ of the input system. Expressed purely in terms of composition of channels this amounts to
\begin{equation}
    [\U_g, \E] = 0 \mbox{ for all } g \in G.
\end{equation}

We then have that $\rho \succ \sigma$ when there is a $G$--covariant channel $\E$ such that $\sigma = \E(\rho)$. Moreover, a \emph{measure} of the system's asymmetry is any real-valued function $\M$ on quantum states such that if $\rho \succ \sigma$ then it must be the case that $\M(\rho) \ge \M(\sigma)$.

A number of measures of asymmetry have been developed, such as relative entropy measures \cite{Vaccaro2008Tradeoff,Gour2009Measuring}, the skew-Fisher information \cite{marvian2012symmetry,Marvian2014Extending,takagi2019skew}, and the purity of coherence \cite{marvian2020coherence}. Any such monotone $\M(\rho)$ provides a \emph{necessary} condition for a transformation to be possible. However, what is a harder question is whether one can determine a \emph{sufficient} set of monotones. Any such set of measures would encode all the features of the quantum system that relate to the symmetry constraint.

Very recently~\cite{gour2018quantum} just such a complete set of measures has been found, in terms of single-shot entropies. The monotones appearing in these relations are the quantum conditional min-entropies~\cite{renner2005security}, which are defined, for some state $\Omega_{RA}$ on a bipartite system $RA$, as:
\begin{align}
 H_{\mathrm{min}}(R|A)_\Omega :=  -\log \inf_{X_A \ge 0}\{ \tr[X_A] : \id_R \otimes X_A \ge \Omega_{RA} \} ,\nonumber
\end{align}\label{eq:H_min_def}
where the infimum ranges over all positive semidefinite operators $X_A$ on Hilbert space $\H_A$. For any state $\rho_A$ a complete set of measures is then given by
\begin{align}
H_\eta (\rho) \coloneqq H_{\mathrm{min}}(R|A)_{\G(\eta \otimes \rho)}.
\label{eq:H_eta_def}
\end{align}
where $\eta_R$ is an arbitrary quantum state on an \emph{external} reference frame system $R$, and the single-shot entropy is evaluated on the bipartite state
\begin{equation}
\G(\eta_R \otimes \rho_A) := \int \!\! dg \,\, \U_g( \eta_R) \otimes \U_g (\rho_A).
\end{equation}
In terms of transformations between quantum states under a $G$--covariant channel, we now have the following result. 
\begin{theorem}~\cite{gour2018quantum}
Let $A$, $B$, and $R$ be three quantum systems with respective Hamiltonians $H_A$, $H_B$ and $H_R$ and dimensions $d_A$, $d_B$ and $d_R$. Furthermore, let the reference system $R$ be such that $d_R=d_B$ and $H_R = -(H_{B})^T$. The state transformation $\rho_{A} \rightarrow \sigma_{B}$ is possible under a $G$--covariant operation if and only if
\begin{equation}
  \Delta  H_\eta \ge 0,
\end{equation}
for all states $\eta_R$ on $\H_R$, where we have defined
\begin{align}
\Delta H_\eta =  H_\eta(\sigma) - H_\eta(\rho),
\end{align}
as the difference in entropy between input and output systems.
\label{thrm:gour}
\end{theorem}
As shown in the original paper, the infinite set of entropic conditions outlined in \thmref{thrm:gour} can be reformulated as a semi-definite program that can be solved efficiently for sufficiently low-dimensional quantum systems. However, for larger system sizes it quickly becomes computationally intensive. Moreover, without simplification, working with these expressions analytically is not an option, since we have an infinite set of conditions and the physics involved remains hidden.

\subsection{Appearance of relativistic features in the quantum-information framework}
In quantum gravity, one has the Wheeler-de Witt equation \cite{dewitt1967} that provides a Hamiltonian constraint ${H |\Psi\> = 0}$ for a global wavefunction $|\Psi\>$. This in particular implies global time-translation covariance, and raised the question of how observed dynamics are consistent with this condition. One answer to this question was presented by Page and Wootters \cite{PageWooters1983}, who argued that time-evolution of subsystems should be properly viewed in terms of relational correlations between subsystems.

The above complete set of entropic conditions for general covariant transformations has links with this formalism. In particular, the following features appear from the quantum information-theoretic treatment when specialised to $G$ being the time-translation group~\cite{gour2018quantum}:
\begin{itemize}
\item External reference frame systems automatically appear in the information-theoretic analysis.
    \item The Hamiltonians on $R$ and $B$ obey $H_R + H_B^T = 0$ as matrices.
    \item The properties of any system $A$ that transform non-trivially under the symmetry group are fully described by correlations between $R$ and $A$.
    \item A \emph{single-shot} Page--Wootters condition emerges in the classical reference frame limit in terms of optimal guessing-probabilities of the time parameter.
    \item Local gauge symmetries can be formulated with a causal structure on asymmetry resources~\cite{Cirstoiu2017Global}.
\end{itemize}
The appearance of a Page--Wootters condition is surprising. The classical limit here is when reference frame $\eta$ acts as a good clock, in the sense that one can encode the classical information $t$ into it in such a way that one can discriminate between two different values $t_1$ and $t_2$ with high probability.

For this regime, the state $\Omega_{RA}=\G (\eta_R \otimes \rho_A)$ tends to a classical-quantum state, with $R$ behaving as a classical `register' for $t$. However, it can be shown that for a classical quantum state $\Omega_{RA}$, where $R$ is classical, the single-shot entropy corresponds exactly to an optimal guessing probability~\cite{konig2009operational}. More precisely, it can be proved that
\begin{equation}
    -\log H_{\rm min} (R|A)_\Omega = p_{\rm opt} (t),
\end{equation}
where $p_{\rm opt}(t)$ is the optimal guessing probability for the value of $t$, over all generalised POVM measurements, given the state $U(t) \rho U(t)^\dagger$. This provides a refinement of the Page--Wootters formalism.

This interpretation of the $H_{\rm min} (R|A)_\Omega$ terms extends to arbitrary states on $R$. In the fully general case it quantifies the optimal singlet fraction~\cite{konig2009operational}, the degree to which the state $\Omega_{RA}$ can be transformed to a maximally entangled (perfectly correlated) state through action on $A$ alone. It is also possible to include thermodynamics into this setting without much complication. For this extension, if we consider varying the state $\eta_R$ we also smoothly interpolate between free energy--like conditions and clock conditions~\cite{gour2018quantum}.

The above features come solely from the single-shot quantum-information formalism of the problem, and show that these aspects are fundamental. It therefore motivates a deeper analysis of the complete set of entropic conditions with the aim of unpacking the physical content and determining the minimal information-theoretic conditions that describe fully general symmetric transformations of quantum systems.

\subsection{Warm-up example: a curious dependence on reference frame states}
It is useful to first illustrate special cases of the $H_{\rm min}(R|A)$ conditions for the elementary case of channels sending a single qubit to a single qubit under a $U(1)$ covariant symmetry constraint. For concreteness we take this to be time-translation $U(t) = e^{itH_A}$ under a qubit Hamiltonian $H_A = \sigma_z$, where $(\sigma_x,\sigma_y,\sigma_z)$ are the Pauli matrices for the qubit system.

We must therefore consider an auxiliary qubit reference frame $R$ in a state $\eta_R$ with $H_R = - \sigma_z$, and compute the conditional min-entropy on the joint state
\begin{equation}
\G (\eta_R \otimes \rho_A) = \int_0^{2\pi} \!\! \frac{dt}{2\pi} \, \U_t( \eta_R) \otimes \U_t( \rho_A).
\end{equation}
We also choose $\rho_A = \frac{1}{2}(\id + \frac{1}{2} \sigma_x + \frac{1}{2}\sigma_z)$ and look at how each choice of reference frame state $\eta_R$ constrains the region of quantum states accessible under time-covariant quantum channels. For fixed input state $\rho$, we define the set
\begin{equation}
\T_\eta := \{ \sigma_A : H_\eta(\sigma) \ge H_\eta(\rho) \}.
\end{equation}
In other words, $\T_\eta$ is the region of quantum states the reference frame state $\eta_R$ classes as admissible for time-covariant transformations. As such, $\sigma \in \T_\eta$ constitutes a necessary, but not sufficient, condition on the pre-order $\rho \succ \sigma$. 

The natural first choice for a reference frame state is $\eta_R = |+\>\<+|$; with a uniform superposition over energy eigenstates, this is in a sense the `best' clock state one can find for the qubit in that it can encode a single bit of data about the parameter $t$, which is the maximum allowed by the Holevo bound~\cite{watrous2018theory}. The region $S_\eta$ in this case is plotted in \figref{fig:single_entropy_qubit}(c).

\begin{figure}[t]
	\centering
	\includegraphics[scale=0.9]{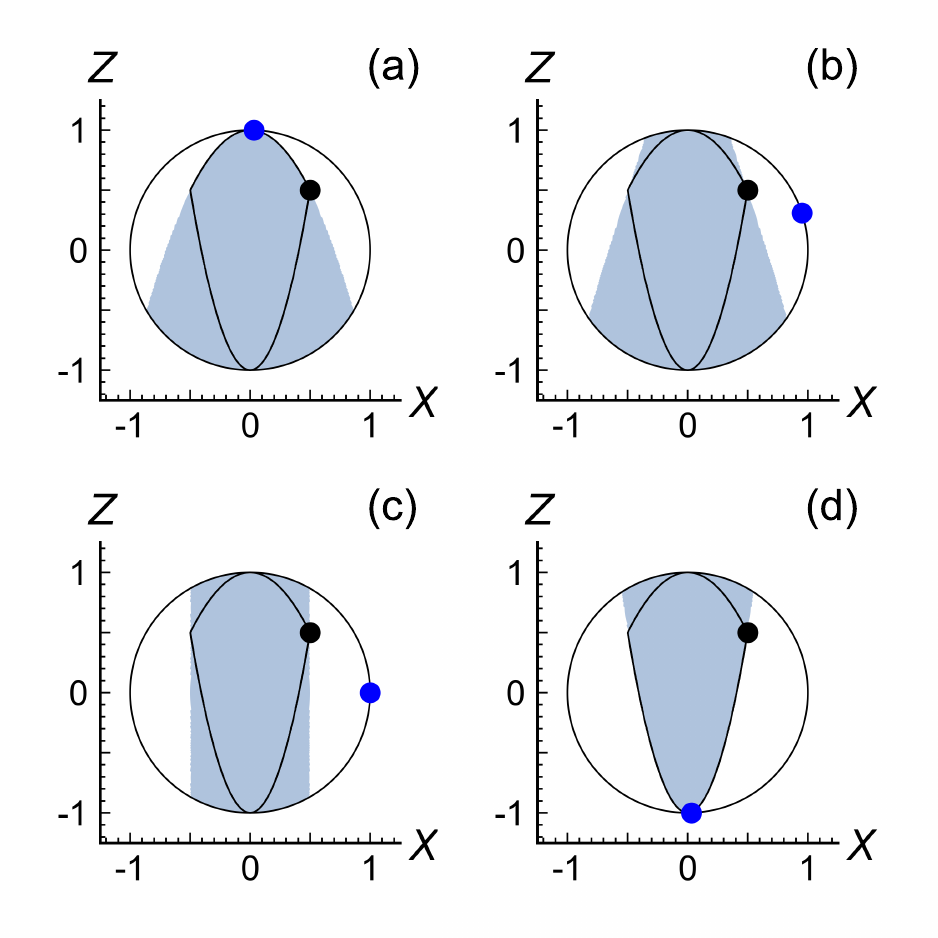}
	\caption{\textbf{``Less is more":~near--symmetric reference frame states are optimal}. We consider a qubit system under $G=U(1)$ symmetry. Given a reference frame state $\eta_R$ (blue dot), the shaded region $\T_\eta$ corresponds to potential states in the Bloch sphere which $\eta_R$ classes as accessible from the state $\rho$ (the black dot) under a covariant transformation. The black curve marks the boundary of all states that are covariantly accessible from $\rho$, and is obtained from the intersection of all regions $\T_\eta$.  Suprisingly, the high coherence state $\eta_R = |+\>\<+|$ gives a weak bound (figure (c)), while in contrast taking $\eta_R$ very close to $|0\>\<0|$ or $|1\>\<1|$ provides complete constraints (the combination of figures (a) and (d)). }
\label{fig:single_entropy_qubit}
\end{figure}

We now consider other choices of reference frame states, and find a surprising result. If we take $\eta_R = |\varphi \>\<\varphi| $ with $|\varphi\>$ approaching either $|0\>$ or $|1\>$, the region $\T_\eta$ provides a \emph{better} approximation to the actual region of quantum states accessible under time-covariant channels. This is shown in \figref{fig:single_entropy_qubit}.

We also find the following striking result: the accessible region is exactly recovered if we consider just \emph{two} reference frame states $\eta_0 = |\varphi_0\>\<\varphi_0|$ and $\eta_1 = |\varphi_1\>\<\varphi_1|$, where $\eta_0$ is infinitesimally close to $|0\>\<0|$ and $\eta_1$ infinitesimally close to $|1\>\<1|$ (a proof may be found in Appendix \ref{appx:qubit_u1_phi_pure}). However, if we took the reference frame states to be exactly equal to these pure symmetric states, then $\T_\eta$ is the entire Bloch sphere, and the constraints determined by the reference frame state disappear completely! 

This suggests that the constraints coming from the reference frame, via the correlations in the state $\Omega(\rho_A)$, have a non-trivial and counter-intuitive dependence on the state $\eta_R$. However, this simple example also illustrates that there are significant redundancies in the entropic set of conditions -- we have reduced from having to compute infinitely many conditions to just two conditions. The question then becomes whether such features carry over to more general situations, and to what degree can we reduce the set of reference frames so as to determine the minimal relational data needed to specify the asymmetry of a quantum system with respect to a general group $G$.

\section{Sufficient surfaces of reference frame states for a general group $G$}
We shall begin our analysis by reducing the set of reference frames needed substantially and establishing high-level results. These shed light on the structure of the problem and lead to our tractable set of conditions in Section~\ref{depolarize}.

\subsection{Basic reference frame redundancies}
\label{sec:initial_redundancies}

The entropic conditions in Theorem~\ref{thrm:gour} are over-complete, and contain a large number of redundancies.
Firstly, given two reference states $\eta_0$ and $\eta_1$ on $R$ such that
\begin{align}
    \eta_1 = \V (\eta_0), 
\end{align}
for some unitary channel $\V_R$ such that 
\begin{equation}
[\V_R \otimes \text{id}_A,\G]=0, 
\label{eqn:commutant}
\end{equation}
then it can be shown (see \appref{appx:invariance_unitaries_com_acom}) that ${\Delta H_{\eta_0} = \Delta H_{\eta_1}}$ for all possible input and output states. This invariance is a special case of the following lemma, which we prove in  \appref{appx:invariance_unitaries_com_acom}:
\begin{restatable}[]{lemma}{covisometryinvariance}
\label{lemma:cov_isometry_invariance}
   Let $\V_R: \B(\H_R) \rightarrow \B(\H_{R'})$ and $\W_A: \B(\H_A) \rightarrow \B(H_{A'})$ be local isometries that jointly commute with the $G$-twirl, i.e.
    \begin{align}
        [\V_R \otimes \W_A, \G] = 0.
    \end{align}
    We then have:
    \begin{align}
        H_{\V(\eta)}(\W(\rho)) = H_\eta(\rho)
    \end{align}
    for any pair of quantum states $\eta$ and $\rho$ on $\H_R$ and $\H_A$ respectively.
\end{restatable}

A second set of redundancies comes from considering the asymmetric modes of the input state. Let the asymmetric modes of a state $\rho$ be denoted modes($\rho$). Then it is known~\cite{marvian2014modes} that \begin{align}
   \text{ modes}(\sigma)\subseteq \text{modes}( \rho), 
   \label{eq:modes_sig_necessary_condition}
\end{align}is a necessary condition for a $G$--covariant transition from $\rho$ to $\sigma$. Moreover, if we assume that the condition~\eqref{eq:modes_sig_necessary_condition} holds, then it suffices to range only over reference frame states $\eta$ such that modes($\eta$) = modes($\rho$) (a proof is given in \appref{appx:modes_rho_modes_eta_redundency}).

\subsection{Necessary and sufficient surfaces of reference frame states}

 It turns out that the entropic set of conditions have a more non-obvious kind of redundancy. Any reference state $\eta$ can be written in an orthogonal basis of Hermitian operators $\{\frac{\id}{d} , X_1, \dots ,X_{d^2-1}\}$ as 
 \begin{align}
     \eta(\bm{x}) \coloneqq \frac{\id}{d} + \sum_{k=1}^{d^2-1} x_k X_k,
\label{eq:eta_bloch_x}
 \end{align}
 where $x_k \in \mathbb{R}$ and $\norm{X_k}_\infty = \frac{1}{d}$ for all $k \in \{ 1, \dots, d^2-1 \}$ provide coordinates for the state. With this in mind, we now have the following result:
\begin{restatable}[]{theorem}{sufficientsurface} (Sufficient surfaces of states).
Let all states and systems be defined as in \thmref{thrm:gour}. Let $\partial \D$ be any closed $d^2-2$ dimensional surface in the state space of $R$ that has $\frac{\id}{d}$ in its interior. The state transformation $\rho \rightarrow \sigma$ is possible under a $G$--covariant channel if and
only if
\begin{equation}
  \Delta  H_\eta \ge 0, 
\end{equation}
for all reference frame states $\eta \in \partial \D$.
\label{thm:ball_sufficient_refs}
\end{restatable}
A proof is given in \appref{appx:proof_sufficient_surface}, and follows from fact that the conditional min-entropies behave particularly nicely under the application of a partial depolarizing channel on the reference system (see Lemma \ref{lemma:partially-depol1}). 

Combined with the redundancies of Sec.~\ref{sec:initial_redundancies}, we have that only a subset of $\partial\D$ will produce non-trivial constraints -- namely the intersection of $\partial \D$ with states having asymmetric modes, quotiented by the action of the unitary sub-group of channels $\V_R$ obeying Equation (\ref{eqn:commutant}).

\begin{figure}[t]
	\centering
	\includegraphics[width=.9\linewidth]{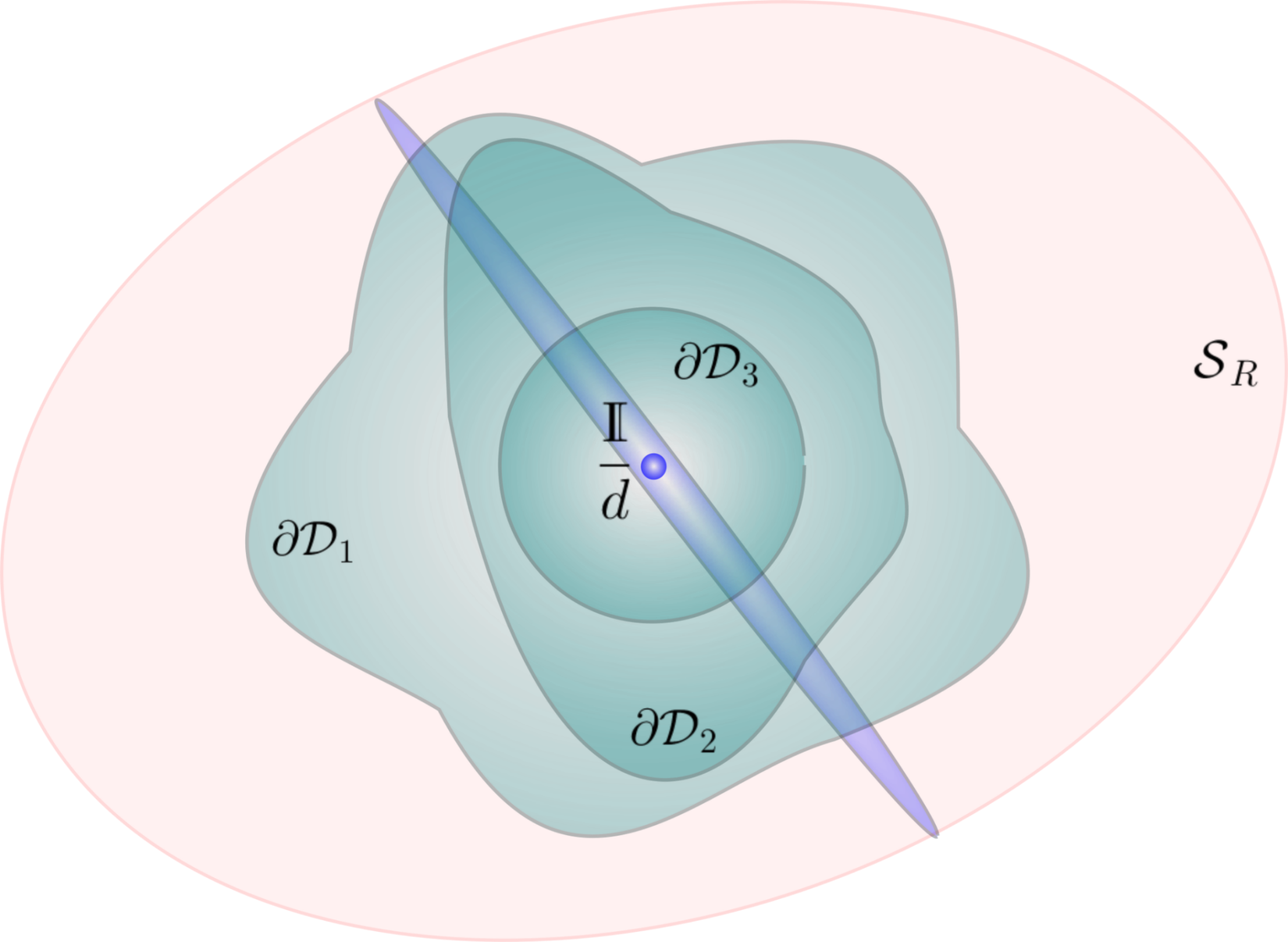}
	\caption{\textbf{(Sufficient surfaces of reference frames)}. There is extensive freedom in the choice of sufficient reference states. According to Theorem \ref{thm:ball_sufficient_refs}, any surface inside the set of all reference states $\S_R$ that encloses the maximally mixed state is a sufficient set of reference frames -- the three surfaces shown all provide the same information. The blue region is the set of symmetric states.}
\end{figure}

\subsection{A finite set of reference states under $\varepsilon$--smoothing}

It is natural to consider a `smoothed' version of the asymmetry conditions in which we are limited to some $\varepsilon$-ball resolution around states~\cite{renner2005security, tomamichel2012framework}, where an \textit{$\varepsilon$-ball} around a state $\tau$ is defined by 
\begin{equation}
\B_\varepsilon (\tau) \coloneqq \{ \tilde{\tau}\in \B(\H) :  D(\tau , \tilde{\tau}) \le  \varepsilon \}, 
\end{equation}
and $D$ is the generalized trace distance~\cite{tomamichel2010duality}. This is physically motivated by the fact that in any practical experimental scenario, two states that are in a sense close cannot be distinguished up to some finite precision in the measurement apparatus. Here we prove that, if we allow for some $\varepsilon$-probability of error in the transformation, then we can restrict to a finite set of reference frame states. 

To perform the $\varepsilon$-smoothing over the reference system, we need the following lemma, which states that the entropic relations are continuous functions of the reference frame state.
\begin{restatable}[]{lemma}{coarsegrain}
\label{lemma:coarse_grain_less_than_r}
For any $\tilde{\eta} \in \B_\varepsilon (\eta) $, we have 
\begin{align}
   | \Delta H_{\tilde{\eta}} - \Delta H_{\eta}| \le \frac{2 d_R^2 \varepsilon}{\ln 2(1-2\varepsilon)}.
\label{eq:DeltaS_delta_epsilon}
\end{align}
Note that if we further restrict $\tilde{\eta}$ to normalised states and define $r(\varepsilon) \coloneqq \frac{2 d_R^2 \varepsilon}{\ln 2}$ this simplifies to
\begin{align}
 | \Delta H_{\tilde{\eta}} - \Delta H_{\eta}| \le r(\varepsilon) .
\label{eq:DeltaS_delta_epsilon_normalised}
\end{align}
\end{restatable}
\begin{proof}
A proof is given in \appref{appx:Hmin_relations_continuity}.
\end{proof}

Therefore, an $\varepsilon$-variation in the choice of the reference state $\eta$ corresponds to an $\varepsilon$-small variation in the entropy difference $\Delta H_\eta$. 

A corollary of \thmref{thm:ball_sufficient_refs} is that it is sufficient to consider only reference frame states of the form \eqref{eq:eta_bloch_x} that live on the surface of a sphere about the maximally mixed state. It can be shown \cite{tkocz2019introduction,Ledoux1991} that there exists a finite $\varepsilon$-net $\{ \eta_1, \eta_2, \dots, \eta_N \}$ covering this set of reference states. This gives rise to the following result:
\begin{restatable}[]{theorem}{epsilonfinitechecks} \label{thm:epsilon_finite_checks} Given any $\varepsilon >0$ resolving scale, there is a finite set of reference frame states $\N \coloneqq \{\eta_k \}_{k=1}^N$ with $N = O \left ( (1+\frac{1}{\varepsilon})^{d^2-1} \right)$ that leads to the following cases: 
\begin{itemize}
\item if $\Delta H_{\eta_k} <0$ for any $\eta_k\in \N$ then $\rho \rightarrow \sigma$ is forbidden under all $G$--covariant quantum channels.
\item if $\Delta H_{\eta_k} \ge r(\varepsilon)$ for all $\eta_k\in \N$ then $\rho \rightarrow \sigma$ under a $G$--covariant quantum channel.
\item For each $\eta_k \in \N$ that has $0\le \Delta H_{\eta_k} < r(\varepsilon)$ we can obtain $O(\varepsilon)$ upper bound estimates of the minimal asymmetry resources needed to realise the transformation.
\end{itemize} 
\end{restatable}
A proof is provided in \appref{appx:smoothing_proof}.

This implies that the entropic relation can be checked on a finite number of reference frames states, and furthermore each individual reference frame state can give us some information. 

The above result is potentially of interest in numerical studies of low-dimensional systems. However, it does not shed much additional light on the structure of covariant state transformations. Therefore, instead of developing this line further here, we look at a limiting regime of reference frame states that make things clearer. This in turn lead to a more user-friendly set of conditions for `smoothed' interconversions in Section~\ref{depolarize}.

\section{Infinitesimal reference frames and a single minimality condition for asymmetry}

While we have a finite number of conditions for smoothed asymmetry, these are, by construction, of only approximate validity, and are not very physically informative. The surface condition of Theorem \ref{thm:ball_sufficient_refs} reduces the problem significantly, but still leaves us with an infinite set of reference frames to check.

However we can always take the region $\D$ in \thmref{thm:ball_sufficient_refs} to be an arbitrarily small region around the maximally mixed state, and so restrict to reference frame states that are arbitrarily close to being trivial, and this does not affect the completeness of the set of reference frames. This then shows that the counter-intuitive features we highlighted in the qubit case are in fact generic and appear for any dimension and any group action.

A statement of this is as follows. 
\begin{theorem}
\label{thm:local_min_max_mixed_condition}
Given systems $A$ and $B$, it is possible to transform a quantum state $\rho_A$ of $A$ into a state $\sigma_B$ of $B$ under a $G$--covariant quantum channel if and only if $\Delta H_\eta$ has a \emph{local minimum} at $\eta_R = \frac{\id}{d}$.
\end{theorem}
\begin{proof}	
	We first note that $\Delta H_\eta = 0$ whenever $\eta_R$ is symmetric (see Lemma~\ref{lemma:phi_dephased_states} in the appendices), and therefore $\Delta H_\eta = 0$ when $\eta_R = \frac{\id}{d}$. If we assume $\rho_A \xrightarrow{G} \sigma_B$, namely $\rho_A$ can be transformed into $\sigma_B$ via a $G$--covariant channel, then Theorem \ref{thrm:gour} implies that $\Delta H_\eta$ has a global minimum at $\eta_R = \frac{\id}{d}$, which must therefore be a local minimum as well. Conversely, if we assume $\Delta H_\eta$ has a local minimum at $\eta_R = \frac{\id}{d}$, then there exists a neighbourhood $\D$ around $\eta_R = \frac{\id}{d}$ in which $\Delta H_\eta \ge \Delta H_{\id/d}= 0$. The conditional entropies are continuous in $\eta_R$, so we have $\Delta H_\eta \ge 0$ on $\partial \D$ as well. We conclude by Theorem \ref{thm:ball_sufficient_refs} that $\rho_A \xrightarrow{G} \sigma_B$. Therefore, $\rho_A \xrightarrow{G} \sigma_B$ if and only if $\eta_R = \frac{\id}{d}$ is a local minimum of $\Delta H_\eta$. 
\end{proof}
This result is surprising, since we would expect that `optimal' information would be obtained by evaluating the entropic relations on reference states that are closest to being a ``classical" reference frame~\cite{bartlett2007reference}, namely a state $|\phi\>$ whose orbit $ |\phi(g)\> := U(g) |\psi\>$ under $G$ encodes all group elements completely distinguishably in the sense that
\begin{align}
    \braket{\phi(g)}{\phi(g')} = \delta(g^{-1}g'),
\end{align}
for all $g, g'\in G$. 

The use of such a reference frame $\ket{\phi}$ allows us to `relativise' all symmetries and construct \emph{covariant} versions of every aspect of quantum theory~\cite{bartlett2007reference,bartlett2009quantum,marvian2012symmetry,Loveridge2017Relativity,loveridge2018symmetry,Loveridge2020WAY} This is done via a relativising map 
\begin{align}
A \rightarrow \tilde{A} \coloneqq \int dg\ \U_g(A) \otimes \ketbra{\phi(g)},
\end{align}
which can be viewed as a quantum to classical-quantum channel. In the limit of a classical reference frame with $\bra{\phi(g)}\ket{\phi(g')}=\delta(g^{-1}g')$, the mapping becomes reversible via a readout from the classical register. However, for reference frame states that are not classical, the encoding is fundamentally noisy and so it is expected that the asymmetry features of a quantum state are not properly described within the encoding. Theorem~\ref{thm:local_min_max_mixed_condition}, however, tells us this is not the case. 

This ability to restrict relational data to the case of infinitesimally small reference frames suggests that asymmetry theory admits a differential geometry description in terms of tangent space of operators at the maximally mixed state. Given that the interconversion of states under $G$--covariant channels corresponds to a local minimum condition we might also expect that asymmetry is described by information geometry \cite{hayashi2006quantum,bengtsson2017geometry}, and a single curvature computed from the $H_{\rm min}(R|A)$ entropy. This would imply that the asymmetry properties of a system are fully described by a form of quantum Fisher information \cite{marvian2012symmetry}. 

We find that, for our warm-up example of $G=U(1)$ on a qubit, something like this does indeed occur. We show in Appendix \ref{appx:qubit_u1_phi_pure} that
\begin{equation}
\rho \xrightarrow{U(1)} \sigma \mbox{ if and only if } \partial_\theta^2 (\Delta H_\eta) \ge 0, \mbox{ at } \theta = 0,\pi,
\end{equation}
where $\theta$ is the angle the Bloch vector of $\eta$ makes with the $Z$--axis. Therefore we reduce the problem down to checking just two conditions, framed as a curvature term in the angular direction. However, in the radial direction one does not have a smooth variation. Instead, we conjecture that under $\varepsilon$--smoothing a complete curvature condition exists in all directions with the angular directions providing the non-trivial constraint. In the next section we give explicit details on this troublesome radial behaviour.

\subsection{Conical behaviour at the maximally mixed state}
We now consider the behaviour of the $H_\eta$ entropies in the neighbourhood of the maximally mixed state. Once again, we characterise reference frame states using the co-ordinate system in Equation \ref{eq:eta_bloch_x} as $\eta(\bm{x})$, where the maximally mixed state is located at $\bm{x}  = \mathbf{0}$. We further define
\begin{align}
	\Phi_\tau(\bm{x}) \coloneqq 2^{-H_{\eta(\bm{x})}(\tau)}.
	\label{eq:phi_H_relationship}
\end{align}
and $\tilde{\Phi}_\tau(\bm{x}) \coloneqq \Phi_\tau(\bm{x}) - \Phi_\tau(\mathbf{0})$, which gives the difference in $\Phi_\tau$ between the maximally mixed state to the reference state at $\bm{x}$. 

As a result of Lemma \ref{lemma:partially-depol}, $\tilde{\Phi}_\tau(\bm{x})$ has the following properties:
\begin{restatable}[]{lemma}{cusp}
	\label{lemma:cusp}
	Let $\lambda \geq 0$. Then for all $\bm{x}, \lambda \bm{x} \in \S$, where $\S$ is the set of all co-ordinates corresponding to reference states: 
	\begin{align}
	\tilde{\Phi}_\tau(\lambda\bm{x}) = \lambda \tilde{\Phi}_\tau(\bm{x})
	\end{align}
	Furthermore,
    \begin{align}
        \tilde{\Phi}_\tau(\bm{x}) \geq 0.
    \end{align}
\end{restatable}
\begin{proof}
A proof can be found at Appendix \ref{appx:cusp}.
\end{proof}
We conclude from the above lemma that $\tilde{\Phi}_\tau(\bm{x})$, and consequently $\Phi_\tau(\bm{x})$, is linearly non-decreasing in every direction out of the maximally mixed state. This means $\Phi_\tau(\bm{x})$ will, in general, have a conical form at the maximally mixed state; as a result, unless $\Delta \Phi(\bm{x})$ is completely linear, it too will have a conical form at the maximally mixed state. Using the defining relationship between $\Phi_\tau(\bm{x})$ and the min--entropy, we further derive from Lemma \ref{lemma:cusp} that, for sufficiently small $\varepsilon \geq 0$, 
\begin{align}
\Delta H(\varepsilon\bm{x}) = d\varepsilon \Delta \Phi(\bm{x}) + O(\varepsilon^2), 
\end{align}
where $\Delta \Phi(\bm{x}) \coloneqq \Phi_\rho(\bm{x}) - \Phi_\sigma(\bm{x})$.
In the neighbourhood of the maximally mixed state, the behaviour of $\Delta H_\eta$ is thus given by that of $\Delta \Phi$, and so in this single-shot regime we do not in general have smooth behaviour.

\subsection{Structure of $\Phi_\tau(x)$ for simple cases}
To illustrate this in practice, we now provide two examples on a qubit system for $G=U(1)$ and $G=SU(2)$.
\subsubsection{The case of time--covariant $U(1)$}
\label{sec:qubit_u1}
We first present $\Phi_\tau$ for time-covariant transformations in a non-degenerate qubit. Consider a qubit with the Hamiltonian $\sigma_z$. The states of this qubit are restricted to transforming among each other exclusively via channels that commute with all time translations $\{\U_t:t\in[0,2\pi]  \}$, where $\U_t(\cdot) =e^{-i\sigma_zt}(\cdot)e^{i\sigma_zt}$. This set of time-translations form a unitary representation of the group $U(1)$. 

We parameterise an arbitrary state $\tau$ of this qubit in its energy eigenbasis as:
\begin{align}
\tau \coloneqq \begin{pmatrix}
p_\tau& c_\tau\\
c^*_\tau& 1-p_\tau
\end{pmatrix}
\label{eq:qubit_tau_parameterisation}
\end{align}
We further use the (scaled) Pauli operators $\{\frac{\id}{2}, \frac{\sigma_x}{2},\frac{\sigma_y}{2},\frac{\sigma_z}{2}\}$ as our basis for characterising reference frame states according to Equation \ref{eq:eta_bloch_x}. The Bloch vector of a state, $(x,y,z)$, gives its co-ordinates in this basis. A direct computation (see Appendix \ref{appx:qubit_u1_phi_derivation}) gives
\begin{subequations}
\label{eq:qubit_u1_phi}
\begin{equation}
\Phi_\tau(x,y,z)= \frac{\abs{c_\tau}^2}{1-p_\tau} \frac{x^2+y^2}{4z} + \frac{z}{2} + \frac{1}{2}
\end{equation}
for the region $0 \le \frac{\sqrt{x^2+y^2}}{2z} \le \frac{1-p_\tau}{\abs{c_\tau}}$, and
\begin{equation}
\Phi_\tau(x,y,z)= \left(p_\tau - \frac{1}{2}\right)z+\abs{c_\tau}\sqrt{x^2+y^2} + \frac{1}{2}
\end{equation}
for the region $\frac{\sqrt{x^2+y^2}}{2z} \ge \frac{1-p_\tau}{\abs{c_\tau}}$ and $ \frac{\sqrt{x^2+y^2}}{2z} \le -\frac{p_\tau}{\abs{c_\tau}}$, and finally
\begin{equation}
\Phi_\tau(x,y,z)=-\frac{\abs{c_\tau}^2}{p_\tau} \frac{x^2+y^2}{4z} - \frac{z}{2} + \frac{1}{2}
\end{equation}
\end{subequations}
for the region $0 \ge \frac{\sqrt{x^2+y^2}}{2z} \ge -\frac{p_\tau}{\abs{c_\tau}}$. When neither $\rho$ nor $\sigma$ is symmetric, one can find neighbourhoods around the poles of the Bloch sphere in which $\Delta \Phi(\bm{x})$ is not completely linear if and only if $\frac{\abs{c_\rho}^2}{1-p_\rho} = \frac{\abs{c_\sigma}^2}{1-p_\sigma}$ and $\frac{\abs{c_\rho}^2}{p_\rho} = \frac{\abs{c_\sigma}^2}{p_\sigma}$. Since these conditions are equivalent to $\rho = \U_t(\sigma)$ for some $t$, for arbitrary choices of $\rho$ and $\sigma$ we almost always expect a conical singularity in $\Delta H_\eta$ at $\eta = \frac{\id}{d}$. 

\begin{figure}[t]
\centering
\includegraphics[scale=0.9]{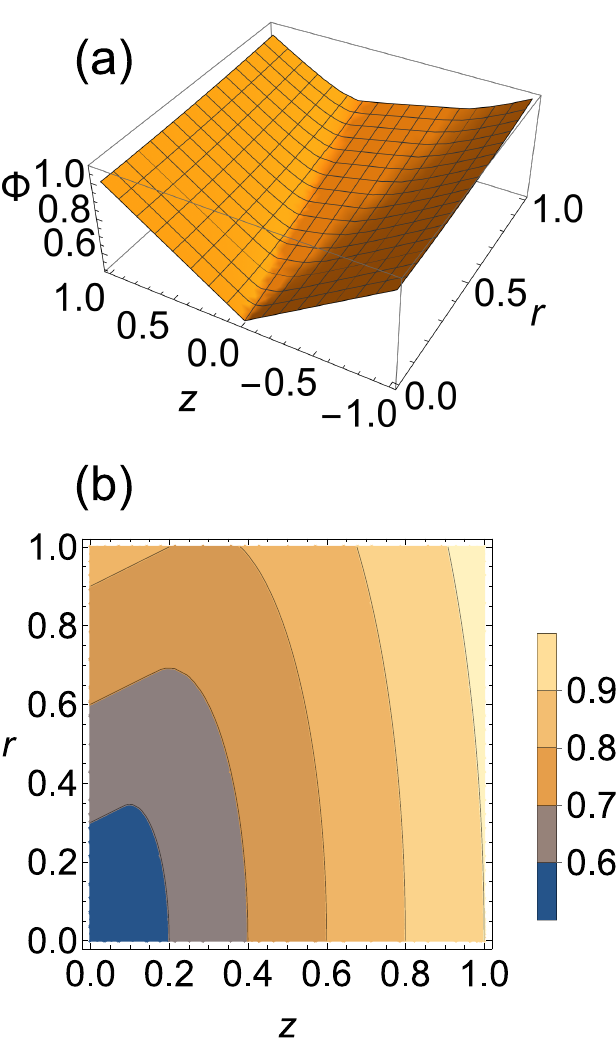}
\caption{\textbf{Conical structure time-covariance.} Shown here is $\Phi_\tau(\bm{x})$ for time-covariant transformations in a non-degenerate qubit from the state $\tau$ with $c_\tau = \frac{1}{3}$ and $p_\tau = \frac{1}{3}$, plotted for $0 \le r \le 1$, where $r \coloneqq \sqrt{x^2+y^2}$, and $-1 \le z \le 1$. The non-negative linear behaviour of $\Phi_\tau(\bm{x})$ in every direction out of the maximally mixed state (i.e. along any straight line out of $r=0,z=0$) is readily seen in (a) and generates the conical behaviour discussed in the text. Figure (b) gives a contour plot of $\Phi_\tau(\bm{x})$ as a function of the reference frame state Bloch vector.} 
\end{figure}

This analysis illustrates how the original complete set of entropic conditions has many redundancies. We explicitly see the conical behaviour as $(x,y,z) \rightarrow (0,0,0)$ in that
\begin{align}
	\Phi_\tau( \lambda(x,y,z)) = \lambda \left(\Phi_\tau(x,y,z) - \frac{1}{2}\right) + \frac{1}{2},
\end{align} 
where $\lambda$ is a positive scaling factor. Furthermore, because $U(1)$ is Abelian, we have that $[\U_t \otimes \id, \G] = 0$ for all $t$, so Lemma \ref{lemma:cov_isometry_invariance} implies, for any $r \ge 0$, that
\begin{align}
	\Phi_\tau(r\cos(t),r\sin(t),z) = \Phi_\tau(r,0, z),
	\label{eq:phi_qubit_u1_cylindrical_asymmetry}
\end{align}
and so $\Phi_\tau$ has cylindrical symmetry around the $z$-axis. 
More non-trivially, Lemma \ref{lemma:cov_isometry_invariance} may also be applied to ${\X \otimes \X}$, where $\X(\cdot) \coloneqq \sigma_x (\cdot) \sigma_x$, since $[\X \otimes \X, \G] = 0$. This means
\begin{align}
	\Phi_{\X(\tau)}(x,y,z) = \Phi_\tau(x,-y,-z)
	\label{eq:phi_qubit_u1_top_bottom_bloch_sphere}	
\end{align}
According to the parameterisation of $\tau$ we have chosen, $\X(\tau)$ means $c_\tau \rightarrow c^*_\tau$ and $p_\tau \rightarrow 1-p_\tau$. In this way, $\Phi_\tau$ for reference states in the bottom half of the Bloch sphere (i.e. $z \leq 0$) can be calculated from $\Phi_{\X(\tau)}$ for reference states in the top half ($z \ge 0$).

Given any $\rho$ and $\sigma$, we can look at the minimality condition at $(0,0,0)$ and obtain reference frame independent conditions that recover known results~\cite{matteo_kamil_bound} on necessary and sufficient conditions for time-covariant transitions in a non-degenerate qubit:
\begin{align}
	\frac{\abs{c_\rho}^2}{1-p_\rho} \ge \frac{\abs{c_\sigma}^2}{1-p_\sigma}& \text{ for } p_\sigma \ge p_\rho,
\end{align}
and
\begin{align}
	\frac{\abs{c_\rho}^2}{p_\rho} \ge \frac{\abs{c_\sigma}^2}{p_\sigma}& \text{ for } p_\sigma \le p_\rho.
\end{align}
Comparing with \eqref{eq:qubit_u1_phi}, this means when $p_\sigma \ge p_\rho$, checking $\Delta \Phi_\eta \ge 0$ for a \emph{single} reference state with co-ordinates in the range
\begin{equation}
0 \le \frac{\sqrt{x^2+y^2}}{2z} \le \min \left\{\frac{1-p_\rho}{\abs{c_\rho}}, \frac{1-p_\tau}{\abs{c_\tau}} \right\}
\end{equation}
  is sufficient to determine whether a covariant transition can occur. Similarly, when $p_\sigma \le p_\rho$, checking $\Delta \Phi_\eta$ for a \emph{single} reference state with co-ordinates in the range 
  \begin{equation}
  0 \ge \frac{\sqrt{x^2+y^2}}{2z} \ge \max\left\{-\frac{p_\rho}{\abs{c_\rho}}, -\frac{p_\sigma}{\abs{c_\sigma}} \right\}
  \end{equation}
is sufficient. 

\subsubsection{The case of $SU(2)$--covariant transformations.}
We now consider the case of $G=SU(2)$ on a qubit. In this case, $G$--covariant channels partially depolarise and may additionally invert the input state about the maximally mixed state \cite{Cirstoiu2020_Noether}. We will continue to write $\eta$ in its Bloch basis as in the $U(1)$ example, and will  parameterise $\tau$ as before. Using the simplifying abbreviations $\bm{\sigma} \coloneqq (\sigma_x,\sigma_y,\sigma_z)^T$ and $\bm{x} = (x,y,z)^T$, we have
\begin{align}
\eta(\bm{x}) &\coloneqq \frac{1}{2}(\id + \bm{x} \cdot \bm{\sigma}), \\
\rho(\bm{r}) &\coloneqq \frac{1}{2}(\id + \bm{r} \cdot \bm{\sigma}),
\end{align}
Given $\bm{x} = (x,y,z)$ we define $\overline{\bm{x}}= (x,-y,z)$. 

The form of $\Phi_\tau(\bm{x})$ is then given by (see Appendix \ref{appx:qubit_su2_phi_derivation}):
\begin{align}
\label{eq:qubit_phi_su2}
\Phi_\tau(\bm{x}) &= \begin{cases}
\frac{1}{2} (1 + \overline{\bm{x}} \cdot \bm{r})  & \text{if } \overline{\bm{x}} \cdot \bm{r} \ge 0,\\
\frac{1}{2}  (1- \frac{1}{3} \overline{\bm{x}}) \cdot \bm{r}  & \mbox{ otherwise}.
\end{cases}
\end{align}
We see that $\Phi_\tau(\bm{x})$ is piecewise linear with the plane ${\overline{\bm{x}} \cdot \bm{r} = 0}$ distinguishing the two regions.

Consider an input state $\rho$ and an output state $\sigma$ that are not maximally mixed. Letting $\bm{r}$ and $\bm{s}$ be the Bloch vectors of $\rho$ and $\sigma$ respectively, this means $\bm{r}, \bm{s} \neq \bm{0}$. Let us further restrict ourselves to the case where $\rho$ and $\sigma$ are not located along the same diameter in the Bloch sphere. This implies both $\bm{r} - \bm{s} \neq \bm{0}$ and that $\bm{r} - \bm{s}$ cannot be anti-parallel to $\bm{s}$. Therefore, it is always possible to find $\overline{\bm{x}'} \neq \bm{0}$ such that $\overline{\bm{x}'}$ lies strictly above both the plane $\bm{x} \cdot \bm{s} = 0$ and the plane $\bm{x} \cdot (\bm{r} - \bm{s}) = 0$. This means 
\begin{align}
	\overline{\bm{x}'} \cdot \bm{r} > \overline{\bm{x}'} \cdot \bm{s} > 0,
\end{align} 
so $\Delta \Phi(\overline{\bm{x}'})$ must be calculated from the top solution in \eqref{eq:qubit_phi_su2} as
\begin{align}
	\Delta \Phi(\overline{\bm{x}'}) = \frac{1}{2}\bm{x}' \cdot (\bm{r} - \bm{s})
\end{align}

Conversely, $-\overline{\bm{x}'}$ must lie strictly below both the plane $\bm{x} \cdot \bm{s} = 0$ and the plane $\bm{x} \cdot (\bm{r} - \bm{s}) = 0$, so $\Delta \Phi(-\overline{\bm{x}'})$ must be calculated from the bottom solution in \eqref{eq:qubit_phi_su2}, which means
\begin{align}
	\Delta \Phi(-\overline{\bm{x}'}) = \frac{1}{6} \bm{x}' \cdot (\bm{r} - \bm{s}).
\end{align} 

We must therefore conclude that if $\rho$ and $\sigma$ are neither maximally mixed nor located along the same diameter of the Bloch sphere, then there exists $\overline{\bm{x}'}$ such that:
\begin{align}
	\Delta \Phi(\bm{0}) = 0 &\neq \frac{2}{3} \bm{x}' \cdot (\bm{r} - \bm{s})\\ 
	&= \Delta \Phi(\overline{\bm{x}'}) + \Delta \Phi(-\overline{\bm{x}'})
\end{align} 
Since most choices of $\rho$ and $\sigma$ satisfy these requirements, we see that $\Delta \Phi(\bm{x})$ is almost never completely linear. In this example, $\Delta H_\eta$ also almost always has a conical singularity at $\eta = \frac{\id}{d}$.

\section{Robust symmetric transformations of general states with minimal depolarization}\label{depolarize}

\begin{figure}[b]
\includegraphics[width=9cm]{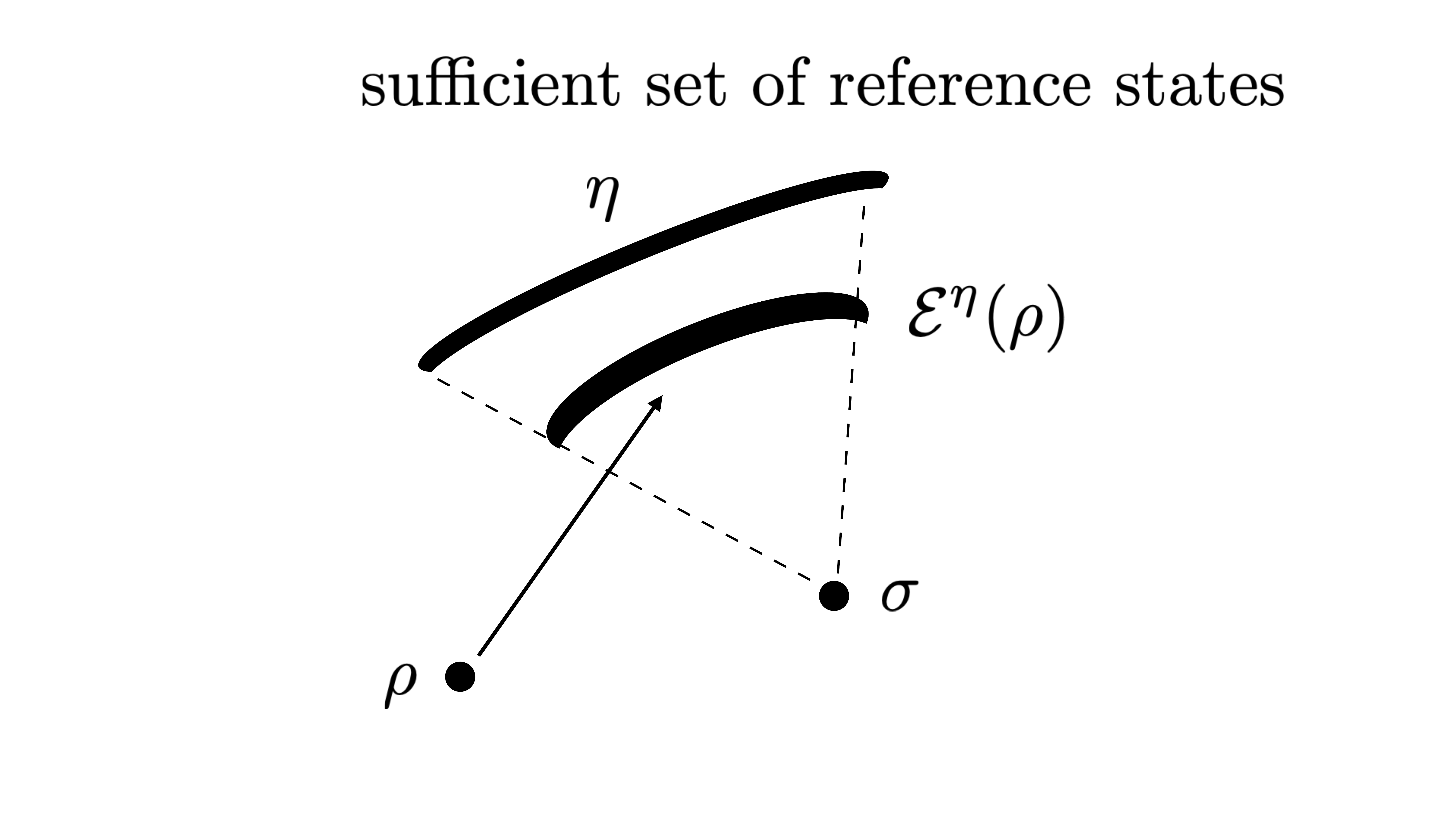}
\caption{\textbf{Constructing state interconversion conditions.} If, for a complete set of reference frame states $\{\eta\}$, we construct some family of covariant protocols that transform from a state $\rho$ to $\E^\eta(\rho)$ that has a higher overlap with $\eta$ than $\sigma$ has with $\eta$ then it is possible to transform from $\rho$ to $\sigma$ under a $G$--covariant channel.}
\end{figure}
In principle the condition given in Theorem \ref{thm:local_min_max_mixed_condition} gives a complete description of the asymmetry properties of quantum states. However, as the preceding examples have shown, standard tests for local minima are typically not applicable for the functional $\Delta H_\eta$, and thus computing this necessary and sufficient condition presents a technical challenge which we must leave for future study. Instead, we can adopt a more physical perspective on the problem and look for a complete set of conditions where we weaken the assumptions for the interconversion.
For example instead of $\rho \xrightarrow{G} \sigma$ we could ask the question:
\begin{center}
\textit{What is the minimal amount of depolarization noise we need to add to $\sigma$ so as to make it accessible from the initial state $\rho$ via a $G$--covariant channel? }
\end{center} 
Since the maximally mixed state is invariant for any group action this form applies to all symmetry groups $G$. It also incorporates robustness. Suppose, for example, that $\sigma$ was essentially identical to $\rho$ except it has a very small, but non-zero, $O(\varepsilon)$ mode that does not appear in $\rho$. The strict conditions would say that it is impossible to transform from $\rho$ to $\sigma$, yet it is clear that we only require $O(\varepsilon)$ amount of depolarising noise in order to make the transformation possible. Therefore the above question is more physically relevant than the simple `yes/no' question of exact interconversion.

The formulation of the problem therefore involves smoothing our output state with the maximally mixed state:
\begin{align}
\sigma \rightarrow \sigma_p:=(1-p) \sigma + p \frac{\id}{d}, 
\end{align}
where  $p$ is an error probability, and we wish to estimate how small $p$ can be so as to make $\rho \xrightarrow{G} \sigma_p$ possible via a covariant quantum channel. As we will see, this set of sufficient conditions has the benefit of being straightforward to compute.

We make use of two core ingredients for our results. First, note that any state $\rho$ can be decomposed into independent modes of asymmetry~\cite{marvian2014modes} labelled by $(\lambda,j)$:
\begin{align}
  \rho = \sum_{\lambda,j} \rho^\lambda_j: \, \rho^{\lambda}_j = \sum_\alpha \tr\left[X^{(\lambda,\alpha)\dagger}_j \rho\right] X^{(\lambda,\alpha)}_j,
\end{align}
where $\lambda$ labels an irreducible representation (irrep) of $G$,  $j$ labels the basis vector of the given irrep $\lambda$, $\alpha$ labels any multiplicity degrees of freedom, and the set $\left\{X^{(\lambda,\alpha)}_j\right\}$ form an orthonormal irreducible tensor operator (ITO) basis for $\B(\H)$ (see \appref{appx:background_details} for details). We denote the trivial irrep of the group by $\lambda=0$. It was shown Ref.~\cite{marvian2014modes} that every $G$--covariant operation $\E: \E(\rho) = \sigma$ acts independently on the different modes of the input state such that
\begin{align}
    \E(\rho^\lambda_j) = \sigma^\lambda_j,
\label{eq:modes_cov_map}
\end{align}
for any $(\lambda,j)$. In other words, a $G$--covariant quantum channel always maps any given mode of the input state to the very same mode of the output state, with no ``mixing'' between the different modes.

Secondly, we have the Sandwiched $\alpha$--R\'{e}nyi divergence $D_\alpha(\rho ||\sigma)$ for two states $\rho, \sigma$ of a quantum system, which is defined as~\cite{muller2013quantum,wilde2014strong}
\begin{equation}
    D_\alpha(\rho||\sigma):= \frac{1}{\alpha-1} \log \tr \left [\sigma^{\frac{1-\alpha}{2\alpha}}\rho \sigma^{\frac{1-\alpha}{2\alpha}} \right ]^\alpha,
\end{equation}
whenever the support of $\rho$ lies in the support of $\sigma$, and is infinite otherwise. Our results turn out to be most compactly expressed in terms of the following generalization of the $\alpha=2$ Sandwiched R\'{e}nyi divergence, which extends the domain of the first argument to all linear operators in the support of $\sigma$, and reproduces the standard definition when that first argument is Hermitian:
\begin{equation}
D_2(X||\sigma):= \log \tr \left( \left[\sigma^{-\frac{1}{4}}X \sigma^{-\frac{1}{4}} \right]^\dagger \left[\sigma^{-\frac{1}{4}}X \sigma^{-\frac{1}{4}} \right] \right).
\end{equation}

We now have the following theorem, which gives an estimate of the minimal amount of depolarization needed in order to make a transformation possible under covariant channels. 
\begin{restatable}{theorem}{SuffCond}
\label{thrm:suff_cond_depolarized}
Let $0  \le p \le 1$ be a probability. There exists a $G$--covariant channel transforming $\rho$ into $\sigma_p \coloneqq (1-p)\sigma + p \frac{\id}{d}$ if
\begin{align}
 D_2 ( \rho^\lambda_j || \G(\rho))   \ge  \log g^{\lambda}_j(\sigma) - \log n^{-1}\left (\lambda_{\rm min} + \frac{p}{d(1-p)} \right),
\end{align}\label{eqn:depolarization}
for all $\lambda \neq 0, j$, where $\lambda_{\mathrm{min}}$ is the smallest non-zero eigenvalue of $\G(\sigma)$, and $g^{\lambda}_j(\sigma) \coloneqq \sum_\alpha \abs{\tr[X^{(\lambda,\alpha)\dagger}_j \sigma] }$. The operators $X^{(\lambda,\alpha)}_j$ form an ITO basis for $\B(\H_B)$, where $\H_B$ is the Hilbert space of the output system, and $n$ is the sum of the dimensions of all distinct non-trivial irreps appearing in the representation of $G$ on $\B(\H_B)$.
\end{restatable}
A proof is given in \appref{appx:sufficient_conditions}, and exploits the SDP duality structure for covariant interconversion to determine an admissible range of values for $p$. This analysis is done by using a family of Pretty Good Measurement schemes~\cite{hausladen1994pretty} that attempt to generate as large a fidelity with the set of all reference frame states $\eta_R$ as possible. By modifying this general strategy, we anticipate that the results presented here can almost certainly be improved upon, and it would be of interest to study how well similar families perform relative to the exact SDP solution to the covariant interconversion problem. 
\begin{figure}
\centering
\includegraphics[scale=0.9]{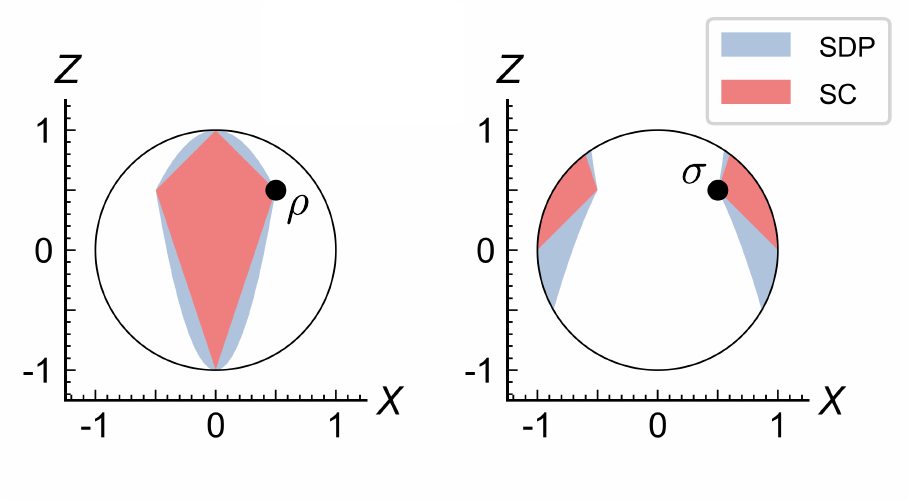}
\caption{\textbf{Depolarization conditions for a qubit system.} An exact treatment of how states transform under $G$--covariant channels requires non-trivial computations, however the closed, robust conditions given in Theorem~\ref{thrm:suff_cond_q} provide a simple means to estimate the interconversion structure. Here we demonstrate this for the case of $d=2$ and $G=U(1)$.  \textbf{Left:} The black dot shows an initial qubit state $\rho$ with Bloch vector $\bm{r}\coloneqq (\frac{1}{2},0,\frac{1}{2})$. The large blue shaded region (SDP) defines the full set of output qubit states that can be reached under covariant maps for the group $G=U(1)$ of time-translations generated by $H = \sigma_z$, computed via the exact interconversion conditions. The smaller pink shaded region (SC) overlapping this shows the region given by the conditions given in \thmref{thrm:suff_cond_q}. \textbf{Right:} Here the black dot now represents the \emph{output} qubit state $\sigma$, and the shaded regions correspond to the set of input states $\{\rho\}$ that can be transformed into $\sigma$ via a $G$--covariant channel for the full set of conditions (SDP) and the robust approximation conditions (SC). (Note that despite appearances the boundary of the SDP region is not linear.)}
\label{fig:qubit_sufficient_cond}
\end{figure}

When the input and output systems are the same, we can provide a strengthening of the above conditions to the following form:
\begin{restatable}{theorem}{SuffCondq}
\label{thrm:suff_cond_q}
Consider transformations from a quantum system $A$ to itself. Assume for simplicity that $\G(\sigma)$ is full-rank. There exists a $G$--covariant channel transforming $\rho$ into $\sigma_p \coloneqq (1-p)\sigma + p \frac{\id}{d}$ if $\rho = \sigma_p$ or if for any $q \in (q^*,1]$ we have
\begin{align}
 D_2 ( \rho^\lambda_j || \G(\rho))   \ge  \log g^{\lambda}_j(\sigma_p(q)) - \log n^{-1} \lambda_{\mathrm{ min}}(\G[\sigma_p(q)]),
\end{align}
for all $\lambda\ne 0 , j$, where we have $\sigma_p(q) \coloneqq \sigma_p -(1-q)\rho$, $q^* \coloneqq \min \{ q \in \mathbb{R}^+ : \G(\sigma_p(q))\ge 0\}$, $\lambda_{\mathrm{ min}}(\G[\sigma_p(q)])$ is the smallest eigenvalue of $\G[\sigma_p(q)]$, and all other terms are as in Theorem~\ref{thrm:suff_cond_depolarized}. 
\end{restatable}

Given that the analysis is built on Pretty Good Measurement schemes for resolving group elements, it is expected that a measure-and-prepare strategy (such as above, or a slightly modified version) will behave well when $\rho$ has many large modes of asymmetry for the group. While this can be achieved for systems with a large dimension, we find that even for low dimensional systems the conditions perform well. For example, in \figref{fig:qubit_sufficient_cond}, we plot the performance of the sufficient condition for the group $G=U(1)$ of time translations generated by the Hamiltonian $H = \sigma_z$ for a qubit system in initial state $\rho = \frac{1}{2}(\id + \frac{1}{2} \sigma_x + \frac{1}{2} \sigma_z) $. On the left, we plot the set of output states $\{ \sigma\}$ for which our sufficient condition tells us are accessible from $\rho$ (pink shaded region) relative to the full set of accessible output states granted by the complete set of conditions stated in \thmref{thrm:gour} (blue shaded region). 

We also note that we can recast our sufficient conditions in terms of familiar norms. We first note that we always have $g^\lambda_j(\sigma) \le \norm{\sigma_j^\lambda}_1$, where $\norm{X}_1 \coloneqq \tr\left[\sqrt{X^\dagger X}\right]$ is the trace norm. Therefore it follows from \thmref{thrm:suff_cond_depolarized} that exists a $G$--covariant operation transforming $\rho$ into $\sigma$ if
\begin{align}
n^{-1} \norm{\overline{\rho}^\lambda_j}_2^2     \ge  \norm{\tilde{\sigma}_j^\lambda}_1   , \quad \forall \lambda \neq 0, j,
\end{align}
where the notation $\overline{\rho}^\lambda_j \coloneqq {\G(\rho)}^{-\frac{1}{4}} \rho^\lambda_j {\G(\rho)}^{-\frac{1}{4}}$ and $\tilde{\sigma}= \sigma \lambda_{\mathrm{ min}}^{-1}$, and $\norm{X}_2 \coloneqq \sqrt{\tr[X^\dagger X]}$ is the Frobenius norm. We note that $\norm{\sigma_j^\lambda}_1$ for $\lambda \neq 0$ is a known asymmetry monotone that measures the the degree of asymmetry in the $(\lambda,j)$--mode of $\sigma$ \cite{marvian2014modes}.

\section{Outlook}
In this work we have shown that the recent complete set of entropic conditions for asymmetry can be greatly simplified, and more importantly, can be converted into useful forms. The fact that the reference frames that are needed to describe asymmetry can be taken to have arbitrarily small modes of asymmetry suggest that a deeper analysis should be possible in terms of differential geometry, as opposed to quantities such as the degree to which a quantum state encodes group data. We expect that this should take the form of a Fisher-like information~\cite{marvian2012symmetry,Marvian2014Extending,takagi2019skew}, and in particular it is of interest to see if it is possible to replace the $H_{\rm min}(R|A)$ entropy with the conditional von-Neumann entropy $H(R|A)$, which would allow explicit analytic computations. 

Beyond this, a range of other interesting questions exist. For example we have not exploited the duality relations~\cite{konig2009operational, tomamichel2010duality} between $H_{\rm min}(R|A)$ and $H_{\rm max}(R|S)$ where $S$ is a purifying system for the state $\Omega_{RA}$. For example, for the case of time-translation symmetry the joint purified state admits two notable forms. The first is an energetic form:
\begin{equation}
\Omega_{RAS} = \sum_E \sqrt{p(E)} |\varphi(E)\>_{RA} \otimes | E\>_S
\end{equation}
obtained from considering $\Omega_{RA} = \G(\eta_R \otimes \rho_A) = \sum_{E} \Pi(E)( \eta_R \otimes \rho_A) \Pi(E)$ as an ensemble of states over energy sectors, $\Pi(E)$ being the projector onto the energy $E$ subspace of $RA$. While the second is a temporal form, given by
\begin{equation}
\Omega_{RAS} =\int \!\!dt \,\, |\phi(t)\>_{RAS},
\end{equation}
with $|\phi(t)\>_{RAS}$ being a purification of $\U_t(\eta_R) \otimes \U_t(\rho_A)$. It would be of interest to explore these two forms and also their connection to entropic uncertainty relations.

Finally, it would also be valuable to see how the explicit conditions given by Theorem~\ref{thrm:suff_cond_depolarized} and Theorem~\ref{thrm:suff_cond_q} could be used in concrete settings, such as for covariant quantum error-correcting codes~\cite{Faist2020_Approx_QEC,Woods2020continuousgroupsof,yang2020covariant}, thermodynamics~\cite{gour2018quantum} or metrology~\cite{Hall2012Nonlinear}. Moreover, the method of constructing these conditions can certainly be improved upon by using more detailed covariant protocols. 

\section{Acknowledgements}

We would like to thank Iman Marvian for helpful and insightful discussions, and in particular for pointing out that our depolarization result is compactly expressed in terms of a Sandwiched R\'{e}nyi divergence. RA is supported by the EPSRC Centre for Doctoral Training in Controlled Quantum Dynamics. SGC is supported by the Bell Burnell Graduate Scholarship Fund and the University of Leeds. DJ is supported by the Royal Society and also a University Academic Fellowship.

\bibliographystyle{apsrev4-2}
\bibliography{paper_referencesU}

\begin{thebibliography}{81}%
\makeatletter
\providecommand \@ifxundefined [1]{%
 \@ifx{#1\undefined}
}%
\providecommand \@ifnum [1]{%
 \ifnum #1\expandafter \@firstoftwo
 \else \expandafter \@secondoftwo
 \fi
}%
\providecommand \@ifx [1]{%
 \ifx #1\expandafter \@firstoftwo
 \else \expandafter \@secondoftwo
 \fi
}%
\providecommand \natexlab [1]{#1}%
\providecommand \enquote  [1]{``#1''}%
\providecommand \bibnamefont  [1]{#1}%
\providecommand \bibfnamefont [1]{#1}%
\providecommand \citenamefont [1]{#1}%
\providecommand \href@noop [0]{\@secondoftwo}%
\providecommand \href [0]{\begingroup \@sanitize@url \@href}%
\providecommand \@href[1]{\@@startlink{#1}\@@href}%
\providecommand \@@href[1]{\endgroup#1\@@endlink}%
\providecommand \@sanitize@url [0]{\catcode `\\12\catcode `\$12\catcode
  `\&12\catcode `\#12\catcode `\^12\catcode `\_12\catcode `\%12\relax}%
\providecommand \@@startlink[1]{}%
\providecommand \@@endlink[0]{}%
\providecommand \url  [0]{\begingroup\@sanitize@url \@url }%
\providecommand \@url [1]{\endgroup\@href {#1}{\urlprefix }}%
\providecommand \urlprefix  [0]{URL }%
\providecommand \Eprint [0]{\href }%
\providecommand \doibase [0]{https://doi.org/}%
\providecommand \selectlanguage [0]{\@gobble}%
\providecommand \bibinfo  [0]{\@secondoftwo}%
\providecommand \bibfield  [0]{\@secondoftwo}%
\providecommand \translation [1]{[#1]}%
\providecommand \BibitemOpen [0]{}%
\providecommand \bibitemStop [0]{}%
\providecommand \bibitemNoStop [0]{.\EOS\space}%
\providecommand \EOS [0]{\spacefactor3000\relax}%
\providecommand \BibitemShut  [1]{\csname bibitem#1\endcsname}%
\let\auto@bib@innerbib\@empty
\bibitem [{\citenamefont {Noether}(1918)}]{noether1918invarianten}%
  \BibitemOpen
  \bibfield  {author} {\bibinfo {author} {\bibfnamefont {E.}~\bibnamefont
  {Noether}},\ }\bibfield  {title} {\emph {\bibinfo {title} {Invarianten
  beliebiger {D}ifferentialausdr{\"u}cke}},\ }\href@noop {} {\bibfield
  {journal} {\bibinfo  {journal} {Nachrichten von der Gesellschaft der
  Wissenschaften zu G{\"o}ttingen, mathematisch-physikalische Klasse}\ }\textbf
  {\bibinfo {volume} {1918}},\ \bibinfo {pages} {37} (\bibinfo {year}
  {1918})}\BibitemShut {NoStop}%
\bibitem [{\citenamefont {Watrous}(2018)}]{watrous2018theory}%
  \BibitemOpen
  \bibfield  {author} {\bibinfo {author} {\bibfnamefont {J.}~\bibnamefont
  {Watrous}},\ }\href@noop {} {\emph {\bibinfo {title} {The Theory of Quantum
  Information}}}\ (\bibinfo  {publisher} {Cambridge University Press},\
  \bibinfo {year} {2018})\BibitemShut {NoStop}%
\bibitem [{\citenamefont {Plenio}\ and\ \citenamefont
  {Virmani}(2007)}]{plenio2007entanglement}%
  \BibitemOpen
  \bibfield  {author} {\bibinfo {author} {\bibfnamefont {M.}~\bibnamefont
  {Plenio}}\ and\ \bibinfo {author} {\bibfnamefont {S.}~\bibnamefont
  {Virmani}},\ }\bibfield  {title} {\emph {{\selectlanguage {english}\bibinfo
  {title} {An Introduction to Entanglement Measures}}},\ }\href@noop {}
  {\bibfield  {journal} {\bibinfo  {journal} {Quantum Information \&
  Computation}\ }\textbf {\bibinfo {volume} {7}},\ \bibinfo {pages} {1}
  (\bibinfo {year} {2007})}\BibitemShut {NoStop}%
\bibitem [{\citenamefont {Horodecki}\ \emph {et~al.}(2009)\citenamefont
  {Horodecki}, \citenamefont {Horodecki}, \citenamefont {Horodecki},\ and\
  \citenamefont {Horodecki}}]{horodecki2009entanglement}%
  \BibitemOpen
  \bibfield  {author} {\bibinfo {author} {\bibfnamefont {R.}~\bibnamefont
  {Horodecki}}, \bibinfo {author} {\bibfnamefont {P.}~\bibnamefont
  {Horodecki}}, \bibinfo {author} {\bibfnamefont {M.}~\bibnamefont
  {Horodecki}},\ and\ \bibinfo {author} {\bibfnamefont {K.}~\bibnamefont
  {Horodecki}},\ }\bibfield  {title} {\emph {\bibinfo {title} {Quantum
  Entanglement}},\ }\href {https://doi.org/10.1103/RevModPhys.81.865}
  {\bibfield  {journal} {\bibinfo  {journal} {Rev. Mod. Phys.}\ }\textbf
  {\bibinfo {volume} {81}},\ \bibinfo {pages} {865} (\bibinfo {year}
  {2009})}\BibitemShut {NoStop}%
\bibitem [{\citenamefont {\AA{}berg}(2006)}]{aberg2006quantifying}%
  \BibitemOpen
  \bibfield  {author} {\bibinfo {author} {\bibfnamefont {J.}~\bibnamefont
  {\AA{}berg}},\ }\bibfield  {title} {\emph {\bibinfo {title} {Quantifying
  Superposition}},\ }\href@noop {} {\  (\bibinfo {year} {2006})},\ \Eprint
  {https://arxiv.org/abs/quant-ph/0612146} {arXiv:quant-ph/0612146}
  \BibitemShut {NoStop}%
\bibitem [{\citenamefont {Marvian}\ and\ \citenamefont
  {Spekkens}(2013)}]{marvian2013theory}%
  \BibitemOpen
  \bibfield  {author} {\bibinfo {author} {\bibfnamefont {I.}~\bibnamefont
  {Marvian}}\ and\ \bibinfo {author} {\bibfnamefont {R.~W.}\ \bibnamefont
  {Spekkens}},\ }\bibfield  {title} {\emph {\bibinfo {title} {The Theory of
  Manipulations of Pure State Asymmetry: I. {B}asic Tools, Equivalence Classes
  and Single Copy Transformations}},\ }\href
  {https://doi.org/10.1088/1367-2630/15/3/033001} {\bibfield  {journal}
  {\bibinfo  {journal} {New J. Phys.}\ }\textbf {\bibinfo {volume} {15}},\
  \bibinfo {pages} {033001} (\bibinfo {year} {2013})}\BibitemShut {NoStop}%
\bibitem [{\citenamefont {Baumgratz}\ \emph {et~al.}(2014)\citenamefont
  {Baumgratz}, \citenamefont {Cramer},\ and\ \citenamefont
  {Plenio}}]{baumgratz2014coherence}%
  \BibitemOpen
  \bibfield  {author} {\bibinfo {author} {\bibfnamefont {T.}~\bibnamefont
  {Baumgratz}}, \bibinfo {author} {\bibfnamefont {M.}~\bibnamefont {Cramer}},\
  and\ \bibinfo {author} {\bibfnamefont {M.~B.}\ \bibnamefont {Plenio}},\
  }\bibfield  {title} {\emph {\bibinfo {title} {Quantifying Coherence}},\
  }\href {https://doi.org/10.1103/PhysRevLett.113.140401} {\bibfield  {journal}
  {\bibinfo  {journal} {Phys. Rev. Lett.}\ }\textbf {\bibinfo {volume} {113}},\
  \bibinfo {pages} {140401} (\bibinfo {year} {2014})}\BibitemShut {NoStop}%
\bibitem [{\citenamefont {Streltsov}\ \emph {et~al.}(2017)\citenamefont
  {Streltsov}, \citenamefont {Adesso},\ and\ \citenamefont
  {Plenio}}]{streltsov2017coherence}%
  \BibitemOpen
  \bibfield  {author} {\bibinfo {author} {\bibfnamefont {A.}~\bibnamefont
  {Streltsov}}, \bibinfo {author} {\bibfnamefont {G.}~\bibnamefont {Adesso}},\
  and\ \bibinfo {author} {\bibfnamefont {M.~B.}\ \bibnamefont {Plenio}},\
  }\bibfield  {title} {\emph {\bibinfo {title} {Colloquium: Quantum Coherence
  As a Resource}},\ }\href {https://doi.org/10.1103/RevModPhys.89.041003}
  {\bibfield  {journal} {\bibinfo  {journal} {Rev. Mod. Phys.}\ }\textbf
  {\bibinfo {volume} {89}},\ \bibinfo {pages} {041003} (\bibinfo {year}
  {2017})}\BibitemShut {NoStop}%
\bibitem [{\citenamefont {Brand\~ao}\ \emph {et~al.}(2013)\citenamefont
  {Brand\~ao}, \citenamefont {Horodecki}, \citenamefont {Oppenheim},
  \citenamefont {Renes},\ and\ \citenamefont
  {Spekkens}}]{brandao2013athermality}%
  \BibitemOpen
  \bibfield  {author} {\bibinfo {author} {\bibfnamefont {F.~G. S.~L.}\
  \bibnamefont {Brand\~ao}}, \bibinfo {author} {\bibfnamefont {M.}~\bibnamefont
  {Horodecki}}, \bibinfo {author} {\bibfnamefont {J.}~\bibnamefont
  {Oppenheim}}, \bibinfo {author} {\bibfnamefont {J.~M.}\ \bibnamefont
  {Renes}},\ and\ \bibinfo {author} {\bibfnamefont {R.~W.}\ \bibnamefont
  {Spekkens}},\ }\bibfield  {title} {\emph {\bibinfo {title} {Resource Theory
  of Quantum States out of Thermal Equilibrium}},\ }\href
  {https://doi.org/10.1103/PhysRevLett.111.250404} {\bibfield  {journal}
  {\bibinfo  {journal} {Phys. Rev. Lett.}\ }\textbf {\bibinfo {volume} {111}},\
  \bibinfo {pages} {250404} (\bibinfo {year} {2013})}\BibitemShut {NoStop}%
\bibitem [{\citenamefont {Horodecki}\ and\ \citenamefont
  {Oppenheim}(2013)}]{horodecki2013fundamental}%
  \BibitemOpen
  \bibfield  {author} {\bibinfo {author} {\bibfnamefont {M.}~\bibnamefont
  {Horodecki}}\ and\ \bibinfo {author} {\bibfnamefont {J.}~\bibnamefont
  {Oppenheim}},\ }\bibfield  {title} {\emph {\bibinfo {title} {Fundamental
  Limitations for Quantum and Nanoscale Thermodynamics}},\ }\href
  {https://doi.org/10.1038/ncomms3059} {\bibfield  {journal} {\bibinfo
  {journal} {Nat. Commun.}\ }\textbf {\bibinfo {volume} {4}},\ \bibinfo {pages}
  {1} (\bibinfo {year} {2013})}\BibitemShut {NoStop}%
\bibitem [{\citenamefont {Gour}\ \emph {et~al.}(2015)\citenamefont {Gour},
  \citenamefont {M{\"u}ller}, \citenamefont {Narasimhachar}, \citenamefont
  {Spekkens},\ and\ \citenamefont {Halpern}}]{gour2015resource}%
  \BibitemOpen
  \bibfield  {author} {\bibinfo {author} {\bibfnamefont {G.}~\bibnamefont
  {Gour}}, \bibinfo {author} {\bibfnamefont {M.~P.}\ \bibnamefont
  {M{\"u}ller}}, \bibinfo {author} {\bibfnamefont {V.}~\bibnamefont
  {Narasimhachar}}, \bibinfo {author} {\bibfnamefont {R.~W.}\ \bibnamefont
  {Spekkens}},\ and\ \bibinfo {author} {\bibfnamefont {N.~Y.}\ \bibnamefont
  {Halpern}},\ }\bibfield  {title} {\emph {\bibinfo {title} {The Resource
  Theory of Informational Nonequilibrium in Thermodynamics}},\ }\href
  {https://doi.org/https://doi.org/10.1016/j.physrep.2015.04.003} {\bibfield
  {journal} {\bibinfo  {journal} {Phys. Rep.}\ }\textbf {\bibinfo {volume}
  {583}},\ \bibinfo {pages} {1} (\bibinfo {year} {2015})}\BibitemShut {NoStop}%
\bibitem [{\citenamefont {Albarelli}\ \emph {et~al.}(2018)\citenamefont
  {Albarelli}, \citenamefont {Genoni}, \citenamefont {Paris},\ and\
  \citenamefont {Ferraro}}]{albarelli2018resource}%
  \BibitemOpen
  \bibfield  {author} {\bibinfo {author} {\bibfnamefont {F.}~\bibnamefont
  {Albarelli}}, \bibinfo {author} {\bibfnamefont {M.~G.}\ \bibnamefont
  {Genoni}}, \bibinfo {author} {\bibfnamefont {M.~G.~A.}\ \bibnamefont
  {Paris}},\ and\ \bibinfo {author} {\bibfnamefont {A.}~\bibnamefont
  {Ferraro}},\ }\bibfield  {title} {\emph {\bibinfo {title} {Resource Theory of
  Quantum Non-{G}aussianity and {W}igner Negativity}},\ }\href
  {https://doi.org/10.1103/PhysRevA.98.052350} {\bibfield  {journal} {\bibinfo
  {journal} {Phys. Rev. A}\ }\textbf {\bibinfo {volume} {98}},\ \bibinfo
  {pages} {052350} (\bibinfo {year} {2018})}\BibitemShut {NoStop}%
\bibitem [{\citenamefont {Takagi}\ and\ \citenamefont
  {Zhuang}(2018)}]{takagi2018non-gaussianity}%
  \BibitemOpen
  \bibfield  {author} {\bibinfo {author} {\bibfnamefont {R.}~\bibnamefont
  {Takagi}}\ and\ \bibinfo {author} {\bibfnamefont {Q.}~\bibnamefont
  {Zhuang}},\ }\bibfield  {title} {\emph {\bibinfo {title} {Convex Resource
  Theory of Non-{G}aussianity}},\ }\href
  {https://doi.org/10.1103/PhysRevA.97.062337} {\bibfield  {journal} {\bibinfo
  {journal} {Phys. Rev. A}\ }\textbf {\bibinfo {volume} {97}},\ \bibinfo
  {pages} {062337} (\bibinfo {year} {2018})}\BibitemShut {NoStop}%
\bibitem [{\citenamefont {Veitch}\ \emph {et~al.}(2014)\citenamefont {Veitch},
  \citenamefont {Mousavian}, \citenamefont {Gottesman},\ and\ \citenamefont
  {Emerson}}]{veitch2014resource}%
  \BibitemOpen
  \bibfield  {author} {\bibinfo {author} {\bibfnamefont {V.}~\bibnamefont
  {Veitch}}, \bibinfo {author} {\bibfnamefont {S.~A.~H.}\ \bibnamefont
  {Mousavian}}, \bibinfo {author} {\bibfnamefont {D.}~\bibnamefont
  {Gottesman}},\ and\ \bibinfo {author} {\bibfnamefont {J.}~\bibnamefont
  {Emerson}},\ }\bibfield  {title} {\emph {\bibinfo {title} {The Resource
  Theory of Stabilizer Quantum Computation}},\ }\href
  {https://doi.org/10.1088/1367-2630/16/1/013009} {\bibfield  {journal}
  {\bibinfo  {journal} {New J. Phys.}\ }\textbf {\bibinfo {volume} {16}},\
  \bibinfo {pages} {013009} (\bibinfo {year} {2014})}\BibitemShut {NoStop}%
\bibitem [{\citenamefont {Howard}\ and\ \citenamefont
  {Campbell}(2017)}]{howard2017application}%
  \BibitemOpen
  \bibfield  {author} {\bibinfo {author} {\bibfnamefont {M.}~\bibnamefont
  {Howard}}\ and\ \bibinfo {author} {\bibfnamefont {E.}~\bibnamefont
  {Campbell}},\ }\bibfield  {title} {\emph {\bibinfo {title} {Application of a
  Resource Theory for Magic States to Fault-Tolerant Quantum Computing}},\
  }\href {https://doi.org/10.1103/PhysRevLett.118.090501} {\bibfield  {journal}
  {\bibinfo  {journal} {Phys. Rev. Lett.}\ }\textbf {\bibinfo {volume} {118}},\
  \bibinfo {pages} {090501} (\bibinfo {year} {2017})}\BibitemShut {NoStop}%
\bibitem [{\citenamefont {Chitambar}\ and\ \citenamefont
  {Gour}(2019)}]{Chitambar2019}%
  \BibitemOpen
  \bibfield  {author} {\bibinfo {author} {\bibfnamefont {E.}~\bibnamefont
  {Chitambar}}\ and\ \bibinfo {author} {\bibfnamefont {G.}~\bibnamefont
  {Gour}},\ }\bibfield  {title} {\emph {\bibinfo {title} {Quantum Resource
  Theories}},\ }\href {https://doi.org/10.1103/RevModPhys.91.025001} {\bibfield
   {journal} {\bibinfo  {journal} {Rev. Mod. Phys.}\ }\textbf {\bibinfo
  {volume} {91}},\ \bibinfo {pages} {025001} (\bibinfo {year}
  {2019})}\BibitemShut {NoStop}%
\bibitem [{\citenamefont {Marvian}\ and\ \citenamefont
  {Spekkens}(2014{\natexlab{a}})}]{Marvian2014Extending}%
  \BibitemOpen
  \bibfield  {author} {\bibinfo {author} {\bibfnamefont {I.}~\bibnamefont
  {Marvian}}\ and\ \bibinfo {author} {\bibfnamefont {R.~W.}\ \bibnamefont
  {Spekkens}},\ }\bibfield  {title} {\emph {\bibinfo {title} {Extending
  {N}oether’s Theorem by Quantifying the Asymmetry of Quantum States}},\
  }\href {https://doi.org/10.1038/ncomms4821} {\bibfield  {journal} {\bibinfo
  {journal} {Nat. Commun.}\ }\textbf {\bibinfo {volume} {5}},\ \bibinfo {pages}
  {1} (\bibinfo {year} {2014}{\natexlab{a}})}\BibitemShut {NoStop}%
\bibitem [{\citenamefont {Takagi}(2019)}]{takagi2019skew}%
  \BibitemOpen
  \bibfield  {author} {\bibinfo {author} {\bibfnamefont {R.}~\bibnamefont
  {Takagi}},\ }\bibfield  {title} {\emph {\bibinfo {title} {Skew Informations
  from an Operational View Via Resource Theory of Asymmetry}},\ }\href
  {https://doi.org/10.1038/s41598-019-50279-w} {\bibfield  {journal} {\bibinfo
  {journal} {Sci. Rep.}\ }\textbf {\bibinfo {volume} {9}},\ \bibinfo {pages}
  {1} (\bibinfo {year} {2019})}\BibitemShut {NoStop}%
\bibitem [{\citenamefont {Bartlett}\ \emph {et~al.}(2007)\citenamefont
  {Bartlett}, \citenamefont {Rudolph},\ and\ \citenamefont
  {Spekkens}}]{bartlett2007reference}%
  \BibitemOpen
  \bibfield  {author} {\bibinfo {author} {\bibfnamefont {S.~D.}\ \bibnamefont
  {Bartlett}}, \bibinfo {author} {\bibfnamefont {T.}~\bibnamefont {Rudolph}},\
  and\ \bibinfo {author} {\bibfnamefont {R.~W.}\ \bibnamefont {Spekkens}},\
  }\bibfield  {title} {\emph {\bibinfo {title} {Reference Frames,
  Superselection Rules, and Quantum Information}},\ }\href
  {https://doi.org/10.1103/RevModPhys.79.555} {\bibfield  {journal} {\bibinfo
  {journal} {Rev. Mod. Phys.}\ }\textbf {\bibinfo {volume} {79}},\ \bibinfo
  {pages} {555} (\bibinfo {year} {2007})}\BibitemShut {NoStop}%
\bibitem [{\citenamefont {Marvian}\ and\ \citenamefont
  {Spekkens}(2016)}]{Marvian2016Coherence}%
  \BibitemOpen
  \bibfield  {author} {\bibinfo {author} {\bibfnamefont {I.}~\bibnamefont
  {Marvian}}\ and\ \bibinfo {author} {\bibfnamefont {R.~W.}\ \bibnamefont
  {Spekkens}},\ }\bibfield  {title} {\emph {\bibinfo {title} {How to Quantify
  Coherence: Distinguishing Speakable and Unspeakable Notions}},\ }\href
  {https://doi.org/10.1103/PhysRevA.94.052324} {\bibfield  {journal} {\bibinfo
  {journal} {Phys. Rev. A}\ }\textbf {\bibinfo {volume} {94}},\ \bibinfo
  {pages} {052324} (\bibinfo {year} {2016})}\BibitemShut {NoStop}%
\bibitem [{\citenamefont {Hall}\ and\ \citenamefont
  {Wiseman}(2012)}]{Hall2012Nonlinear}%
  \BibitemOpen
  \bibfield  {author} {\bibinfo {author} {\bibfnamefont {M.~J.~W.}\
  \bibnamefont {Hall}}\ and\ \bibinfo {author} {\bibfnamefont {H.~M.}\
  \bibnamefont {Wiseman}},\ }\bibfield  {title} {\emph {\bibinfo {title} {Does
  Nonlinear Metrology Offer Improved Resolution? {A}nswers from Quantum
  Information Theory}},\ }\href {https://doi.org/10.1103/PhysRevX.2.041006}
  {\bibfield  {journal} {\bibinfo  {journal} {Phys. Rev. X}\ }\textbf {\bibinfo
  {volume} {2}},\ \bibinfo {pages} {041006} (\bibinfo {year}
  {2012})}\BibitemShut {NoStop}%
\bibitem [{\citenamefont {C\^{\i}rstoiu}\ \emph {et~al.}(2020)\citenamefont
  {C\^{\i}rstoiu}, \citenamefont {Korzekwa},\ and\ \citenamefont
  {Jennings}}]{Cirstoiu2020_Noether}%
  \BibitemOpen
  \bibfield  {author} {\bibinfo {author} {\bibfnamefont {C.}~\bibnamefont
  {C\^{\i}rstoiu}}, \bibinfo {author} {\bibfnamefont {K.}~\bibnamefont
  {Korzekwa}},\ and\ \bibinfo {author} {\bibfnamefont {D.}~\bibnamefont
  {Jennings}},\ }\bibfield  {title} {\emph {\bibinfo {title} {Robustness of
  {N}oether's Principle: Maximal Disconnects Between Conservation Laws and
  Symmetries in Quantum Theory}},\ }\href
  {https://doi.org/10.1103/PhysRevX.10.041035} {\bibfield  {journal} {\bibinfo
  {journal} {Phys. Rev. X}\ }\textbf {\bibinfo {volume} {10}},\ \bibinfo
  {pages} {041035} (\bibinfo {year} {2020})}\BibitemShut {NoStop}%
\bibitem [{\citenamefont {Chiribella}\ \emph {et~al.}(2021)\citenamefont
  {Chiribella}, \citenamefont {Aurell},\ and\ \citenamefont
  {\ifmmode~\dot{Z}\else \.{Z}\fi{}yczkowski}}]{Chiribella2021Symmetries}%
  \BibitemOpen
  \bibfield  {author} {\bibinfo {author} {\bibfnamefont {G.}~\bibnamefont
  {Chiribella}}, \bibinfo {author} {\bibfnamefont {E.}~\bibnamefont {Aurell}},\
  and\ \bibinfo {author} {\bibfnamefont {K.}~\bibnamefont
  {\ifmmode~\dot{Z}\else \.{Z}\fi{}yczkowski}},\ }\bibfield  {title} {\emph
  {\bibinfo {title} {Symmetries of Quantum Evolutions}},\ }\href
  {https://doi.org/10.1103/PhysRevResearch.3.033028} {\bibfield  {journal}
  {\bibinfo  {journal} {Phys. Rev. Research}\ }\textbf {\bibinfo {volume}
  {3}},\ \bibinfo {pages} {033028} (\bibinfo {year} {2021})}\BibitemShut
  {NoStop}%
\bibitem [{\citenamefont {Aharonov}\ and\ \citenamefont
  {Susskind}(1967)}]{AharonovSusskind1967}%
  \BibitemOpen
  \bibfield  {author} {\bibinfo {author} {\bibfnamefont {Y.}~\bibnamefont
  {Aharonov}}\ and\ \bibinfo {author} {\bibfnamefont {L.}~\bibnamefont
  {Susskind}},\ }\bibfield  {title} {\emph {\bibinfo {title} {Observability of
  the Sign Change of Spinors Under $2\ensuremath{\pi}$ Rotations}},\ }\href
  {https://doi.org/10.1103/PhysRev.158.1237} {\bibfield  {journal} {\bibinfo
  {journal} {Phys. Rev.}\ }\textbf {\bibinfo {volume} {158}},\ \bibinfo {pages}
  {1237} (\bibinfo {year} {1967})}\BibitemShut {NoStop}%
\bibitem [{\citenamefont {Chiribella}\ \emph {et~al.}(2004)\citenamefont
  {Chiribella}, \citenamefont {D'Ariano}, \citenamefont {Perinotti},\ and\
  \citenamefont {Sacchi}}]{Chiribella2004Efficient}%
  \BibitemOpen
  \bibfield  {author} {\bibinfo {author} {\bibfnamefont {G.}~\bibnamefont
  {Chiribella}}, \bibinfo {author} {\bibfnamefont {G.~M.}\ \bibnamefont
  {D'Ariano}}, \bibinfo {author} {\bibfnamefont {P.}~\bibnamefont
  {Perinotti}},\ and\ \bibinfo {author} {\bibfnamefont {M.~F.}\ \bibnamefont
  {Sacchi}},\ }\bibfield  {title} {\emph {\bibinfo {title} {Efficient Use of
  Quantum Resources for the Transmission of a Reference Frame}},\ }\href
  {https://doi.org/10.1103/PhysRevLett.93.180503} {\bibfield  {journal}
  {\bibinfo  {journal} {Phys. Rev. Lett.}\ }\textbf {\bibinfo {volume} {93}},\
  \bibinfo {pages} {180503} (\bibinfo {year} {2004})}\BibitemShut {NoStop}%
\bibitem [{\citenamefont {Jones}\ \emph {et~al.}(2006)\citenamefont {Jones},
  \citenamefont {Wiseman}, \citenamefont {Bartlett}, \citenamefont {Vaccaro},\
  and\ \citenamefont {Pope}}]{Jones2006Entanglement}%
  \BibitemOpen
  \bibfield  {author} {\bibinfo {author} {\bibfnamefont {S.~J.}\ \bibnamefont
  {Jones}}, \bibinfo {author} {\bibfnamefont {H.~M.}\ \bibnamefont {Wiseman}},
  \bibinfo {author} {\bibfnamefont {S.~D.}\ \bibnamefont {Bartlett}}, \bibinfo
  {author} {\bibfnamefont {J.~A.}\ \bibnamefont {Vaccaro}},\ and\ \bibinfo
  {author} {\bibfnamefont {D.~T.}\ \bibnamefont {Pope}},\ }\bibfield  {title}
  {\emph {\bibinfo {title} {Entanglement and Symmetry: A Case Study in
  Superselection Rules, Reference Frames, and Beyond}},\ }\href
  {https://doi.org/10.1103/PhysRevA.74.062313} {\bibfield  {journal} {\bibinfo
  {journal} {Phys. Rev. A}\ }\textbf {\bibinfo {volume} {74}},\ \bibinfo
  {pages} {062313} (\bibinfo {year} {2006})}\BibitemShut {NoStop}%
\bibitem [{\citenamefont {Gour}\ and\ \citenamefont
  {Spekkens}(2008)}]{gourspekkens2008resource}%
  \BibitemOpen
  \bibfield  {author} {\bibinfo {author} {\bibfnamefont {G.}~\bibnamefont
  {Gour}}\ and\ \bibinfo {author} {\bibfnamefont {R.~W.}\ \bibnamefont
  {Spekkens}},\ }\bibfield  {title} {\emph {\bibinfo {title} {The Resource
  Theory of Quantum Reference Frames: Manipulations and Monotones}},\ }\href
  {https://doi.org/10.1088/1367-2630/10/3/033023} {\bibfield  {journal}
  {\bibinfo  {journal} {New J. Phys.}\ }\textbf {\bibinfo {volume} {10}},\
  \bibinfo {pages} {033023} (\bibinfo {year} {2008})}\BibitemShut {NoStop}%
\bibitem [{\citenamefont {Vaccaro}\ \emph {et~al.}(2008)\citenamefont
  {Vaccaro}, \citenamefont {Anselmi}, \citenamefont {Wiseman},\ and\
  \citenamefont {Jacobs}}]{Vaccaro2008Tradeoff}%
  \BibitemOpen
  \bibfield  {author} {\bibinfo {author} {\bibfnamefont {J.~A.}\ \bibnamefont
  {Vaccaro}}, \bibinfo {author} {\bibfnamefont {F.}~\bibnamefont {Anselmi}},
  \bibinfo {author} {\bibfnamefont {H.~M.}\ \bibnamefont {Wiseman}},\ and\
  \bibinfo {author} {\bibfnamefont {K.}~\bibnamefont {Jacobs}},\ }\bibfield
  {title} {\emph {\bibinfo {title} {Tradeoff Between Extractable Mechanical
  Work, Accessible Entanglement, and Ability to Act As a Reference System,
  Under Arbitrary Superselection Rules}},\ }\href
  {https://doi.org/10.1103/PhysRevA.77.032114} {\bibfield  {journal} {\bibinfo
  {journal} {Phys. Rev. A}\ }\textbf {\bibinfo {volume} {77}},\ \bibinfo
  {pages} {032114} (\bibinfo {year} {2008})}\BibitemShut {NoStop}%
\bibitem [{\citenamefont {Lostaglio}\ \emph
  {et~al.}(2015{\natexlab{a}})\citenamefont {Lostaglio}, \citenamefont
  {Jennings},\ and\ \citenamefont {Rudolph}}]{lostaglio2015description}%
  \BibitemOpen
  \bibfield  {author} {\bibinfo {author} {\bibfnamefont {M.}~\bibnamefont
  {Lostaglio}}, \bibinfo {author} {\bibfnamefont {D.}~\bibnamefont
  {Jennings}},\ and\ \bibinfo {author} {\bibfnamefont {T.}~\bibnamefont
  {Rudolph}},\ }\bibfield  {title} {\emph {\bibinfo {title} {Description of
  Quantum Coherence in Thermodynamic Processes Requires Constraints Beyond Free
  Energy}},\ }\href {https://doi.org/10.1038/ncomms7383} {\bibfield  {journal}
  {\bibinfo  {journal} {Nat. Commun.}\ }\textbf {\bibinfo {volume} {6}},\
  \bibinfo {pages} {1} (\bibinfo {year} {2015}{\natexlab{a}})}\BibitemShut
  {NoStop}%
\bibitem [{\citenamefont {Lostaglio}\ \emph
  {et~al.}(2015{\natexlab{b}})\citenamefont {Lostaglio}, \citenamefont
  {Korzekwa}, \citenamefont {Jennings},\ and\ \citenamefont
  {Rudolph}}]{lostaglio2015quantum}%
  \BibitemOpen
  \bibfield  {author} {\bibinfo {author} {\bibfnamefont {M.}~\bibnamefont
  {Lostaglio}}, \bibinfo {author} {\bibfnamefont {K.}~\bibnamefont {Korzekwa}},
  \bibinfo {author} {\bibfnamefont {D.}~\bibnamefont {Jennings}},\ and\
  \bibinfo {author} {\bibfnamefont {T.}~\bibnamefont {Rudolph}},\ }\bibfield
  {title} {\emph {\bibinfo {title} {Quantum Coherence, Time-Translation
  Symmetry, and Thermodynamics}},\ }\href
  {https://doi.org/10.1103/PhysRevX.5.021001} {\bibfield  {journal} {\bibinfo
  {journal} {Phys. Rev. X}\ }\textbf {\bibinfo {volume} {5}},\ \bibinfo {pages}
  {021001} (\bibinfo {year} {2015}{\natexlab{b}})}\BibitemShut {NoStop}%
\bibitem [{\citenamefont {Marvian}(2020)}]{marvian2020coherence}%
  \BibitemOpen
  \bibfield  {author} {\bibinfo {author} {\bibfnamefont {I.}~\bibnamefont
  {Marvian}},\ }\bibfield  {title} {\emph {\bibinfo {title} {Coherence
  Distillation Machines Are Impossible in Quantum Thermodynamics}},\ }\href
  {https://doi.org/10.1038/s41467-019-13846-3} {\bibfield  {journal} {\bibinfo
  {journal} {Nat. Commun.}\ }\textbf {\bibinfo {volume} {11}},\ \bibinfo
  {pages} {1} (\bibinfo {year} {2020})}\BibitemShut {NoStop}%
\bibitem [{\citenamefont {Wigner}(1952)}]{wigner1952WAY}%
  \BibitemOpen
  \bibfield  {author} {\bibinfo {author} {\bibfnamefont {E.~P.}\ \bibnamefont
  {Wigner}},\ }\bibfield  {title} {\emph {\bibinfo {title} {Die {M}essung
  Quantenmechanischer {O}peratoren; Z}},\ }\href@noop {} {\bibfield  {journal}
  {\bibinfo  {journal} {Z. Phys.}\ }\textbf {\bibinfo {volume} {133}},\
  \bibinfo {pages} {101} (\bibinfo {year} {1952})}\BibitemShut {NoStop}%
\bibitem [{\citenamefont {Araki}\ and\ \citenamefont
  {Yanase}(1960)}]{araki1960measurement}%
  \BibitemOpen
  \bibfield  {author} {\bibinfo {author} {\bibfnamefont {H.}~\bibnamefont
  {Araki}}\ and\ \bibinfo {author} {\bibfnamefont {M.~M.}\ \bibnamefont
  {Yanase}},\ }\bibfield  {title} {\emph {\bibinfo {title} {Measurement of
  Quantum Mechanical Operators}},\ }\href
  {https://doi.org/10.1103/PhysRev.120.622} {\bibfield  {journal} {\bibinfo
  {journal} {Phys. Rev.}\ }\textbf {\bibinfo {volume} {120}},\ \bibinfo {pages}
  {622} (\bibinfo {year} {1960})}\BibitemShut {NoStop}%
\bibitem [{\citenamefont {Yanase}(1961)}]{yanase1961optimal}%
  \BibitemOpen
  \bibfield  {author} {\bibinfo {author} {\bibfnamefont {M.~M.}\ \bibnamefont
  {Yanase}},\ }\bibfield  {title} {\emph {\bibinfo {title} {Optimal Measuring
  Apparatus}},\ }\href {https://doi.org/10.1103/PhysRev.123.666} {\bibfield
  {journal} {\bibinfo  {journal} {Phys. Rev.}\ }\textbf {\bibinfo {volume}
  {123}},\ \bibinfo {pages} {666} (\bibinfo {year} {1961})}\BibitemShut
  {NoStop}%
\bibitem [{\citenamefont {Marvian}\ and\ \citenamefont
  {Spekkens}(2012)}]{marvian2012information}%
  \BibitemOpen
  \bibfield  {author} {\bibinfo {author} {\bibfnamefont {I.}~\bibnamefont
  {Marvian}}\ and\ \bibinfo {author} {\bibfnamefont {R.~W.}\ \bibnamefont
  {Spekkens}},\ }\bibfield  {title} {\emph {\bibinfo {title} {An
  Information-Theoretic Account of the {W}igner-{A}raki-{Y}anase Theorem}},\
  }\href@noop {} {\  (\bibinfo {year} {2012})},\ \Eprint
  {https://arxiv.org/abs/1212.3378} {arXiv:1212.3378} \BibitemShut {NoStop}%
\bibitem [{\citenamefont {Ahmadi}\ \emph {et~al.}(2013)\citenamefont {Ahmadi},
  \citenamefont {Jennings},\ and\ \citenamefont {Rudolph}}]{ahmadi2013WAY}%
  \BibitemOpen
  \bibfield  {author} {\bibinfo {author} {\bibfnamefont {M.}~\bibnamefont
  {Ahmadi}}, \bibinfo {author} {\bibfnamefont {D.}~\bibnamefont {Jennings}},\
  and\ \bibinfo {author} {\bibfnamefont {T.}~\bibnamefont {Rudolph}},\
  }\bibfield  {title} {\emph {\bibinfo {title} {The {W}igner-{A}raki-{Y}anase
  Theorem and the Quantum Resource Theory of Asymmetry}},\ }\href
  {https://doi.org/10.1088/1367-2630/15/1/013057} {\bibfield  {journal}
  {\bibinfo  {journal} {New J. Phys.}\ }\textbf {\bibinfo {volume} {15}},\
  \bibinfo {pages} {013057} (\bibinfo {year} {2013})}\BibitemShut {NoStop}%
\bibitem [{\citenamefont {Yadin}\ and\ \citenamefont
  {Vedral}(2016)}]{yadin2016general}%
  \BibitemOpen
  \bibfield  {author} {\bibinfo {author} {\bibfnamefont {B.}~\bibnamefont
  {Yadin}}\ and\ \bibinfo {author} {\bibfnamefont {V.}~\bibnamefont {Vedral}},\
  }\bibfield  {title} {\emph {\bibinfo {title} {General Framework for Quantum
  Macroscopicity in Terms of Coherence}},\ }\href
  {https://doi.org/10.1103/PhysRevA.93.022122} {\bibfield  {journal} {\bibinfo
  {journal} {Phys. Rev. A}\ }\textbf {\bibinfo {volume} {93}},\ \bibinfo
  {pages} {022122} (\bibinfo {year} {2016})}\BibitemShut {NoStop}%
\bibitem [{\citenamefont {Marvian}\ \emph {et~al.}(2016)\citenamefont
  {Marvian}, \citenamefont {Spekkens},\ and\ \citenamefont
  {Zanardi}}]{Marvian2016Speed}%
  \BibitemOpen
  \bibfield  {author} {\bibinfo {author} {\bibfnamefont {I.}~\bibnamefont
  {Marvian}}, \bibinfo {author} {\bibfnamefont {R.~W.}\ \bibnamefont
  {Spekkens}},\ and\ \bibinfo {author} {\bibfnamefont {P.}~\bibnamefont
  {Zanardi}},\ }\bibfield  {title} {\emph {\bibinfo {title} {Quantum Speed
  Limits, Coherence, and Asymmetry}},\ }\href
  {https://doi.org/10.1103/PhysRevA.93.052331} {\bibfield  {journal} {\bibinfo
  {journal} {Phys. Rev. A}\ }\textbf {\bibinfo {volume} {93}},\ \bibinfo
  {pages} {052331} (\bibinfo {year} {2016})}\BibitemShut {NoStop}%
\bibitem [{\citenamefont {Rovelli}(1991)}]{rovelli1991quantum}%
  \BibitemOpen
  \bibfield  {author} {\bibinfo {author} {\bibfnamefont {C.}~\bibnamefont
  {Rovelli}},\ }\bibfield  {title} {\emph {\bibinfo {title} {Quantum Reference
  Systems}},\ }\href {https://doi.org/10.1088/0264-9381/8/2/012} {\bibfield
  {journal} {\bibinfo  {journal} {Classical and Quantum Gravity}\ }\textbf
  {\bibinfo {volume} {8}},\ \bibinfo {pages} {317} (\bibinfo {year}
  {1991})}\BibitemShut {NoStop}%
\bibitem [{\citenamefont {Rovelli}(1996)}]{rovelli1996relational}%
  \BibitemOpen
  \bibfield  {author} {\bibinfo {author} {\bibfnamefont {C.}~\bibnamefont
  {Rovelli}},\ }\bibfield  {title} {\emph {\bibinfo {title} {Relational Quantum
  Mechanics}},\ }\href {https://doi.org/10.1007/BF02302261} {\bibfield
  {journal} {\bibinfo  {journal} {Int. J. Theor. Phys.}\ }\textbf {\bibinfo
  {volume} {35}},\ \bibinfo {pages} {1637} (\bibinfo {year}
  {1996})}\BibitemShut {NoStop}%
\bibitem [{\citenamefont {Marletto}\ and\ \citenamefont
  {Vedral}(2017)}]{Vedral2017Evolution}%
  \BibitemOpen
  \bibfield  {author} {\bibinfo {author} {\bibfnamefont {C.}~\bibnamefont
  {Marletto}}\ and\ \bibinfo {author} {\bibfnamefont {V.}~\bibnamefont
  {Vedral}},\ }\bibfield  {title} {\emph {\bibinfo {title} {Evolution Without
  Evolution and Without Ambiguities}},\ }\href
  {https://doi.org/10.1103/PhysRevD.95.043510} {\bibfield  {journal} {\bibinfo
  {journal} {Phys. Rev. D}\ }\textbf {\bibinfo {volume} {95}},\ \bibinfo
  {pages} {043510} (\bibinfo {year} {2017})}\BibitemShut {NoStop}%
\bibitem [{\citenamefont {Nikolova}\ \emph {et~al.}(2018)\citenamefont
  {Nikolova}, \citenamefont {Brennen}, \citenamefont {Osborne}, \citenamefont
  {Milburn},\ and\ \citenamefont {Stace}}]{Nikolova2018Relational}%
  \BibitemOpen
  \bibfield  {author} {\bibinfo {author} {\bibfnamefont {A.}~\bibnamefont
  {Nikolova}}, \bibinfo {author} {\bibfnamefont {G.~K.}\ \bibnamefont
  {Brennen}}, \bibinfo {author} {\bibfnamefont {T.~J.}\ \bibnamefont
  {Osborne}}, \bibinfo {author} {\bibfnamefont {G.~J.}\ \bibnamefont
  {Milburn}},\ and\ \bibinfo {author} {\bibfnamefont {T.~M.}\ \bibnamefont
  {Stace}},\ }\bibfield  {title} {\emph {\bibinfo {title} {Relational Time in
  Anyonic Systems}},\ }\href {https://doi.org/10.1103/PhysRevA.97.030101}
  {\bibfield  {journal} {\bibinfo  {journal} {Phys. Rev. A}\ }\textbf {\bibinfo
  {volume} {97}},\ \bibinfo {pages} {030101(R)} (\bibinfo {year}
  {2018})}\BibitemShut {NoStop}%
\bibitem [{\citenamefont {Giacomini}\ \emph {et~al.}(2019)\citenamefont
  {Giacomini}, \citenamefont {Castro-Ruiz},\ and\ \citenamefont
  {Brukner}}]{giacomini2019covariance}%
  \BibitemOpen
  \bibfield  {author} {\bibinfo {author} {\bibfnamefont {F.}~\bibnamefont
  {Giacomini}}, \bibinfo {author} {\bibfnamefont {E.}~\bibnamefont
  {Castro-Ruiz}},\ and\ \bibinfo {author} {\bibfnamefont
  {{\v{C}}.}~\bibnamefont {Brukner}},\ }\bibfield  {title} {\emph {\bibinfo
  {title} {Quantum Mechanics and the Covariance of Physical Laws in Quantum
  Reference Frames}},\ }\href {https://doi.org/10.1038/s41467-018-08155-0}
  {\bibfield  {journal} {\bibinfo  {journal} {Nat. Commun.}\ }\textbf {\bibinfo
  {volume} {10}},\ \bibinfo {pages} {1} (\bibinfo {year} {2019})}\BibitemShut
  {NoStop}%
\bibitem [{\citenamefont {Loveridge}\ and\ \citenamefont
  {Miyadera}(2019)}]{loveridge2019relative}%
  \BibitemOpen
  \bibfield  {author} {\bibinfo {author} {\bibfnamefont {L.}~\bibnamefont
  {Loveridge}}\ and\ \bibinfo {author} {\bibfnamefont {T.}~\bibnamefont
  {Miyadera}},\ }\bibfield  {title} {\emph {\bibinfo {title} {Relative Quantum
  Time}},\ }\href@noop {} {\bibfield  {journal} {\bibinfo  {journal} {Found.
  Phys.}\ }\textbf {\bibinfo {volume} {49}},\ \bibinfo {pages} {549} (\bibinfo
  {year} {2019})}\BibitemShut {NoStop}%
\bibitem [{\citenamefont {Smith}\ and\ \citenamefont
  {Ahmadi}(2019)}]{smith2019quantizing}%
  \BibitemOpen
  \bibfield  {author} {\bibinfo {author} {\bibfnamefont {A.~R.~H.}\
  \bibnamefont {Smith}}\ and\ \bibinfo {author} {\bibfnamefont
  {M.}~\bibnamefont {Ahmadi}},\ }\bibfield  {title} {\emph {\bibinfo {title}
  {Quantizing Time: Interacting Clocks and Systems}},\ }\href
  {https://doi.org/10.22331/q-2019-07-08-160} {\bibfield  {journal} {\bibinfo
  {journal} {{Quantum}}\ }\textbf {\bibinfo {volume} {3}},\ \bibinfo {pages}
  {160} (\bibinfo {year} {2019})}\BibitemShut {NoStop}%
\bibitem [{\citenamefont {Martinelli}\ and\ \citenamefont
  {Soares-Pinto}(2019)}]{Martinelli2019Quantifying}%
  \BibitemOpen
  \bibfield  {author} {\bibinfo {author} {\bibfnamefont {T.}~\bibnamefont
  {Martinelli}}\ and\ \bibinfo {author} {\bibfnamefont {D.~O.}\ \bibnamefont
  {Soares-Pinto}},\ }\bibfield  {title} {\emph {\bibinfo {title} {Quantifying
  Quantum Reference Frames in Composed Systems: Local, Global, and Mutual
  Asymmetries}},\ }\href {https://doi.org/10.1103/PhysRevA.99.042124}
  {\bibfield  {journal} {\bibinfo  {journal} {Phys. Rev. A}\ }\textbf {\bibinfo
  {volume} {99}},\ \bibinfo {pages} {042124} (\bibinfo {year}
  {2019})}\BibitemShut {NoStop}%
\bibitem [{\citenamefont {Mendes}\ and\ \citenamefont
  {Soares-Pinto}(2019)}]{mendes2019time}%
  \BibitemOpen
  \bibfield  {author} {\bibinfo {author} {\bibfnamefont {L.~R.}\ \bibnamefont
  {Mendes}}\ and\ \bibinfo {author} {\bibfnamefont {D.~O.}\ \bibnamefont
  {Soares-Pinto}},\ }\bibfield  {title} {\emph {\bibinfo {title} {Time As a
  Consequence of Internal Coherence}},\ }\href
  {https://doi.org/10.1098/rspa.2019.0470} {\bibfield  {journal} {\bibinfo
  {journal} {Proc. R. Soc. A}\ }\textbf {\bibinfo {volume} {475}},\ \bibinfo
  {pages} {20190470} (\bibinfo {year} {2019})}\BibitemShut {NoStop}%
\bibitem [{\citenamefont {Vanrietvelde}\ \emph {et~al.}(2020)\citenamefont
  {Vanrietvelde}, \citenamefont {Hoehn}, \citenamefont {Giacomini},\ and\
  \citenamefont {Castro-Ruiz}}]{vanrietvelde2020change}%
  \BibitemOpen
  \bibfield  {author} {\bibinfo {author} {\bibfnamefont {A.}~\bibnamefont
  {Vanrietvelde}}, \bibinfo {author} {\bibfnamefont {P.~A.}\ \bibnamefont
  {Hoehn}}, \bibinfo {author} {\bibfnamefont {F.}~\bibnamefont {Giacomini}},\
  and\ \bibinfo {author} {\bibfnamefont {E.}~\bibnamefont {Castro-Ruiz}},\
  }\bibfield  {title} {\emph {\bibinfo {title} {A Change of Perspective:
  Switching Quantum Reference Frames Via a Perspective-Neutral Framework}},\
  }\href {https://doi.org/10.22331/q-2020-01-27-225} {\bibfield  {journal}
  {\bibinfo  {journal} {{Quantum}}\ }\textbf {\bibinfo {volume} {4}},\ \bibinfo
  {pages} {225} (\bibinfo {year} {2020})}\BibitemShut {NoStop}%
\bibitem [{\citenamefont {Carmo}\ and\ \citenamefont
  {Soares-Pinto}(2021)}]{Carmo2021Quantifying}%
  \BibitemOpen
  \bibfield  {author} {\bibinfo {author} {\bibfnamefont {R.~S.}\ \bibnamefont
  {Carmo}}\ and\ \bibinfo {author} {\bibfnamefont {D.~O.}\ \bibnamefont
  {Soares-Pinto}},\ }\bibfield  {title} {\emph {\bibinfo {title} {Quantifying
  Resources for the {P}age-{W}ootters Mechanism: Shared Asymmetry As Relative
  Entropy of Entanglement}},\ }\href
  {https://doi.org/10.1103/PhysRevA.103.052420} {\bibfield  {journal} {\bibinfo
   {journal} {Phys. Rev. A}\ }\textbf {\bibinfo {volume} {103}},\ \bibinfo
  {pages} {052420} (\bibinfo {year} {2021})}\BibitemShut {NoStop}%
\bibitem [{\citenamefont {Chataignier}(2021)}]{Chataignier2021Relational}%
  \BibitemOpen
  \bibfield  {author} {\bibinfo {author} {\bibfnamefont {L.}~\bibnamefont
  {Chataignier}},\ }\bibfield  {title} {\emph {\bibinfo {title} {Relational
  Observables, Reference Frames, and Conditional Probabilities}},\ }\href
  {https://doi.org/10.1103/PhysRevD.103.026013} {\bibfield  {journal} {\bibinfo
   {journal} {Phys. Rev. D}\ }\textbf {\bibinfo {volume} {103}},\ \bibinfo
  {pages} {026013} (\bibinfo {year} {2021})}\BibitemShut {NoStop}%
\bibitem [{\citenamefont {Eastin}\ and\ \citenamefont
  {Knill}(2009)}]{EastinKnill2009Restrictions}%
  \BibitemOpen
  \bibfield  {author} {\bibinfo {author} {\bibfnamefont {B.}~\bibnamefont
  {Eastin}}\ and\ \bibinfo {author} {\bibfnamefont {E.}~\bibnamefont {Knill}},\
  }\bibfield  {title} {\emph {\bibinfo {title} {Restrictions on Transversal
  Encoded Quantum Gate Sets}},\ }\href
  {https://doi.org/10.1103/PhysRevLett.102.110502} {\bibfield  {journal}
  {\bibinfo  {journal} {Phys. Rev. Lett.}\ }\textbf {\bibinfo {volume} {102}},\
  \bibinfo {pages} {110502} (\bibinfo {year} {2009})}\BibitemShut {NoStop}%
\bibitem [{\citenamefont {Faist}\ \emph {et~al.}(2020)\citenamefont {Faist},
  \citenamefont {Nezami}, \citenamefont {Albert}, \citenamefont {Salton},
  \citenamefont {Pastawski}, \citenamefont {Hayden},\ and\ \citenamefont
  {Preskill}}]{Faist2020_Approx_QEC}%
  \BibitemOpen
  \bibfield  {author} {\bibinfo {author} {\bibfnamefont {P.}~\bibnamefont
  {Faist}}, \bibinfo {author} {\bibfnamefont {S.}~\bibnamefont {Nezami}},
  \bibinfo {author} {\bibfnamefont {V.~V.}\ \bibnamefont {Albert}}, \bibinfo
  {author} {\bibfnamefont {G.}~\bibnamefont {Salton}}, \bibinfo {author}
  {\bibfnamefont {F.}~\bibnamefont {Pastawski}}, \bibinfo {author}
  {\bibfnamefont {P.}~\bibnamefont {Hayden}},\ and\ \bibinfo {author}
  {\bibfnamefont {J.}~\bibnamefont {Preskill}},\ }\bibfield  {title} {\emph
  {\bibinfo {title} {Continuous Symmetries and Approximate Quantum Error
  Correction}},\ }\href {https://doi.org/10.1103/PhysRevX.10.041018} {\bibfield
   {journal} {\bibinfo  {journal} {Phys. Rev. X}\ }\textbf {\bibinfo {volume}
  {10}},\ \bibinfo {pages} {041018} (\bibinfo {year} {2020})}\BibitemShut
  {NoStop}%
\bibitem [{\citenamefont {Woods}\ and\ \citenamefont
  {Alhambra}(2020)}]{Woods2020continuousgroupsof}%
  \BibitemOpen
  \bibfield  {author} {\bibinfo {author} {\bibfnamefont {M.~P.}\ \bibnamefont
  {Woods}}\ and\ \bibinfo {author} {\bibfnamefont {{\'{A}}.~M.}\ \bibnamefont
  {Alhambra}},\ }\bibfield  {title} {\emph {\bibinfo {title} {Continuous Groups
  of Transversal Gates for Quantum Error Correcting Codes from Finite Clock
  Reference Frames}},\ }\href {https://doi.org/10.22331/q-2020-03-23-245}
  {\bibfield  {journal} {\bibinfo  {journal} {{Quantum}}\ }\textbf {\bibinfo
  {volume} {4}},\ \bibinfo {pages} {245} (\bibinfo {year} {2020})}\BibitemShut
  {NoStop}%
\bibitem [{\citenamefont {Yang}\ \emph {et~al.}(2020)\citenamefont {Yang},
  \citenamefont {Mo}, \citenamefont {Renes}, \citenamefont {Chiribella},\ and\
  \citenamefont {Woods}}]{yang2020covariant}%
  \BibitemOpen
  \bibfield  {author} {\bibinfo {author} {\bibfnamefont {Y.}~\bibnamefont
  {Yang}}, \bibinfo {author} {\bibfnamefont {Y.}~\bibnamefont {Mo}}, \bibinfo
  {author} {\bibfnamefont {J.~M.}\ \bibnamefont {Renes}}, \bibinfo {author}
  {\bibfnamefont {G.}~\bibnamefont {Chiribella}},\ and\ \bibinfo {author}
  {\bibfnamefont {M.~P.}\ \bibnamefont {Woods}},\ }\bibfield  {title} {\emph
  {\bibinfo {title} {Covariant Quantum Error Correcting Codes Via Reference
  Frames}},\ }\href@noop {} {\  (\bibinfo {year} {2020})},\ \Eprint
  {https://arxiv.org/abs/arXiv:2007.09154} {arXiv:2007.09154} \BibitemShut
  {NoStop}%
\bibitem [{\citenamefont {Almheiri}\ \emph {et~al.}(2015)\citenamefont
  {Almheiri}, \citenamefont {Dong},\ and\ \citenamefont
  {Harlow}}]{almheiri2015bulk}%
  \BibitemOpen
  \bibfield  {author} {\bibinfo {author} {\bibfnamefont {A.}~\bibnamefont
  {Almheiri}}, \bibinfo {author} {\bibfnamefont {X.}~\bibnamefont {Dong}},\
  and\ \bibinfo {author} {\bibfnamefont {D.}~\bibnamefont {Harlow}},\
  }\bibfield  {title} {\emph {\bibinfo {title} {Bulk Locality and Quantum Error
  Correction in {A}d{S}/{CFT}}},\ }\href
  {https://doi.org/10.1007/JHEP04(2015)163} {\bibfield  {journal} {\bibinfo
  {journal} {J. High Energy Phys.}\ }\textbf {\bibinfo {volume} {2015}}\bibinfo
   {number} { (4)},\ \bibinfo {pages} {163}}\BibitemShut {NoStop}%
\bibitem [{\citenamefont {Pastawski}\ \emph {et~al.}(2015)\citenamefont
  {Pastawski}, \citenamefont {Yoshida}, \citenamefont {Harlow},\ and\
  \citenamefont {Preskill}}]{pastawski2015holographic}%
  \BibitemOpen
\bibfield  {number} {  }\bibfield  {author} {\bibinfo {author} {\bibfnamefont
  {F.}~\bibnamefont {Pastawski}}, \bibinfo {author} {\bibfnamefont
  {B.}~\bibnamefont {Yoshida}}, \bibinfo {author} {\bibfnamefont
  {D.}~\bibnamefont {Harlow}},\ and\ \bibinfo {author} {\bibfnamefont
  {J.}~\bibnamefont {Preskill}},\ }\bibfield  {title} {\emph {\bibinfo {title}
  {Holographic Quantum Error-Correcting Codes: Toy Models for the Bulk/boundary
  Correspondence}},\ }\href {https://doi.org/10.1007/JHEP06(2015)149}
  {\bibfield  {journal} {\bibinfo  {journal} {J. High Energy Phys.}\ }\textbf
  {\bibinfo {volume} {2015}}\bibinfo  {number} { (6)},\ \bibinfo {pages}
  {1}}\BibitemShut {NoStop}%
\bibitem [{\citenamefont {Gschwendtner}\ \emph {et~al.}(2021)\citenamefont
  {Gschwendtner}, \citenamefont {Bluhm},\ and\ \citenamefont
  {Winter}}]{gschwendtner2021programmability}%
  \BibitemOpen
\bibfield  {number} {  }\bibfield  {author} {\bibinfo {author} {\bibfnamefont
  {M.}~\bibnamefont {Gschwendtner}}, \bibinfo {author} {\bibfnamefont
  {A.}~\bibnamefont {Bluhm}},\ and\ \bibinfo {author} {\bibfnamefont
  {A.}~\bibnamefont {Winter}},\ }\bibfield  {title} {\emph {\bibinfo {title}
  {Programmability of Covariant Quantum Channels}},\ }\href
  {https://doi.org/10.22331/q-2021-06-29-488} {\bibfield  {journal} {\bibinfo
  {journal} {{Quantum}}\ }\textbf {\bibinfo {volume} {5}},\ \bibinfo {pages}
  {488} (\bibinfo {year} {2021})}\BibitemShut {NoStop}%
\bibitem [{\citenamefont {Gour}\ \emph {et~al.}(2018)\citenamefont {Gour},
  \citenamefont {Jennings}, \citenamefont {Buscemi}, \citenamefont {Duan},\
  and\ \citenamefont {Marvian}}]{gour2018quantum}%
  \BibitemOpen
  \bibfield  {author} {\bibinfo {author} {\bibfnamefont {G.}~\bibnamefont
  {Gour}}, \bibinfo {author} {\bibfnamefont {D.}~\bibnamefont {Jennings}},
  \bibinfo {author} {\bibfnamefont {F.}~\bibnamefont {Buscemi}}, \bibinfo
  {author} {\bibfnamefont {R.}~\bibnamefont {Duan}},\ and\ \bibinfo {author}
  {\bibfnamefont {I.}~\bibnamefont {Marvian}},\ }\bibfield  {title} {\emph
  {\bibinfo {title} {Quantum Majorization and a Complete Set of Entropic
  Conditions for Quantum Thermodynamics}},\ }\href
  {https://doi.org/10.1038/s41467-018-06261-7} {\bibfield  {journal} {\bibinfo
  {journal} {Nat. Commun.}\ }\textbf {\bibinfo {volume} {9}},\ \bibinfo {pages}
  {1} (\bibinfo {year} {2018})}\BibitemShut {NoStop}%
\bibitem [{\citenamefont {Renner}(2005)}]{renner2005security}%
  \BibitemOpen
  \bibfield  {author} {\bibinfo {author} {\bibfnamefont {R.}~\bibnamefont
  {Renner}},\ }\emph {\bibinfo {title} {Security of {QKD}}},\ \href@noop {}
  {Ph.D. thesis},\ \bibinfo  {school} {ETH, 2005} (\bibinfo {year} {2005}),\
  \Eprint {https://arxiv.org/abs/arXiv:quant-ph/0512258}
  {arXiv:quant-ph/0512258} \BibitemShut {NoStop}%
\bibitem [{\citenamefont {Marvian}\ and\ \citenamefont
  {Spekkens}(2014{\natexlab{b}})}]{marvian2014modes}%
  \BibitemOpen
  \bibfield  {author} {\bibinfo {author} {\bibfnamefont {I.}~\bibnamefont
  {Marvian}}\ and\ \bibinfo {author} {\bibfnamefont {R.~W.}\ \bibnamefont
  {Spekkens}},\ }\bibfield  {title} {\emph {\bibinfo {title} {Modes of
  Asymmetry: The Application of Harmonic Analysis to Symmetric Quantum Dynamics
  and Quantum Reference Frames}},\ }\href
  {https://doi.org/10.1103/PhysRevA.90.062110} {\bibfield  {journal} {\bibinfo
  {journal} {Phys. Rev. A}\ }\textbf {\bibinfo {volume} {90}},\ \bibinfo
  {pages} {062110} (\bibinfo {year} {2014}{\natexlab{b}})}\BibitemShut
  {NoStop}%
\bibitem [{\citenamefont {M{\"u}ller-Lennert}\ \emph
  {et~al.}(2013)\citenamefont {M{\"u}ller-Lennert}, \citenamefont {Dupuis},
  \citenamefont {Szehr}, \citenamefont {Fehr},\ and\ \citenamefont
  {Tomamichel}}]{muller2013quantum}%
  \BibitemOpen
  \bibfield  {author} {\bibinfo {author} {\bibfnamefont {M.}~\bibnamefont
  {M{\"u}ller-Lennert}}, \bibinfo {author} {\bibfnamefont {F.}~\bibnamefont
  {Dupuis}}, \bibinfo {author} {\bibfnamefont {O.}~\bibnamefont {Szehr}},
  \bibinfo {author} {\bibfnamefont {S.}~\bibnamefont {Fehr}},\ and\ \bibinfo
  {author} {\bibfnamefont {M.}~\bibnamefont {Tomamichel}},\ }\bibfield  {title}
  {\emph {\bibinfo {title} {On Quantum {R}{\'e}nyi Entropies: A New
  Generalization and Some Properties}},\ }\href
  {https://doi.org/10.1063/1.4838856} {\bibfield  {journal} {\bibinfo
  {journal} {J. Math. Phys.}\ }\textbf {\bibinfo {volume} {54}},\ \bibinfo
  {pages} {122203} (\bibinfo {year} {2013})}\BibitemShut {NoStop}%
\bibitem [{\citenamefont {Wilde}\ \emph {et~al.}(2014)\citenamefont {Wilde},
  \citenamefont {Winter},\ and\ \citenamefont {Yang}}]{wilde2014strong}%
  \BibitemOpen
  \bibfield  {author} {\bibinfo {author} {\bibfnamefont {M.~M.}\ \bibnamefont
  {Wilde}}, \bibinfo {author} {\bibfnamefont {A.}~\bibnamefont {Winter}},\ and\
  \bibinfo {author} {\bibfnamefont {D.}~\bibnamefont {Yang}},\ }\bibfield
  {title} {\emph {\bibinfo {title} {Strong Converse for the Classical Capacity
  of Entanglement-Breaking and {H}adamard Channels Via a Sandwiched {R}{\'e}nyi
  Relative Entropy}},\ }\href {https://doi.org/10.1007/s00220-014-2122-x}
  {\bibfield  {journal} {\bibinfo  {journal} {Commun. Math. Phys.}\ }\textbf
  {\bibinfo {volume} {331}},\ \bibinfo {pages} {593} (\bibinfo {year}
  {2014})}\BibitemShut {NoStop}%
\bibitem [{\citenamefont {Gour}\ \emph {et~al.}(2009)\citenamefont {Gour},
  \citenamefont {Marvian},\ and\ \citenamefont {Spekkens}}]{Gour2009Measuring}%
  \BibitemOpen
  \bibfield  {author} {\bibinfo {author} {\bibfnamefont {G.}~\bibnamefont
  {Gour}}, \bibinfo {author} {\bibfnamefont {I.}~\bibnamefont {Marvian}},\ and\
  \bibinfo {author} {\bibfnamefont {R.~W.}\ \bibnamefont {Spekkens}},\
  }\bibfield  {title} {\emph {\bibinfo {title} {Measuring the Quality of a
  Quantum Reference Frame: The Relative Entropy of Frameness}},\ }\href
  {https://doi.org/10.1103/PhysRevA.80.012307} {\bibfield  {journal} {\bibinfo
  {journal} {Phys. Rev. A}\ }\textbf {\bibinfo {volume} {80}},\ \bibinfo
  {pages} {012307} (\bibinfo {year} {2009})}\BibitemShut {NoStop}%
\bibitem [{\citenamefont {Marvian~Mashhad}(2012)}]{marvian2012symmetry}%
  \BibitemOpen
  \bibfield  {author} {\bibinfo {author} {\bibfnamefont {I.}~\bibnamefont
  {Marvian~Mashhad}},\ }\emph {\bibinfo {title} {Symmetry, Asymmetry and
  Quantum Information}},\ \href@noop {} {Ph.D. thesis},\ \bibinfo  {school}
  {University of Waterloo} (\bibinfo {year} {2012})\BibitemShut {NoStop}%
\bibitem [{\citenamefont {DeWitt}(1967)}]{dewitt1967}%
  \BibitemOpen
  \bibfield  {author} {\bibinfo {author} {\bibfnamefont {B.~S.}\ \bibnamefont
  {DeWitt}},\ }\bibfield  {title} {\emph {\bibinfo {title} {Quantum Theory of
  Gravity: {I}. {T}he Canonical Theory}},\ }\href
  {https://doi.org/10.1103/PhysRev.160.1113} {\bibfield  {journal} {\bibinfo
  {journal} {Phys. Rev.}\ }\textbf {\bibinfo {volume} {160}},\ \bibinfo {pages}
  {1113} (\bibinfo {year} {1967})}\BibitemShut {NoStop}%
\bibitem [{\citenamefont {Page}\ and\ \citenamefont
  {Wootters}(1983)}]{PageWooters1983}%
  \BibitemOpen
  \bibfield  {author} {\bibinfo {author} {\bibfnamefont {D.~N.}\ \bibnamefont
  {Page}}\ and\ \bibinfo {author} {\bibfnamefont {W.~K.}\ \bibnamefont
  {Wootters}},\ }\bibfield  {title} {\emph {\bibinfo {title} {Evolution Without
  Evolution: Dynamics Described by Stationary Observables}},\ }\href
  {https://doi.org/10.1103/PhysRevD.27.2885} {\bibfield  {journal} {\bibinfo
  {journal} {Phys. Rev. D}\ }\textbf {\bibinfo {volume} {27}},\ \bibinfo
  {pages} {2885} (\bibinfo {year} {1983})}\BibitemShut {NoStop}%
\bibitem [{\citenamefont {Cirstoiu}\ and\ \citenamefont
  {Jennings}(2017)}]{Cirstoiu2017Global}%
  \BibitemOpen
  \bibfield  {author} {\bibinfo {author} {\bibfnamefont {C.}~\bibnamefont
  {Cirstoiu}}\ and\ \bibinfo {author} {\bibfnamefont {D.}~\bibnamefont
  {Jennings}},\ }\bibfield  {title} {\emph {\bibinfo {title} {Global and Local
  Gauge Symmetries Beyond {L}agrangian Formulations}},\ }\href@noop {} {\
  (\bibinfo {year} {2017})},\ \Eprint {https://arxiv.org/abs/arXiv:1707.09826}
  {arXiv:1707.09826} \BibitemShut {NoStop}%
\bibitem [{\citenamefont {Konig}\ \emph {et~al.}(2009)\citenamefont {Konig},
  \citenamefont {Renner},\ and\ \citenamefont
  {Schaffner}}]{konig2009operational}%
  \BibitemOpen
  \bibfield  {author} {\bibinfo {author} {\bibfnamefont {R.}~\bibnamefont
  {Konig}}, \bibinfo {author} {\bibfnamefont {R.}~\bibnamefont {Renner}},\ and\
  \bibinfo {author} {\bibfnamefont {C.}~\bibnamefont {Schaffner}},\ }\bibfield
  {title} {\emph {\bibinfo {title} {The Operational Meaning of Min-and
  Max-Entropy}},\ }\href@noop {} {\bibfield  {journal} {\bibinfo  {journal}
  {IEEE Trans. Inf. Theory}\ }\textbf {\bibinfo {volume} {55}},\ \bibinfo
  {pages} {4337} (\bibinfo {year} {2009})}\BibitemShut {NoStop}%
\bibitem [{\citenamefont {Tomamichel}(2012)}]{tomamichel2012framework}%
  \BibitemOpen
  \bibfield  {author} {\bibinfo {author} {\bibfnamefont {M.}~\bibnamefont
  {Tomamichel}},\ }\bibfield  {title} {\emph {\bibinfo {title} {A Framework for
  Non-Asymptotic Quantum Information Theory}},\ }\href@noop {} {\  (\bibinfo
  {year} {2012})},\ \Eprint {https://arxiv.org/abs/arXiv:1203.2142}
  {arXiv:1203.2142} \BibitemShut {NoStop}%
\bibitem [{\citenamefont {Tomamichel}\ \emph {et~al.}(2010)\citenamefont
  {Tomamichel}, \citenamefont {Colbeck},\ and\ \citenamefont
  {Renner}}]{tomamichel2010duality}%
  \BibitemOpen
  \bibfield  {author} {\bibinfo {author} {\bibfnamefont {M.}~\bibnamefont
  {Tomamichel}}, \bibinfo {author} {\bibfnamefont {R.}~\bibnamefont
  {Colbeck}},\ and\ \bibinfo {author} {\bibfnamefont {R.}~\bibnamefont
  {Renner}},\ }\bibfield  {title} {\emph {\bibinfo {title} {Duality Between
  Smooth Min- and Max-Entropies}},\ }\href
  {https://doi.org/10.1109/TIT.2010.2054130} {\bibfield  {journal} {\bibinfo
  {journal} {IEEE Trans. Inf. Theory}\ }\textbf {\bibinfo {volume} {56}},\
  \bibinfo {pages} {4674} (\bibinfo {year} {2010})}\BibitemShut {NoStop}%
\bibitem [{\citenamefont {Tkocz}(2019)}]{tkocz2019introduction}%
  \BibitemOpen
  \bibfield  {author} {\bibinfo {author} {\bibfnamefont {T.}~\bibnamefont
  {Tkocz}},\ }\href@noop {} {\bibinfo {title} {An Introduction to Convex and
  Discrete Geometry (Lecture Notes)}} (\bibinfo {year} {2019})\BibitemShut
  {NoStop}%
\bibitem [{\citenamefont {Ledoux}\ and\ \citenamefont
  {Talagrand}(1991)}]{Ledoux1991}%
  \BibitemOpen
  \bibfield  {author} {\bibinfo {author} {\bibfnamefont {M.}~\bibnamefont
  {Ledoux}}\ and\ \bibinfo {author} {\bibfnamefont {M.}~\bibnamefont
  {Talagrand}},\ }\href
  {https://doi.org/https://doi.org/10.1007/978-3-642-20212-4} {\emph {\bibinfo
  {title} {Probability in {B}anach Spaces : Isoperimetry and Processes}}},\
  Ergebnisse der Mathematik und ihrer Grenzgebiete ; 3. Folge, Band 23\
  (\bibinfo  {publisher} {Springer-Verlag},\ \bibinfo {year}
  {1991})\BibitemShut {NoStop}%
\bibitem [{\citenamefont {Bartlett}\ \emph {et~al.}(2009)\citenamefont
  {Bartlett}, \citenamefont {Rudolph}, \citenamefont {Spekkens},\ and\
  \citenamefont {Turner}}]{bartlett2009quantum}%
  \BibitemOpen
  \bibfield  {author} {\bibinfo {author} {\bibfnamefont {S.~D.}\ \bibnamefont
  {Bartlett}}, \bibinfo {author} {\bibfnamefont {T.}~\bibnamefont {Rudolph}},
  \bibinfo {author} {\bibfnamefont {R.~W.}\ \bibnamefont {Spekkens}},\ and\
  \bibinfo {author} {\bibfnamefont {P.~S.}\ \bibnamefont {Turner}},\ }\bibfield
   {title} {\emph {\bibinfo {title} {Quantum Communication Using a Bounded-Size
  Quantum Reference Frame}},\ }\href
  {https://doi.org/10.1088/1367-2630/11/6/063013} {\bibfield  {journal}
  {\bibinfo  {journal} {New J. Phys.}\ }\textbf {\bibinfo {volume} {11}},\
  \bibinfo {pages} {063013} (\bibinfo {year} {2009})}\BibitemShut {NoStop}%
\bibitem [{\citenamefont {Loveridge}\ \emph {et~al.}(2017)\citenamefont
  {Loveridge}, \citenamefont {Busch},\ and\ \citenamefont
  {Miyadera}}]{Loveridge2017Relativity}%
  \BibitemOpen
  \bibfield  {author} {\bibinfo {author} {\bibfnamefont {L.}~\bibnamefont
  {Loveridge}}, \bibinfo {author} {\bibfnamefont {P.}~\bibnamefont {Busch}},\
  and\ \bibinfo {author} {\bibfnamefont {T.}~\bibnamefont {Miyadera}},\
  }\bibfield  {title} {\emph {\bibinfo {title} {Relativity of Quantum States
  and Observables}},\ }\href {https://doi.org/10.1209/0295-5075/117/40004}
  {\bibfield  {journal} {\bibinfo  {journal} {Europhys. Lett.}\ }\textbf
  {\bibinfo {volume} {117}},\ \bibinfo {pages} {40004} (\bibinfo {year}
  {2017})}\BibitemShut {NoStop}%
\bibitem [{\citenamefont {Loveridge}\ \emph {et~al.}(2018)\citenamefont
  {Loveridge}, \citenamefont {Miyadera},\ and\ \citenamefont
  {Busch}}]{loveridge2018symmetry}%
  \BibitemOpen
  \bibfield  {author} {\bibinfo {author} {\bibfnamefont {L.}~\bibnamefont
  {Loveridge}}, \bibinfo {author} {\bibfnamefont {T.}~\bibnamefont
  {Miyadera}},\ and\ \bibinfo {author} {\bibfnamefont {P.}~\bibnamefont
  {Busch}},\ }\bibfield  {title} {\emph {\bibinfo {title} {Symmetry, Reference
  Frames, and Relational Quantities in Quantum Mechanics}},\ }\href@noop {}
  {\bibfield  {journal} {\bibinfo  {journal} {Found. Phys.}\ }\textbf {\bibinfo
  {volume} {48}},\ \bibinfo {pages} {135} (\bibinfo {year} {2018})}\BibitemShut
  {NoStop}%
\bibitem [{\citenamefont {Loveridge}(2020)}]{Loveridge2020WAY}%
  \BibitemOpen
  \bibfield  {author} {\bibinfo {author} {\bibfnamefont {L.}~\bibnamefont
  {Loveridge}},\ }\bibfield  {title} {\emph {\bibinfo {title} {A Relational
  Perspective on the {W}igner-{A}raki-{Y}anase Theorem}},\ }\href
  {https://doi.org/10.1088/1742-6596/1638/1/012009} {\bibfield  {journal}
  {\bibinfo  {journal} {J. Phys. Conf. Ser.}\ }\textbf {\bibinfo {volume}
  {1638}},\ \bibinfo {pages} {012009} (\bibinfo {year} {2020})}\BibitemShut
  {NoStop}%
\bibitem [{\citenamefont {Hayashi}(2006)}]{hayashi2006quantum}%
  \BibitemOpen
  \bibfield  {author} {\bibinfo {author} {\bibfnamefont {M.}~\bibnamefont
  {Hayashi}},\ }\href@noop {} {\emph {\bibinfo {title} {Quantum Information}}}\
  (\bibinfo  {publisher} {Springer},\ \bibinfo {year} {2006})\BibitemShut
  {NoStop}%
\bibitem [{\citenamefont {Bengtsson}\ and\ \citenamefont
  {{\.Z}yczkowski}(2017)}]{bengtsson2017geometry}%
  \BibitemOpen
  \bibfield  {author} {\bibinfo {author} {\bibfnamefont {I.}~\bibnamefont
  {Bengtsson}}\ and\ \bibinfo {author} {\bibfnamefont {K.}~\bibnamefont
  {{\.Z}yczkowski}},\ }\href@noop {} {\emph {\bibinfo {title} {Geometry of
  Quantum States: An Introduction to Quantum Entanglement}}}\ (\bibinfo
  {publisher} {Cambridge university press},\ \bibinfo {year}
  {2017})\BibitemShut {NoStop}%
\bibitem [{\citenamefont {{Korzekwa}}\ \emph {et~al.}(2016)\citenamefont
  {{Korzekwa}}, \citenamefont {{Lostaglio}}, \citenamefont {{Oppenheim}},\ and\
  \citenamefont {{Jennings}}}]{matteo_kamil_bound}%
  \BibitemOpen
  \bibfield  {author} {\bibinfo {author} {\bibfnamefont {K.}~\bibnamefont
  {{Korzekwa}}}, \bibinfo {author} {\bibfnamefont {M.}~\bibnamefont
  {{Lostaglio}}}, \bibinfo {author} {\bibfnamefont {J.}~\bibnamefont
  {{Oppenheim}}},\ and\ \bibinfo {author} {\bibfnamefont {D.}~\bibnamefont
  {{Jennings}}},\ }\bibfield  {title} {\emph {\bibinfo {title} {The Extraction
  of Work from Quantum Coherence}},\ }\href
  {https://doi.org/10.1088/1367-2630/18/2/023045} {\bibfield  {journal}
  {\bibinfo  {journal} {New J. Phys.}\ }\textbf {\bibinfo {volume} {18}},\
  \bibinfo {eid} {023045} (\bibinfo {year} {2016})}\BibitemShut {NoStop}%
\bibitem [{\citenamefont {Hausladen}\ and\ \citenamefont
  {Wootters}(1994)}]{hausladen1994pretty}%
  \BibitemOpen
  \bibfield  {author} {\bibinfo {author} {\bibfnamefont {P.}~\bibnamefont
  {Hausladen}}\ and\ \bibinfo {author} {\bibfnamefont {W.~K.}\ \bibnamefont
  {Wootters}},\ }\bibfield  {title} {\emph {\bibinfo {title} {A ‘pretty
  Good’ Measurement for Distinguishing Quantum States}},\ }\href
  {https://doi.org/10.1080/09500349414552221} {\bibfield  {journal} {\bibinfo
  {journal} {J. Mod. Opt.}\ }\textbf {\bibinfo {volume} {41}},\ \bibinfo
  {pages} {2385} (\bibinfo {year} {1994})}\BibitemShut {NoStop}%
\bibitem [{\citenamefont {Horn}\ and\ \citenamefont
  {Johnson}(2012)}]{horn2012matrix}%
  \BibitemOpen
  \bibfield  {author} {\bibinfo {author} {\bibfnamefont {R.~A.}\ \bibnamefont
  {Horn}}\ and\ \bibinfo {author} {\bibfnamefont {C.~R.}\ \bibnamefont
  {Johnson}},\ }\href@noop {} {\emph {\bibinfo {title} {Matrix Analysis}}}\
  (\bibinfo  {publisher} {Cambridge University Press},\ \bibinfo {year}
  {2012})\BibitemShut {NoStop}%
\end{thebibliography}%

\newpage   
\onecolumngrid
\begin{appendices}

\section{Notation and background details}
\label{appx:background_details}
To any quantum system we have an associated Hilbert space $\H$, and the set of bounded linear operators on this space denoted by $\B(\H)$. Given a group $G$ we denote its representation on $\H$ by $U(g)$ and on $\B(\H)$ as $\U_g$ where $\U_g(X) = U(g) X U(g)^\dagger$ for any $X \in \B(\H)$ and any $g\in G$. 

We define $D(\rho,\sigma)\coloneqq \frac{1}{2} \norm{\rho - \sigma}_1 + \frac{1}{2}\abs{\tr \rho - \tr \sigma}$ to be the generalized trace distance between any two $\rho, \sigma \in \B(\H)$.
We also define $\S_\le(\H)$ to be the set of all normalized and sub-normalized quantum states on $\H$.

A quantum channel $\E: \B(\H_A) \rightarrow \B(\H_B)$, from a quantum system $A$ to a quantum system $B$, is a superoperator that is both trace-preserving and completely positive~\cite{watrous2018theory}. A quantum channel is \emph{covariant} with respect to the group action if we have that
\begin{equation}
    \E(\U_g(\rho)) = \U_g (\E(\rho)),
\end{equation}
for all quantum states $\rho \in \B(\H_A)$ and all $g\in G$. Note here that the group representation of the input and output systems are generally different, and we should strictly write $\U_g^A$ and $\U_g^B$ for each action. However, to simplify notation, we shall use $\U_g$ throughout, as it does not cause ambiguity in practice. The above condition can also be compactly written as $[\U_g, \E] = 0$ for all $g\in G$. 

The reference system $R$ for $G$--covariant transformations from $A$ to $B$ is chosen such that the representation of $G$ on $\H_R$, the Hilbert space of system $R$, is \emph{dual} to its representation on $\H_B$.

We also make use of an irreducible tensor operator (ITO) basis. An ITO consists of a basis of operators ${X^{(\lambda,\alpha)}_j \in \B(\H)}$ that have the property that
\begin{equation}
\U_g \left(X^{(\lambda,\alpha)}_j\right) = \sum_i v^\lambda_{ij}(g) X^{(\lambda,\alpha)}_i,
\end{equation}
where $(v^\lambda_{ji}(g))$ are the matrix components of the $\lambda$--irrep of the group $G$ on $\H$. The irrep $\lambda$ may occur with multiplicities, and so we also denote this as $(\lambda,\alpha)$ where $\alpha$ is a multiplicity label for the irrep. Since $\{X^{(\lambda,\alpha)}_j\}$ form a basis for $\B(\H)$ we may decompose any operator, and in particular any quantum state as
\begin{equation}
\rho = \sum_{\lambda,j} \left(\sum_{\alpha}\tr\left(X^{(\lambda,\alpha)\dagger}_j \rho\right) X^{(\lambda,\alpha)}_j\right) \coloneqq \sum_{\lambda,j} \rho^\lambda_j
\end{equation}
Since each $\rho^\lambda_j$ transforms irreducibly under the group action $\U_g(\cdot)$, this defines a decomposition of $\rho$ into modes of asymmetry. If $\E$ is a covariant channel, then it does not mix modes, and so
\begin{equation}
\E(\rho^\lambda_j) = \E(\rho)^\lambda_j,
\end{equation}
for all $\lambda,j$. See~\cite{marvian2012symmetry,marvian2014modes} for more details. 

\section{Properties of min-entropies}

In this section we review some useful properties of the conditional min-entropies $H_{\rm min} (R|A)_\Omega$ and properties that relate to the case of $\Omega_{RA}$ being a $G$-twirled bipartite quantum state. 

It proves useful to define the functional $\Phi(M_{RA}): M_{RA} \mapsto 2^{-H_{\mathrm{min}}(R|A)_M} $ defined on bipartite Hermitian operators $M_{RA}$:

\begin{definition} Let $M_{RA}$ be a Hermitian operator on $\H_{RA}$. Then we define the argument of the conditional min-entropy of $M_{RA}$, $\Phi(M_{RA})$, via
\begin{align}
 \Phi(M_{RA})   \coloneqq \inf_{X_A \ge 0} \{ \tr [X_A] \ : \ \mathbbm{1}_R \otimes X_A - M_{RA} \ge 0 \}.
\end{align}
\end{definition}
We note the following known properties of the functions $\Phi(\cdot)$, proofs of which can be found in \cite{tomamichel2012framework} or are obvious from the definition.
\begin{enumerate}[label=\normalfont \textbf{(P\arabic*)}]
\item \label{property:scalar_multiplication} \textit{(Scalar multiplication).} $\Phi(\lambda M_{RA})= \lambda \Phi(M_{RA})$ for any $\lambda \ge 0$.
\item \label{property:convexity} \textit{(Convexity).} $\Phi(pM_{RA} +(1-p)N_{RA}) \le p \Phi(M_{RA}) + (1-p) \Phi(N_{RA})$ for any $p \in [0,1]$.
\item \label{property:unitary_invariance} \textit{(Invariance under local isometries).} Let $\mathcal{U}_R \coloneqq U_R \otimes \id_A (\cdot) U_R^\dagger \otimes \id_A$ and $ \mathcal{V}_A \coloneqq \id_R \otimes V_A (\cdot)\id_R \otimes V_A^\dagger$ be isometries on subsystems $R$ and $A$ respectively. Then $\Phi(\U_R \circ \V_A (M_{RA})) = \Phi(M_{RA})$.
\item \label{property:F_data_processing} \textit{(Local data processing inequality).} Let $\E_{R}: \B(\H_R) \rightarrow\B( \H_{R'})$ be a unital CPTP map, $\N_A: \B(\H_A) \rightarrow \B(\H_{B})$ be a CPTP map, and $\I_X(\cdot) \coloneqq \id_X (\cdot) \id_X$ denote the identity channel on system $X$. Then $\Phi\left( (\E_R\otimes \I_A) (M_{RA})\right) \le \Phi(M_{RA})$ and $\Phi((\I_R\otimes \N_A) (M_{RA})) \le \Phi(M_{RA})$.
\end{enumerate}
We also introduce the following simplifying notation for the bipartite $\G$--twirled states: 
\begin{align}
\Phi_\eta(\tau) \coloneqq \Phi(\G(\eta_R \otimes \tau_A)) ,
\label{eq:Phi(tau)}
\end{align}
and also $H_\eta(\tau):= -\log \Phi_\eta(\tau)$.

\subsection{Invariance under local isometries that commute with $\G$}
\label{appx:invariance_unitaries_com_acom}

Here we prove the following lemma from the main text, which specializes property \ref{property:F_data_processing} to the particular form of the conditional min-entropies appearing in \thmref{thrm:gour}, which are instead equivalent up to local isometries that jointly commute with the $G$-twirl on the global system.

\covisometryinvariance*

\begin{proof}
    By a straightforward appeal to Property \ref{property:unitary_invariance}, we have
    \begin{align}
        \Phi(\G[\U_R(\eta) \otimes \V_A(\tau)]) = \Phi(\U_R \otimes \V_A \circ \G[\eta\otimes\tau]) = \Phi(\G[\eta\otimes\tau]).
    \end{align}
    Since $H_\eta(\rho) \coloneqq - \log \Phi(\G[\eta \otimes \rho])$, this implies
    \begin{align}
        H_{\U(\eta)}(\V(\rho)) = H_\eta(\rho),
    \end{align}
    as claimed.
\end{proof}

\subsection{Symmetric input states}
\label{appx:Hmin_symmetric_states}

\begin{restatable}[]{lemma}{phidephasedstates}
\label{lemma:phi_dephased_states}
For any input state $\rho$ on system $A$ and reference state $\eta$ on system $R$, the following identity holds
\begin{align}
\label{eq:phi_dephased_states}
    H_{\G(\eta)}(\rho) &= H_{\eta}(\G(\rho)) = H_{\G(\eta)}(\G(\rho)) \\  &= - \log \norm{\G(\eta)}_\infty.
\end{align}
\end{restatable}

\begin{proof}
The first two equalities in Eq.~(\ref{eq:phi_dephased_states}) straightforwardly follow from the fact that $\G(\G(\eta) \otimes \rho)  = \G(\eta) \otimes \G(\rho)=\G(\eta \otimes \G(\rho))=\G(\G(\eta) \otimes \G(\rho))$. To show the final equality, we first examine
\begin{align}
	 \Phi(\G(\eta) \otimes \G(\rho))  &= \inf_{X \ge 0} \{ \tr [X] \ : \ \mathbbm{1}  \otimes X - \G(\eta) \otimes \G(\rho) \ge 0 \}
\end{align}
Since $\G(\eta)$ is Hermitian, it can be diagonalised as $\G(\eta) \coloneqq \sum_i \lambda_i \ketbra{i}$ for some basis $\{\ket{i}\}$ of the reference system. Working in this basis, we obtain
\begin{align}
	\id \otimes X - \G(\eta) \otimes \G(\rho) = \left(\sum_i \ketbra{i} \otimes X \right) - \left(\sum_i \lambda_i\ketbra{i} \otimes \G(\rho)\right) = \sum_i \ketbra{i} \otimes (X - \lambda_i \G(\rho))	 
\end{align}
Therefore, $\id \otimes X - \G(\eta) \otimes \G(\rho) \ge 0$ if and only if $X - \lambda_i\G(\rho) \ge 0$ for all $i$, which in turn is true if and only if $X - \lambda_{\mathrm{ max}}\G(\rho) \ge 0$, where $\lambda_{\mathrm{ max}}$ denotes the largest eigenvalue of $\G(\rho)$. We can lower-bound the $\tr[X]$ needed to achieve this by
\begin{align}
	\tr[X] \ge \lambda_{\mathrm{ max}} \tr[\G(\rho)] = \lambda_{\mathrm{ max}},
\end{align} 
and this minimum can be attained simply by choosing $X \coloneqq \lambda_{\mathrm{ max}}\G(\rho)$. Therefore,
\begin{align}
	\Phi(\G(\eta) \otimes \G(\rho)) = \lambda_{\mathrm{ max}} = \norm{\G(\rho_A)}_\infty \Rightarrow H_{\G(\eta)}(\G(\rho)) = - \log \norm{\G(\rho)}_\infty,
\end{align} 
which completes the proof.\end{proof}

\subsection{Dual Formulation}

From \eqref{eq:Phi(tau)}, we see that $\Phi_\eta(\tau)$ is defined via a semidefinite programme (SDP). In this subsection, we prove a lemma stating what the dual form of this SDP is, which is convenient for proving several results of this paper, including the depolarization conditions for $G$--covariant channels. 

\begin{lemma}
	Let $\mathcal{O}_{\rm cov}$ be the set of $G$--covariant channels from input system $A$ to output system $B$. Let $\eta$ be state of the reference system $R$ for this transformation, and $\tau$ be a state of input system $A$. Then the \emph{dual} formulation of $\Phi_\eta(\tau)$ is 
	\begin{align}
	\Phi_\eta(\tau) = \max_{\mathcal{E} \in \mathcal{O}_{\text{\rm cov}}} \tr(\eta^T \E(\tau)), 
	\end{align}
	\label{lemma:Phi_and_dual}
\end{lemma}

\begin{proof}
	From Lemma 3 in \cite{gour2018quantum}, we see that
	\begin{align}
		\Phi_\eta(\tau) = \max_{\E \in \mathcal{O}_{\rm cov}} \bra{\phi_+} \id \otimes \E[\G(\eta \otimes \tau)] \ket{\phi_+},  
	\end{align} 
	where $\ket{\phi_+} \coloneqq \sum_i \ket{ii}$ for local computational bases $\{\ket{i}\}$ for the reference $R$ and output system $B$.	
	We first note that $\ket{\phi_+} = \ket{vec(\id)}$, which means~\cite{watrous2018theory} that
	\begin{align}
		U^*(g) \otimes U(g)\ket{\phi_+} = U^*(g) \otimes U(g) \ket{vec(\id)} = \ket{vec\left(U^*(g)\id U^T(g)\right)} = \ket{vec(\id)} = \ket{\phi_+},
	\end{align}
	where $U(g)$ is the representation of the group element $g$ on $\H_B$, the Hilbert space of $B$. We therefore see that ${\G(\ketbra{\phi_+}) = \ketbra{\phi_+}}$, which allows us to derive:	
	\begin{align}
		\Phi_\eta(\tau) &= \max_{\E \in \mathcal{O}_{\rm cov}} \tr(\ketbra{\phi_+} \id \otimes \E[\G(\eta\otimes\tau)])\\
		&= \max_{\E \in \mathcal{O}_{\rm cov}} \tr(\ketbra{\phi_+} \G(\eta\otimes \E[\tau]))\\
		&= \max_{\E \in \mathcal{O}_{\rm cov}} \tr\left(\ketbra{\phi_+} \left[\int dg U^*(g) \otimes U(g) (\eta \otimes \E[\tau])U^T(g) \otimes U^\dagger(g) \right] \right)\\
		&= \max_{\E \in \mathcal{O}_{\rm cov}} \tr\left(\left[\int dg U^T(g) \otimes U^\dagger(g) \ketbra{\phi_+} U^*(g) \otimes U(g)\right] \eta \otimes \E[\tau]\right)\\
		&= \max_{\E \in \mathcal{O}_{\rm cov}} \tr\left(\left[\int dg\ \U^*_{g^{-1}} \otimes \U_{g^{-1}} (\ketbra{\phi_+})\right] \eta \otimes \E[\tau]\right)\\
		&= \max_{\E \in \mathcal{O}_{\rm cov}} \tr(\G(\ketbra{\phi_+}) \eta \otimes \E[\tau])\\
		&= \max_{\E \in \mathcal{O}_{\rm cov}} \tr(\ketbra{\phi_+} \eta \otimes \E[\tau])\\
		&= \max_{\E \in \mathcal{O}_{\rm cov}} \sum_{ij} \bra{j} \eta \ket{i} \bra{j}\E[\tau]\ket{i}\\
		&= \max_{\E \in \mathcal{O}_{\rm cov}} \sum_{ij} \bra{i} \eta^T \ket{j}\bra{j}\E[\tau]\ket{i}\\
		&= \max_{\E \in \mathcal{O}_{\rm cov}} \tr(\eta^T\E[\tau]),
	\end{align}
as claimed.
\end{proof}
\subsection{Truncation of Output System}
The following two Lemmas detail when one can truncate the Hilbert space of the output system without affecting the possibility of interconversion to a particular output state $\sigma$. This is of use in our analysis of state interconversion with partial depolarization.

\begin{lemma}
	\label{lemma:where_to_truncate}
	Let $\rho$ be a state of the input system $A$, associated to the Hilbert space $\H_A$. Let $\sigma$ be a state of the output system $B$, associated to the Hilbert space $\H_B$.
	
	\item Let $\H_S$ be any subspace of $\H_B$ with the following two properties: 
	\begin{enumerate}
		\item $\H_S$ carries its own representation of $G$, i.e. $\H_B$ can be decomposed into $\H_B = \H_S \oplus \H_A$ such that $U_B(g) = U_S(g) \oplus U_A(g)$.
		\item The support of $\sigma$ is contained entirely within $\H_S$, i.e. letting $\Pi_S$ be the projector onto $\H_S$, $\Pi_S \sigma \Pi_S = \sigma$.
	\end{enumerate}
	
	Let $\H_S$ be the Hilbert space appropriate to a new output system $S$ truncated from $B$. Then there exists a $G$--covariant operation from $A$ to $B $ that takes $\rho$ to $\sigma$ if and only if there exists a $G$--covariant operation from $A$ to $S$ that takes $\rho$ to $\sigma$. 
\end{lemma}

\begin{proof}	
	Let us first assume that there exists a $G$--covariant operation $\E_{\rm cov}: \B(\H_A) \rightarrow \B(\H_B)$ such that ${\E_{\rm cov}(\rho) = \sigma}$. We then observe that 
	\begin{align}
	\Pi_S U_B(g) = \left(\id_S \oplus 0_A\right) \left(U_S(g) \oplus U_A(g)\right) = U_S(g) \oplus 0_A =  (U_S(g) \id_S) \oplus 0_A =U_S(g) \Pi_S
	\end{align}
	
	Therefore, $\Pi_S(\cdot)\Pi_S$ is a covariant map from $\B(\H_B)$ to $\B(\H_S)$. As a result, $\Pi_S [\E_{\rm cov}(\cdot)] \Pi_S $ is a covariant operation from $\B(\H_A)$ to $\B(\H_S)$ such that
	\begin{align}
	\Pi_S \E_{\rm cov}(\rho) \Pi_S = \Pi_S\sigma\Pi_S = \sigma.
	\end{align}
	
	Conversely, let us now assume that that there exists a covariant transformation ${\F_{\rm cov}: \B(\H_A) \rightarrow \B(\H_S)}$ such that $\F_{\rm cov}(\rho) = \sigma$. We then extend $\H_S$ into the bigger Hilbert space $\H_B = \H_S \oplus \H_A$ such that $\H_S$ still forms its own representation of $G$, i.e. $U_B(g) = U_S(g) \oplus U_A(g)$. Then $\F_{\rm cov}$ can be reinterpreted as a covariant channel from $\B(\H_A)$ to $\B(\H_B)$. 
	
	We therefore conclude that $G$--covariant interconversion from $\rho$ to $\sigma$ is unaffected by treating $\sigma$ as a state of $B$ or as a state of $S$.
\end{proof}

\begin{lemma}
	Given any particular output state $\sigma$, it is always possible to truncate the Hilbert space of the output system, $\H_B$, to the support of $\G(\sigma)$ without affecting the possibility of $G$--covariant interconversion. 
	
	\label{lemma:H_S_construction}
\end{lemma}

\begin{proof}
	The representation of $G$ on $\H_B$ splits up in the following way \cite{bartlett2007reference}:
	\begin{align}
	\H_B = \bigoplus_q \H_q.
	\end{align}
	The $\H_q$ are known as the \emph{charge sectors} of $\H_B$, and they each carry an \emph{inequivalent} representation of $G$. Each $\H_q$ can be further decomposed into a tensor product
	\begin{align}
	\H_q = \M_q \otimes \N_q.
	\end{align} 
	The $\M_q$ carry inequivalent \emph{irreps} of $G$, while the $\N_q$ carry \emph{trivial} representations of $G$. This means every element $g$ is represented on $\H_q$ in the form $U_{\M_q}(g) \otimes \id_{\N_q}$. As a result, given any pure state $\ket{\psi_q}$ in $\N_q$, $\M_q \otimes \Span(\ket{\psi_q})$ is an irrep of $G$. Projectors onto irreps of $G$ thus take the form $\id_{\M_q} \otimes \ketbra{\psi_q}$. 
	
	We note the following properties about the projector $\id_{\M_q} \otimes \ketbra{\psi_q}$. Because $\M_q \otimes \Span\left(\ket{\psi_q}\right)$ is a subspace of $\H_q$, 
	\begin{align}
	\left[\Pi_q, \id_{\M_q} \otimes \ketbra{\psi_q}\right] = 0
	\label{eq:irrep_general_group_prop_one}
	\end{align} 
	For the same reason, $\Pi_q$ is identity on $\M_q \otimes \Span\left(\ket{\psi_q}\right)$, which means
	\begin{align}
	\id_{M_q} \otimes \ketbra{\psi_q} = \Pi_q \left(\id_{\M_q} \otimes \ketbra{\psi_q}\right).
	\label{eq:irrep_general_group_prop_two}
	\end{align} 
	
	A subspace $\H$ of $\H_B$ lies inside the kernel of $\sigma$ if and only if $\tr(\Pi\sigma) = 0$, where $\Pi$ is the projector onto $\H$. Therefore, the irrep $\M_q \otimes \Span\left(\ket{\psi_q}\right)$ lies in the kernel of $\sigma$ if and only if 
	\begin{align}
	\tr(\left(\id_{\M_q} \otimes \ketbra{\psi_q}\right) \sigma)
	&= \tr(\Pi_q \left(\id_{\M_q} \otimes \ketbra{\psi_q}\right) \sigma)\\
	&= \tr(\Pi_q \Pi_q \left(\id_{\M_q} \otimes \ketbra{\psi_q}\right) \sigma)\\
	&= \tr(\Pi_q \left(\id_{\M_q} \otimes \ketbra{\psi_q}\right) \sigma \Pi_q)\\
	&= \tr(\left(\id_{\M_q} \otimes \ketbra{\psi_q}\right)\Pi_q \sigma \Pi_q)\\
	&= \bra{\psi_q}\tr_{\M_q}(\Pi_q\sigma\Pi_q)\ket{\psi_q} = 0	
	\end{align}
	where in the first equality we made use of \eqref{eq:irrep_general_group_prop_two}, and in the fourth equality we made use of \eqref{eq:irrep_general_group_prop_one}. This short calculation means $\M_q \otimes \Span(\ket{\psi_q})$ lies inside the kernel of $\sigma$ if and only if $\ket{\psi_q}$ lies inside the kernel of $\tr_{\M_q}(\Pi_q\sigma\Pi_q)$.
	
	Let $\left\{\ket{\psi_{q,i}}\right\}$ be an orthonormal basis for $\N_q$ in which $\tr_{\M_q}(\Pi_q\sigma\Pi_q)$ is \emph{diagonalised}. One possible irrep decomposition for $\H_B$ is then
	\begin{align}
	\H_B = \bigoplus_{q,i} \M_q \otimes \Span\left(\ket{\psi_{q,i}}\right).
	\label{eq:blahahah}
	\end{align}
	An irrep in this decomposition lies inside the kernel of $\sigma$ if and only if $\ket{\psi_{q,i}}$ is a basis element for the kernel of $\tr_{\M_q}(\Pi_q\sigma\Pi_q)$. This means 
	\begin{align}
		\H_S^\perp \coloneqq \bigoplus_{q,\ \ket{\psi_{q,i}} \in \ker[\tr_{\M_q}(\Pi_q\sigma\Pi_q)]}\M_q \otimes \Span\left(\ket{\psi_{q,i}}\right) = \bigoplus_q \M_q \otimes \ker(\tr_{\M_q}(\Pi_q\sigma\Pi_q))
	\end{align}
	must lie inside the kernel of $\sigma$ on $\H_B$. Conversely, the support of $\sigma$ must lie inside the subspace of $\H_B$ that is orthogonal to $\H_S^\perp$, i.e.
	\begin{align}
	\H_S = \bigoplus_q \M_q \otimes \supp[\tr_{\M_q}(\Pi_q\sigma\Pi_q)].
	\label{eq:H_S_construction}
	\end{align}
	As we see from the above equation, $\H_S$ is also a direct sum over irreps of $G$ and so carries its own representation of $G$. Thus by Lemma~\ref{lemma:where_to_truncate}, the possibility of interconversion is unaffected if we truncate $\H_B$ to $\H_S$.
	
	The action of the $G$-twirl is given by \cite{bartlett2007reference}:
	\begin{align}
	\G = \sum_q (\D_{\M_q} \otimes \I_{\N_q}) \circ \P_q,
	\end{align} 
	where $\P_q \coloneqq \Pi_q(\cdot)\Pi_q$ is the projector onto the charge sector $\H_q$, $\D_{\M_q}$ is the completely depolarising channel on $\M_q$ and $I_{\N_q}$ is the identity channel on $\N_q$. Therefore,
	\begin{align}
	\G(\sigma) = \sum_q \frac{\id}{d_{\M_q}} \otimes \tr_{\M_q}(\Pi_q\sigma\Pi_q),
	\end{align} 
	where $d_{\M_q}$ is the dimension of $\M_q$. Looking back at Equation~(\ref{eq:H_S_construction}), we see that $\H_S = \supp[\G(\sigma)]$. It is therefore always possible to truncate the output Hilbert space to the support of $\G(\sigma)$ without affecting possibility of interconversion.
\end{proof}

\section{Redundancies in the entropic relations}

Here we consider a few basic redundencies in the infinite set of conditions appearing in \thmref{thrm:gour}.
We first note that $H_\eta(\rho) \le H_\eta(\sigma)$ and $\Phi_\eta(\rho) \ge \Phi_\eta(\sigma)$ are equivalent conditions, since we have that $\Phi(X)>0$ for any $X$ that is positive-semidefinite with at least one non-zero eigenvalue  \cite{tomamichel2012framework}, and the fact that $-\log(x)$ is monotonic decreasing in $x$ for $x>0$.

\subsection{Unitaries on the reference}
\label{appx:unitary_equiv_reference}

\lemref{lemma:cov_isometry_invariance} immediately gives rise to the following corollary. 
\begin{corollary}
If $\eta_1 = \V(\eta_0)$ for any pair of quantum states $\eta_0,\eta_1$ on $\H_R$ and for any unitary $\V: \B(\H_R) \rightarrow \B(\H_R)$ such that $[\V \otimes \id ,\G]=0$, then $\Delta H_{\eta_0}\ge 0$ if and only if $\Delta H_{\eta_1}\ge 0$.
\end{corollary}
As an example of this redundency, we can consider the group $G=U(1)$ of time-translations generated by the Hamiltonian $H$, $U(t) = e^{-iHt}$. Here we have that any reference state $\eta$ drawn from the set $\{  U(t) \eta U(t)^\dagger  :  \forall t \in [0,2\pi) \}$ will provide an equivalent constraint.

\subsection{Modes of asymmetry}
\label{appx:modes_rho_modes_eta_redundency}

Given an ITO basis $\left\{X^{(\lambda,\alpha)}_j\right\}$, which in the following we always take to be orthonormal such that 
\begin{align}
\left\langle X^{(\lambda,\alpha)}_j, X^{(\mu,\beta)}_k \right\rangle = \delta_{\lambda,\mu} \delta_{\alpha, \beta} \delta_{j,k},
\end{align}
where $\<A,B\> = \tr[A^\dagger B]$ is the Hilbert-Schmidt inner product on $\B(\H)$, we denote by $A^\lambda_j$ the $(\lambda,j)$ mode of the operator $A$
\begin{equation}
A^\lambda_j \coloneqq \sum_\alpha \left\langle X^{(\lambda,\alpha)}_j, A \right\rangle X^{(\lambda,\alpha)}_j.
\end{equation}
The following gives some basic properties for handling inner products involving modes of asymmetry.
\begin{lemma} 
\label{lemma:mode_decomp_trace_product}
Let $A$ and $B$ be any two linear operators on $\H$. Then we have that
\begin{align}
 \langle A^\lambda_j,B^\mu_k \rangle = \delta_{\lambda, \mu} \delta_{j,k} \langle A^\lambda_j,B^\mu_k \rangle
\end{align}
from which it follows that:
\begin{align}
 \langle A, B \rangle = \sum_{\lambda, j} \langle A^\lambda_j , B^\lambda_j \rangle
\end{align}
\end{lemma}
\begin{proof}
Writing out the mode decompositions of $A$ and $B$ in the ITO basis explicitly in the trace product, from the orthonormality of the basis operators $\{ X^{(\lambda,\alpha)}_j \}$ we obtain
\begin{align}
 \langle A^\lambda_j,B^\mu_k \rangle &= \left\langle \sum_\alpha \langle X^{(\lambda,\alpha)}_j, A \rangle X^{(\lambda,\alpha)}_j , \sum_\beta \langle X^{(\mu,\beta)}_k, B \rangle  X^{(\mu,\beta)}_k \right\rangle \\
&= \sum_{\alpha,\beta}  \langle X^{(\lambda,\alpha)}_j, A \rangle \langle X^{(\mu,\beta)}_k, B \rangle \langle  X^{(\lambda,\alpha)}_j , X^{(\mu,\beta)}_k \rangle \\
&= \sum_{\alpha,\beta}  \langle X^{(\lambda,\alpha)}_j, A \rangle \langle X^{(\mu,\beta)}_k, B \rangle \delta_{\lambda,\mu} \delta_{j,k} \langle  X^{(\lambda,\alpha)}_j , X^{(\mu,\beta)}_k \rangle \\
&= \delta_{\lambda,\mu} \delta_{j,k}  \left\langle \sum_\alpha \langle X^{(\lambda,\alpha)}_j, A \rangle X^{(\lambda,\alpha)}_j , \sum_\beta \langle X^{(\mu,\beta)}_k, B \rangle  X^{(\mu,\beta)}_k \right\rangle \\
&= \delta_{\lambda,\mu} \delta_{j,k}  \langle A^\lambda_j,B^\mu_k \rangle,
\end{align}
as required.
\end{proof}

An immediate consequence of such a mode decomposition is that if $ \text{ modes}(\sigma)\subseteq \text{modes}( \rho)$, then it suffices to range only over reference frame states $\eta$ such that modes($\eta$) = modes($\rho$) in \thmref{thrm:gour}. 
The reasoning is as follows. Let $\eta$ have an irrep mode $\mu$ that does not occur in $\rho$. By hermiticity, it also has the irrep mode $\mu^*$. As seen in the following lemma, when computing $\G[\eta \otimes \rho]$, the only mode terms that survive this G-twirl are of the form $\sum_i \eta^{\mu^*}_i \otimes \rho^{\mu}_i$, which is the unique way to form a singlet from a given irrep:
\begin{lemma}
	\begin{align}
		\G(\eta \otimes \rho) = \sum_{i,\mu} \eta^{\mu^*}_i \otimes \rho^\mu_i.
	\end{align}
	\label{lemma:G_twirl_modal_match}
\end{lemma}
\begin{proof}
	Equation 4.3 of \cite{marvian2014modes} states that
	\begin{align}
		\int dg\ v^\lambda_{i'i}(g) v^\mu_{j'j} = \frac{1}{d_\mu} \delta_{i',i} \delta_{j',j} \delta_{\lambda^*,\mu}.
		\label{eq:marvian_mode_ortho}
	\end{align}
	
	Making use of this result, we can then demonstrate that
	\begin{align}
		\G(\eta \otimes \rho) &= \int dg\ \sum_{\lambda,i,\mu,j} \mathcal{U}_g(\eta^\lambda_i)\otimes\mathcal{U}_g(\rho^\mu_j)\\
		&= \int dg\ \sum_{\lambda,i,\mu,j} \sum_{i',j'} v^{\lambda}_{i'i}(g) v^{\mu}_{j'j}\eta^\lambda_{i'}\otimes\rho^\mu_{j'}\\
		&= \sum_{\lambda,i,\mu,j} \sum_{i'j'} \frac{1}{d_\mu} \delta_{i',j'}\delta_{i,j}\delta_{\lambda^*,\mu} \eta^\lambda_{i'}\otimes\rho^\mu_{j'}\\
		&= \sum_{\mu,i'} \left(\sum_j \frac{1}{d_\mu}\right) \eta^{\mu^*}_{i'}\otimes\rho^\mu_{i'}\\
		&= \sum_{\mu,i} \eta^{\mu^*}_j\otimes\rho^\mu_j
	\end{align} 
\end{proof}
We conclude from this lemma that if $\rho$ does not contain a $\mu$ mode, then $\eta^{\mu^*}_i$ does not contribute to the state $\G[\eta\otimes\rho]$.

If we range over all $\eta$ contained within a small surface around $\mathbbm{1}/d$, then we obtain a complete set of conditions. We know that $\eta^{\mu^*}_j$ does not contribute to $\Omega_{RA}$. Now if the region is chosen sufficiently small, we claim that $\eta' =\eta - \left[\sum_j \left(\eta^\mu_j + \eta^{\mu^*}_j \right) \right] $ is still a valid state for $\mu \ne 0$, but with the $\mu$ mode removed. To see this, firstly note that, by orthonormality, the term in the brackets is traceless and so the net result still has trace one. Secondly, by choosing the surface appropriately, the eigenvalues of $\eta$ can be chosen arbitrarily close to the uniform distribution, and those of the term in bracket made arbitrarily small. Therefore the eigenvalues of the resultant operator $\eta'$ are all non-negative. Therefore we have a reference frame state $\eta'$ with the $\mu$ mode entirely removed, and gives the state joint state $\G[\eta\otimes\rho]$ as did $\eta$. This implies it suffices to range over reference frame states $\eta$ with modes the same as $\rho$.

\section{A sufficient surface of reference frames}
\label{appx:partially-depol}

\subsection{Depolarizing the reference state }

Let us define the partially depolarizing channel for some fixed probability $p$:
\begin{align}
\Lambda_p [\rho] \coloneqq p \rho + (1-p) \frac{\mathbbm{1}}{d}.
\end{align}
In general, from \ref{property:convexity} we know that the functional $\Phi_\eta(\rho)$ is convex in the reference system, i.e., $\eta = p \eta_0 + (1-p) \eta_1$ implies 
\begin{align}
    \Phi_\eta (\rho) \le p \Phi_{\eta_0}(\rho) + (1-p) \Phi_{\eta_1}.
\end{align}
The following lemma shows that the functional $\Phi_\eta(\rho)$ behaves linearly when we take convex combinations of the reference state with the maximally mixed state.

\begin{restatable}[]{lemma}{partiallydepolone}
\label{lemma:partially-depol1}
Let all states and systems be defined as in Theorem \ref{thrm:gour}. For any reference state $\eta_R$ and input state $\tau_A$, we have
\begin{align}
\Phi_{\Lambda_p(\eta)}(\tau) &=  p \Phi_{\eta}(\tau)+  (1-p) \Phi_{\id/d}(\tau) \notag \\
&= p \Phi_{\eta}(\tau) +\frac{1-p}{d}.
\end{align}
\end{restatable}

\begin{proof}
Since $\mathbbm{1}$ is symmetric for any group $G$, we have $\G[ \Lambda_p(\eta_R)\otimes \tau_A] = p \G[\eta_R \otimes \tau_A ] + \frac{1-p}{d} \mathbbm{1} \otimes \G[\tau_A]$. Substituting this into Eq. (\ref{eq:Phi(tau)}) and rearranging terms gives
\begin{align}
\Phi_{\Lambda_p[\eta]}\left(\tau \right) 
= \inf_{X_A \ge 0} \left\{ \tr[X_A] \ : \ \mathbbm{1}_R \otimes\left( X_A - \frac{1-p}{d} \G[\tau_A] \right) - p \G\left[ \eta_R \otimes \tau_A \right] \ge 0 \right\}. \label{eq:F-calc} 
\end{align}
For any positive semidefinite operators $T_A$ and $Z_{RA}$, we have
\begin{equation}
\{ X_A \ : \ X_A \ge 0, \ \mathbbm{1}_R \otimes [X_A-T_A] - Z_{RA} \ge 0 \} = \{ X_A \ : \ X_A - Y_A \ge 0, \ \mathbbm{1}_R \otimes [X_A-YTA] - Z_{RA} \ge 0 \}.
\end{equation}
Therefore we can rewrite the feasible set over which we perform the optimization in Eq. (\ref{eq:F-calc}) as follows
\begin{align}
\Phi_{\Lambda_p[\eta]}\left(\tau \right) 
&= \inf_{X_A - \frac{1-p}{d} \G[\tau_A]\ge 0} \left\{ \tr[X_A] \ : \ \mathbbm{1}_R \otimes\left( X_A - \frac{1-p}{d} \G[\tau_A] \right) - p \G\left[ \eta_R \otimes \tau_A \right] \ge 0 \right\} \\
&= \inf_{Y_A \ge 0} \left\{ \tr \left[Y_A + \frac{1-p}{d} \G[\tau_A]\right] \ : \ \mathbbm{1}_R \otimes Y_A - p \G\left[ \eta_R \otimes \tau_A \right] \ge 0 \right\} ,
\end{align}
where we have defined $Y_A \coloneqq X_A - \frac{1-p}{d} \G(\tau_A)$. Now since $\tr [\G(\tau_A)]=1$, we can take the constant term out of the infimum
\begin{align}
\Phi_{\Lambda_p[\eta]}\left(\tau \right) &= \inf_{Y_A \ge 0} \left\{ \tr[Y_A] \ : \ \mathbbm{1}_R \otimes Y_A - p \G\left[ \eta_R \otimes \tau_A \right] \ge 0 \right\} + \frac{1-p}{d}.
\end{align}
Finally we make use of property \ref{property:scalar_multiplication} to arrive at
\begin{align}
\Phi_{\Lambda_p[\eta]}\left(\tau \right) &= p \Phi_\eta(\tau)+ \frac{1-p}{d}.
\label{sdfsdffff}
\end{align}
which concludes the proof.
\end{proof}

An automatic consequence of \lemref{lemma:partially-depol1} is that taking a convex mixture of any reference state with the maximally mixed state will not change the entropic relation in Theorem \ref{thrm:gour}, which we state in the following lemma:
\begin{lemma}
Let all states and systems be defined as in Theorem \ref{thrm:gour} and let us further define the partially depolarizing quantum channel $\Lambda_p [\rho] \coloneqq p \rho + (1-p) \frac{\mathbbm{1}}{d}$, where $p$ is a probability and $d\coloneqq\mathrm{dim}(\H_R)$.  Then the following two statements are equivalent for any $\eta$:
\begin{enumerate}
    \item $\Delta H_\eta \ge 0$.
    \item $\Delta H_{\Lambda_p[\eta]} \ge 0$ for any $p\in (0,1]$.
\end{enumerate}
\label{lemma:partially-depol}
\end{lemma}

\begin{proof}
Defining $\Delta \Phi_\eta \coloneqq \Phi_\eta(\rho) - \Phi_\eta(\sigma)$, \lemref{lemma:partially-depol1} implies that $\Delta \Phi_{\Lambda_p[\eta]} =  p \Delta \Phi_\eta$ for any $p$ probability. Therefore, $\Delta \Phi_\eta \ge 0 $ if and only if $\Delta  \Phi_{\Lambda_p[\eta]}\ge 0$, for any $p\in [0,1]$. Since $\Phi_\eta$ and $H_\eta$ are monotonically related, this then gives the statement of the lemma. \end{proof}

\subsection{Proof of \thmref{thm:ball_sufficient_refs}}
\label{appx:proof_sufficient_surface}
We now present a proof of \thmref{thm:ball_sufficient_refs}, which we restate here for clarity:
\sufficientsurface*
\begin{proof}
If the transformation is possible under a $G$--covariant channel then $\Delta H_\eta \ge 0$ for all states $\eta$, and hence in particular for all $\eta$ restricted to $\partial \D$. Conversely, suppose $\Delta H_\eta \ge 0 $ for all $\eta \in \partial \D$. Let $\eta'$ be an arbitrary quantum state of $R$ that is not the maximally mixed state, and consider the one-parameter family of states $\eta'(p) := \Lambda_p(\eta')$ for $p\in [0,1]$. This defines a continuous line of states connecting $\eta'$ to the maximally mixed state $\id/d$. Since $\partial \D$ encloses the maximally mixed state the set $\{\eta'(p) : 0 \le p \le 1\}$ must either intersect $\partial \D$ for some value $p_\star$ with $0< p_\star \le 1$ or lie entirely within the interior of $\D$. If the set is entirely inside $\D$ then we can find a quantum state $\eta'' \in \partial \D$ such that $\Delta_q [\eta''] = \eta'$ for some $q \in (0,1)$. However from lemma~\ref{lemma:partially-depol} we have that $\eta'$ and $\eta'(p_\star)$ (or $\eta'$ and $\eta''$ for the second case) give equivalent entropic constraints. Since $\eta'$ was arbitrary it therefore suffices to restrict to states lying solely on the surface $\partial \D$, which completes the proof.
 
\end{proof}

\section{Smoothed asymmetry theory}
\subsection{Continuity of entropic relations under variations of the reference state}
\label{appx:Hmin_relations_continuity}

In this section, we consider the following definition of an \textit{$\varepsilon$-ball} of operators on $\H$ around some $\rho \in \S_\le(\H)$
\begin{align}
B^\epsilon(\rho) \coloneqq \{ \tilde{\rho} \in S_\le(\H) : D(\tilde{\rho},\rho) \le \epsilon \},
\end{align}
but we note that all the results derived in this section also apply if we use the purified distance $P(\cdot,\cdot)$ as our distance measure instead, due to the property $D(\rho,\sigma) \le P(\rho, \sigma)$ for all $\rho$, $\sigma$ \cite{tomamichel2010duality}.

We also make use of the following theorem, which was proven in Ref.~\cite{tomamichel2012framework}.
\begin{theorem}
\label{thm:continuity} (Continuity of min-entropy). Let $\rho, \sigma \in \S_\le(\H_{RA})$. Then
\begin{equation}
    |H_{\mathrm{min}}(R|A)_\rho - H_{\mathrm{min}}(R|A)_\sigma | \le  \frac{d_R \min\{d_R, d_A \}}{\ln 2 \, \min \{\tr \rho , \tr \sigma \}} D(\rho,\sigma),
\end{equation}
where $D(\rho,\sigma)\coloneqq \frac{1}{2} \norm{\rho - \sigma}_1 + \frac{1}{2}\abs{\tr \rho - \tr \sigma}$ is the generalized trace distance.
\end{theorem}

\begin{lemma}
\label{lemma:variation_reference_to_omega}
Consider the following bipartite quantum states on $\H_{RA}$:
\begin{align}
    \Omega^{RA} \coloneqq \eta_R \otimes \rho_A, \quad \tilde{\Omega}^{RA} \coloneqq \tilde{\eta}_R \otimes \rho_A.
\end{align}
If $\tilde{\eta} \in \B_\varepsilon(\eta)$, then
\begin{align}
  D(\G[\Omega], \G[\tilde{\Omega}])    &\le \varepsilon.
\end{align}
\end{lemma}
\begin{proof}
The trace distance is contractive under quantum operations, and thus 
\begin{align}
\frac{1}{2}\norm{ \G[\Omega] - \G[\tilde{\Omega}]}_1 &\le \frac{1}{2} \norm{ \Omega - \tilde{\Omega}}_1\\ &= \frac{1}{2} \norm{(\eta- \tilde{\eta}) \otimes  \rho}_1 \\  &= \frac{1}{2}\norm{(\eta- \tilde{\eta}) }_1,
\label{subeq:trace_norm_omega}
\end{align}
where in the second equality we have used the identity $\norm{A \otimes B}_1 = \norm{A}_1 \norm{B}_1$.
Similarly, since $\G$ is trace-preserving and $\tr[ A \otimes B] = \tr [A] \tr [B]$, we have
\begin{equation}
   \frac{1}{2} \left| \tr \left(\G[\Omega]\right) - \tr \left(\G[\tilde{\Omega}]\right) \right| =   \frac{1}{2} \left| \tr(\Omega)- \tr (\tilde{\Omega}) \right| =\frac{1}{2} \left| \tr(\eta) - \tr(\tilde{\eta}) \right|.
    \label{subeq:abs_omega}
\end{equation}
Thus, combining results from Eqs.~(\ref{subeq:trace_norm_omega}) and (\ref{subeq:abs_omega}) we find the generalized trace distance between $\G[\Omega]$ and $\G[\tilde{\Omega}]$ is upper bounded as
\begin{align}
    D(\G[\Omega], \G[\tilde{\Omega}]) &\le  \left( \frac{1}{2}\norm{(\eta- \tilde{\eta})}_1  + \frac{1}{2}\left| \tr(\eta) - \tr (\tilde{\eta}) \right|  \right) =  D(\eta, \tilde{\eta})  .
\label{subeq:D_omega_less_P}
\end{align}
Therefore, if $\tilde{\eta} \in \B_\varepsilon (\eta)$ then $ D(\Omega, \tilde{\Omega})   \le \varepsilon$ immediately follows from \eqref{subeq:D_omega_less_P}, which concludes the proof of \lemref{lemma:variation_reference_to_omega}. \end{proof}

\begin{lemma}
\label{lemma:S_eta_tau_ball_difference}
If $\tilde{\eta} \in \B_\varepsilon(\eta)$, then 
\begin{equation}
     |H_{\eta} (\tau)- H_{\tilde{\eta}} (\tau) | \le \frac{d_R^2}{\ln 2} \left(\frac{\varepsilon}{1-2\varepsilon}\right),
\end{equation}
where $H_{\eta} (\tau) \coloneqq H_{\mathrm{min}}(R|A)_{\G[\Omega]}$.
\end{lemma}
\begin{proof}
Combining \lemref{lemma:variation_reference_to_omega} and Theorem \ref{thm:continuity} (and using the fact that the $G$-twirl is a trace-preserving map) we immediately have that 
\begin{equation}
       |H_{\eta} (\tau)- H_{\tilde{\eta}} (\tau) | \le \frac{d_R \min\{d_R, d_A \}}{\ln 2 \, \min \{\tr \Omega , \tr \tilde{\Omega} \}} \, \varepsilon,
\label{subeq:s_dif}
\end{equation}
for any $\tilde{\eta} \in \B_\varepsilon (\eta)$.
To get the simplified form as stated, we note that  $\min\{ \tr (\Omega),\tr (\tilde{\Omega}) \} =\tr (\tilde{\Omega})$ since $\Omega$ is normalised and $\tilde{\Omega}$ is sub-normalized. This then evaluates to
\begin{align}
    \tr (\tilde{\Omega}) = \tr (\tau) \tr(\tilde{\eta})= \tr(\tilde{\eta}).
\label{subeq:11}
\end{align}
Now $D(\eta, \tilde{\eta}) \le \varepsilon$ implies
$
   |1 - \tr (\tilde{\eta}) | \le 2 \varepsilon 
$, and therefore
\begin{align}
  \tr (\tilde{\eta}) \ge 1 - 2 \varepsilon.
\label{subeq:22}
\end{align}
Substituting Eq.~(\ref{subeq:22}) into Eq.~(\ref{subeq:11}) thus gives 
\begin{equation}
\min\{ \tr (\Omega),\tr (\tilde{\Omega}) \} \ge 1-2\varepsilon.
\label{subeq:blahh}
\end{equation}
Also, clearly $d_R  \ge \min\{d_R, d_A \} $. Substituting this and Eq. (\ref{subeq:blahh}) into Eq. (\ref{subeq:s_dif}) gives
\begin{equation}
     |H_{\eta} (\tau)- H_{\tilde{\eta}} (\tau) | \le \frac{d_R^2}{\ln 2} \frac{\varepsilon}{1-2\varepsilon},
\end{equation}
as claimed.
\end{proof}

We are now able to prove the result presented in the main text:

\coarsegrain*

\begin{proof}
First note:
\begin{align}
  \Delta H_{\tilde{\eta}}- \Delta H_{\eta} &= [ H_{\tilde{\eta}}(\sigma) - H_{\tilde{\eta}}(\rho)] - [ H_{\eta}(\sigma) - H_{\eta}(\rho)]\\
  &= [ H_{\tilde{\eta}}(\sigma) - H_{\eta}(\sigma)] - [ H_{\tilde{\eta}}(\rho) - H_{\eta}(\rho)]\eqqcolon x.
\end{align}
It then follows immediately from \lemref{lemma:S_eta_tau_ball_difference} that for any $\tilde{\eta} \in \B_\varepsilon (\eta)$ we have $\abs{x} \le  \frac{2 d_R^2}{\ln 2} \frac{\varepsilon}{1-2\varepsilon}$, which concludes the proof.
\end{proof}

\subsection{Proof of \thmref{thm:epsilon_finite_checks}}
\label{appx:smoothing_proof}

First we need the following theorem (e.g.~see \cite{tkocz2019introduction,Ledoux1991}):

\begin{theorem}
Let $\norm{\cdot}$ be a norm on $\mathbb{R}^d$. Then for every $\delta >0$, the unit sphere $\{ x \in \mathbb{R}^d, \norm{x}=1 \}$ admits a $\delta$-net, $\mathcal{N}$ with respect to the distance measured by $\norm{\cdot}$, of cardinality $|\mathcal{N}|$ such that
\begin{align}
|\mathcal{N}| \le \left(1 + \frac{2}{\delta}\right)^d.
\end{align}
\label{thrm:delta_net}
\end{theorem}
This theorem implies that there exists an $\varepsilon$-covering of a unit-sphere, with a finite number of elements, and which can be applied to a sphere of reference frame states around the maximally mixed state.
We restate the theorem which we seek to prove:
\epsilonfinitechecks*
\begin{proof}
Given \thmref{thm:ball_sufficient_refs}, we choose as our sufficient set of reference frame states the surface
\begin{equation}
\partial \D = \left\{ \eta : \eta = \frac{1}{d} (\mathbbm{1} + A), \mbox{ where } ||A||_\infty=1 \mbox{ and } \tr(A)=0\right\}.
\end{equation}
It is seen by inspection that this gives a closed surface of quantum states that contain the maximally mixed state. We can describe $\partial \D$ entirely in the space of traceless $d \times d$ Hermitian matrices, which is embedded in $\mathbb{R}^{d^2-1}$. In this embedding space the surface is the unit ball defined by $||A||_\infty =1$, and so by the above theorem admits an $\delta$--net covering in the $|| \cdot ||_\infty$ norm. The cardinality of this covering $\mathcal{N}$ obeys
\begin{equation}
|\mathcal{N}| \le \left ( 1 + \frac{2}{\delta} \right )^{d^2-1}.
\end{equation}
We want an $\varepsilon$-net in the $||\cdot ||_1$ norm, so that for any $\eta = \frac{1}{d}(\mathbbm{1} + A)$ in $\partial \D$ there is an $\eta_k= \frac{1}{d}(\mathbbm{1} + A_k)$ in the net such that
\begin{equation}
D(\eta, \eta_k) = \frac{1}{2}||\eta -\eta_k||_1 - \frac{1}{2}\abs{\tr \eta - \tr \eta_k}= \frac{1}{2d} || A - A_k ||_1 \le \varepsilon.
\end{equation}
However in finite dimensions all matrix norms are equivalent, and we have that~\cite{horn2012matrix}
\begin{equation}
||A- A_k||_1 \le \mbox{rank}(A-A_k) ||A-A_k||_\infty \le d\delta.
\end{equation}
Therefore choosing $\delta = 2 \varepsilon$ ensures that $D(\eta,\eta_k) \le \varepsilon$ as required, and we can always find a sufficient set of reference frame states $\mathcal{N}=\{ \eta_k \}_{k=1}^N$ such that $N \le (1+ 1/\varepsilon)^{d^2-1}$.

We now check the $\Delta H_\eta$ condition on each $\eta = \eta_k$ in the $\varepsilon$--net. If $\Delta H_{\eta_k} <0$ for one of these states then the transformation is impossible. If however we find that 
\begin{equation}
\Delta H_{\eta_k} \ge r(\varepsilon) \mbox{ for all }\eta_k \in \mathcal{N},
\end{equation}
then we know from the continuity of the function $\Delta H_\eta$, \lemref{lemma:coarse_grain_less_than_r}, that this implies that $\Delta H_{\eta} \ge 0$ for all reference frame states $\eta \in \partial \D$. However this is a sufficient set of states (\thmref{thm:ball_sufficient_refs}) and therefore we deduce from these $N$ conditions that the transformation is possible covariantly. The final case of at least one of the $\eta_k$ conditions giving $0 \le \Delta H_{\eta_k} \le  r(\varepsilon)$ can be handled as follows: we can supplement the state $\rho_A$ with an additional reference frame state $\chi_{A'}$. If we take $\chi_{A'}$ to be large and approximating a perfect reference frame state (perfectly encoding the group element) then it is possible to transform to any quantum state covariantly. This does so by reducing the entropy $H_{\rm min}(R|AA')$. This can be used to increase $\Delta H_{\eta_k}$ and ensure that $\Delta H_{\eta_k} \ge  r(\varepsilon)$, from which we can deduce that the transformation is now possible. The state $\chi_{A'}$ therefore provides an upper bound estimate on the minimal additional asymmetry required to realise the transformation. Since $r(\varepsilon) = O (\varepsilon)$ the requirement on $\chi$ is to provide $O(\varepsilon)$ resources as measured by the single-shot entropy.
 \end{proof}

\section{The conical structure of $\Phi_\tau(\mathbf{x})$ (proof of \lemref{lemma:cusp})}
\label{appx:cusp}

\cusp*

\begin{proof}
Let $\P$ be the completely depolarising channel on the reference system, and define $A(\bm{x}) \coloneqq \sum_{k=1}^{d^2-1} x_k X_k$. Then by Property \ref{property:F_data_processing} of $\Phi_\tau(\mathbf{x})$, we have: 
\begin{align}
    \Phi\left( \G\left[ \left(\frac{\id}{d} + A(\mathbf{x})\right) \otimes \tau\right] \right) &\geq \Phi\left(\P \otimes \id \circ \G\left[ \left(\frac{\id}{d} + A(\mathbf{x})\right) \otimes \tau\right] \right)\\
    &= \Phi\left(\G\left[ \P\left(\left(\frac{\id}{d} + A(\mathbf{x})\right)\right) \otimes \tau\right]  \right)\\
    &= \Phi\left(\G\left[ \frac{\id}{d} \otimes \tau\right] \right) \equiv \Phi_\tau(\mathbf{0}).
\end{align}
Therefore, $\tilde{\Phi}_\tau(\mathbf{x}) \geq 0$, as claimed.

Let $p$ be a probability. Then
\begin{align}
p\eta(\mathbf{x}) + (1-p)\frac{\id}{d} = p\frac{\id}{d} + pA(\mathbf{x}) + (1-p)\frac{\id}{d} = \frac{\id}{d} + A(p\mathbf{x})
\end{align}

Using this equation, we can rewrite \lemref{lemma:partially-depol1} as:
\begin{lemma}
	Let $p$ be a probability; i.e. $0 \leq p \leq 1$. Then for all $\mathbf{x} \in \S$:
	\begin{align}
	\Phi_\tau\left(p\mathbf{x} \right) = p\Phi_\tau(\mathbf{x}) + \frac{1-p}{d} 
	\end{align}
	\label{lemma:tilde_phi_positivity}
\end{lemma}
We then immediately have
\begin{align}
\tilde{\Phi}_\tau(p\mathbf{x}) = p\Phi_\tau(\mathbf{x}) + \frac{1-p}{d} - \frac{1}{d} = p\Phi_\tau(\mathbf{x}) - \frac{p}{d} = p\tilde{\Phi}_\tau(\mathbf{x}),\ \forall \mathbf{x} \in \S. 
\end{align}
Making the change of variables $\mathbf{x}' \coloneqq p\mathbf{x}$, we find that the equation above is equivalent to:
\begin{align}
\tilde{\Phi}_\tau(\mathbf{x}') = p\tilde{\Phi}\left(\frac{1}{p}\mathbf{x}'\right) \Rightarrow \frac{1}{p}\tilde{\Phi}_\tau(\mathbf{x}') = \tilde{\Phi}_\tau\left(\frac{1}{p}\mathbf{x}'\right), \forall\ \frac{1}{p}\mathbf{x'} \in \S. 
\end{align}
Since $1 \leq \frac{1}{p} \leq \infty$, we can combine the facts above and conclude that, for all $\lambda \geq 0$ such that $\mathbf{x}, \lambda \mathbf{x} \in \S$: 
\begin{align}
	\tilde{\Phi}_\tau(\lambda\mathbf{x}) = \lambda \tilde{\Phi}_\tau(\mathbf{x}),
\end{align}
as claimed.
\end{proof}

\section{Calculating $\Phi_\tau(\mathbf{x})$ for covariant transformations in a qubit}
\subsection{Time-Covariant Transformations ($U(1)$)}
\label{appx:qubit_u1_phi_derivation}

We will consider a qubit with the Hamiltonian $\sigma_z$, transformations among whose states are limited to those that are symmetric under all time translations $\{e^{i\sigma_z t}\ |\ 0 \le t < 2\pi \}$. The reference system is then a qubit with Hamiltonian $- \sigma_z$. We will calculate $\Phi_\tau(\mathbf{x})$ for an arbitrary state $\tau$ of this qubit, while reference states are given using the Bloch co-ordinates:
\begin{align}
	\eta(x,y,z) = \frac{\id}{2} + x\frac{\sigma_x}{2} + y\frac{\sigma_y}{2} + z\frac{\sigma_z}{2},
\end{align}
where the $\sigma_i$ are the Pauli matrices, and $x^2 + y^2 + z^2 \leq 1$. The derivation will be conducted entirely in the energy eigenbasis.

Recall from \eqref{eq:Phi(tau)}:
\begin{align}
	\Phi_\tau(x,y,z) \equiv \Phi_{\eta(x,y,z)}(\tau) \coloneqq \inf_{X \ge 0} \{\tr(X) : \id \otimes X - \G(\eta(x,y,z) \otimes \tau) \ge 0\} 
\end{align}

We first introduce this simplifying lemma that allows us to vastly reduce which $X$ we must consider:

\begin{lemma}
	When calculating $\Phi_\eta(\tau)$, it is sufficient to minimise over $X$ such that $X = \G(X)$.
\end{lemma}
\begin{proof}
	Since $\U_g$ can be regarded as an (active) change of basis, it does not affect the eigenvalues of a Hermitian operator $K$. Therefore,
	\begin{align}
		&K \ge 0 \Rightarrow \U_g(K) \ge 0\\
		&\tr(K) = \tr[\U_g(K)]
	\end{align} 
	for all $g$ and Hermitian $K$. Averaging over the entire group $G$, we arrive at
	\begin{align}
		&K \ge 0 \Rightarrow \G(K) \ge 0 \label{eq:1111} \\
		&\tr(K) = \tr[\G(K)] \label{eq:2222}
	\end{align} 	
	
	If $X$ is a feasible solution, then it obeys the two conic constraints in the SDP defining $\Phi_\eta(\tau)$: $X \ge 0$ and ${\id \otimes X - \G(\eta \otimes \tau) \ge 0}$. In this case, we see immediately from Equation \ref{eq:1111} that $\G(X)$ also obeys the first constraint. Furthermore, since $\G[\id \otimes X - \G(\eta \otimes \tau)] = \id \otimes \G(X) - \G(\eta \otimes \tau)$, applying Equation \ref{eq:1111} to $\id \otimes X - \G(\eta \otimes \tau)$ shows that $\G(X)$ also obeys the second constraint. Therefore, if $X$ is a feasible solution, then so is $\G(X)$.  
	
	By Equation \ref{eq:2222}, $\G(X)$ produces the same value as $X$ on the objective function to be minimised in the SDP defining $\Phi_\eta(\tau)$. We can thus further conclude that if $X$ is a feasible solution, then $\G(X)$ is an equally good feasible solution. It is then sufficient to only minimise over $X$ such that $X = \G(X)$ when calculating $\Phi_\eta(\tau)$. 
\end{proof}

This lemma means we can take $X$ to be diagonal without loss of generality, i.e. 
\begin{align}
	X = \begin{pmatrix}x_1&0\\0&x_2\end{pmatrix} 
	\Rightarrow \id \otimes X = \begin{pmatrix} x_1&0&0&0\\0&x_2&0&0\\0&0&x_1&0\\0&0&0&x_2\end{pmatrix}
\end{align}

We have already seen from the main text (\eqref{eq:phi_qubit_u1_cylindrical_asymmetry}) that $\Phi_\tau(x,y,z)$ is cylindrically symmetric about the $z$-axis. This means we can restrict our calculation to $x \ge 0, y=0$, and then equate
\begin{align}
\Phi_\tau(x,y,z) = \Phi_\tau(\sqrt{x^2+y^2},0,z)
\label{eq:simplifier_a}
\end{align}
Having applied this restriction, we further see from \eqref{eq:phi_qubit_u1_top_bottom_bloch_sphere} that
\begin{align}
	\Phi_{\sigma_x \tau \sigma_x}(x,0,z) = \Phi_{\tau}(x,0,-z),
	\label{eq:simplifier_b} 
\end{align}	
which means we can additionally restrict our attention to $z \ge 0$. We therefore only need to consider reference states of the form:  
\begin{align}
	\eta = \frac{1}{2}
	\begin{pmatrix}1+z&x\\x&1-z
	\end{pmatrix},\ x,z \ge 0.
\end{align} 

Recalling the parameterisation (\eqref{eq:qubit_tau_parameterisation})
\begin{align}
	\tau \coloneqq \begin{pmatrix}
	p_\tau& c_\tau\\
	c^*_\tau& 1-p_\tau
	\end{pmatrix},
\end{align}
we calculate
\begin{align}
\eta \otimes \tau &= \frac{1}{2} \begin{pmatrix} 1+z& x\\ 
x&1-z\end{pmatrix} \otimes 
\begin{pmatrix}p_\tau&c_\tau\\
c^*_\tau&1-p_\tau\end{pmatrix}\\ 
&= \frac{1}{2} \begin{pmatrix} p_\tau(1+z)&c_\tau(1+z)&p_\tau x&c_\tau x\\
c^*_\tau(1+z)&(1-p_\tau)(1+z)&c^*_\tau x&(1-p_\tau)x\\
p_\tau x&c_\tau x&p_\tau(1-z)&c_\tau(1-z)\\
c^*_\tau x&(1-p_\tau)x&c^*_\tau(1-z)&(1-p_\tau)(1-z)\end{pmatrix}
\end{align}

The $G$-twirling only leaves elements in $\vspan(|00\rangle,|11\rangle), \vspan(|01\rangle)$ and $\vspan(|10\rangle)$ intact:
\begin{align}
\G(\eta \otimes \tau) = \frac{1}{2} \begin{pmatrix} p_\tau(1+z)&0&0&c_\tau x\\
0&(1-p_\tau)(1+z)&0&0\\
0&0&p_\tau(1-z)&0\\
c^*_\tau x&0&0&(1-p_\tau)(1-z)\end{pmatrix}.
\end{align}
Therefore:
\begin{align}
	\id \otimes X  - \mathcal{G}(\eta \otimes \tau) = 
	\begin{pmatrix} x_1 - p_\tau\frac{1+z}{2}&0&0&-c_\tau\frac{x}{2}\\
	0&x_2-(1-p_\tau)\frac{1+z}{2}&0&0\\
	0&0&x_1-p_\tau \frac{1-z}{2}&0\\
	-c^*_\tau \frac{x}{2}&0&0&x_2-(1-p_\tau)\frac{1-z}{2}\end{pmatrix}.
\end{align}
Making a (passive) change of basis so the matrix appears block-diagonal, we obtain:
\begin{align}
	\id \otimes X  - \mathcal{G}(\eta \otimes \tau) = 
	\begin{pmatrix}x_1 - p_\tau \frac{1+z}{2}&-c_\tau\frac{x}{2}&0&0\\
	-c^*_\tau \frac{x}{2}&x_2-(1-p_\tau)\frac{1-z}{2}&0&0\\
	0&0&x_2-(1-p_\tau)\frac{1+z}{2}&0\\
	0&0&0&x_1-p_\tau\frac{1-z}{2}\end{pmatrix}.
	\label{eq:qubit_u1_positivity}
\end{align}
The Sylvester Criterion \cite{horn2012matrix} states that a matrix is semi-definite positive if and only if all its upper-left determinants are greater than or equal to 0. This produces the following criteria:
\begin{enumerate}
	\item $x_1 \geq p_\tau \frac{1+z}{2}$.
	\item $(x_1 - p_\tau \frac{1+z}{2})(x_2 -(1-p_\tau )\frac{1-z}{2}) \geq \frac{\abs{c_\tau}^2x^2}{4}$.
	\item $x_2 \geq (1-p_\tau)\frac{1+z}{2}$.
	\item $x_1 \geq p_\tau\frac{1-z}{2}$.
\end{enumerate}

Since we have restricted ourselves to $x \ge 0$, condition 4 is redundant given condition 1. Furthermore, given any value of $x_2$ that satisfies condition 3, the smallest value of $x_1$ satisfying condition 1 is 
\begin{align}
	x_1 - p_\tau \frac{1+z}{2} =  \frac{\abs{c_\tau}^2x^2}{4} \frac{1}{x_2 -(1-p_\tau)\frac{1-z}{2}}.
\end{align}
We are therefore looking to minimise
\begin{align}
	x_1 + x_2 = x_2 + \frac{\abs{c_\tau}^2x^2}{4} \frac{1}{x_2 -(1-p_\tau)\frac{1-z}{2}} + p_\tau\frac{1+z}{2}.
\end{align} 
This occurs at
\begin{align}
 	x_2 = (1-p_\tau)\frac{1-z}{2} \pm \abs{c_\tau}\frac{x}{2}.
 \end{align}
However, $x_2$ also has to satisfy condition 3. Therefore, the value $x_2$ should take is whichever of the following two
\begin{align}
x_2 = 
\begin{cases}
(1-p_\tau)\frac{1+z}{2}\\
(1-p_\tau)\frac{1-z}{2} + \abs{c_\tau}\frac{x}{2}
\end{cases}
\label{eq:qubit_magic_two_solutions}
\end{align}
is bigger, which results in
\begin{align}
	\Phi_\tau(x \ge 0,y=0,z\ge 0) =
	\begin{cases}
		\frac{\abs{c_\tau}^2}{1-p_\tau} \frac{x^2}{4z} + \frac{z}{2} + \frac{1}{2}& \text{ for } \frac{x}{2z} \le \frac{1-p_\tau}{\abs{c_\tau}}\\
		\left(p_\tau - \frac{1}{2}\right)z + \abs{c_\tau}x + \frac{1}{2}& \text{ for } \frac{x}{2z} \ge \frac{1-p_\tau}{\abs{c_\tau}}
	\end{cases} 
\end{align}
Finally, applying Equations \ref{eq:simplifier_a} and \ref{eq:simplifier_b}, and noting that $\sigma_x \tau \sigma_x$ translates to $p_\tau \rightarrow 1-p_\tau$ and $c_\tau \rightarrow c^*_\tau$ in our parameterisation of $\tau$, we can calculate $\Phi_\tau$ for all reference states from the above result as seen in \eqref{eq:qubit_u1_phi}:
\begin{align}
\Phi_\tau(x,y,z) =
\begin{cases}
\frac{\abs{c_\tau}^2}{1-p_\tau} \frac{x^2+y^2}{4z} + \frac{z}{2} + \frac{1}{2}& \text{ for } 0 \le \frac{\sqrt{x^2+y^2}}{2z} \le \frac{1-p_\tau}{\abs{c_\tau}}\\
\left(p_\tau - \frac{1}{2}\right)z + \abs{c_\tau}\sqrt{x^2+y^2} + \frac{1}{2}& \text{ for } \frac{\sqrt{x^2+y^2}}{2z} \ge \frac{1-p_\tau}{\abs{c_\tau}} \text{ and } \frac{\sqrt{x^2+y^2}}{2z} \le -\frac{p_\tau}{\abs{c_\tau}}\\
-\frac{\abs{c_\tau}^2}{p_\tau} \frac{x^2+y^2}{4z} - \frac{z}{2} + \frac{1}{2}& \text{ for } 0 \ge \frac{\sqrt{x^2+y^2}}{2z} \ge -\frac{p_\tau}{\abs{c_\tau}}.
\end{cases} 
\label{eq:qubit_u1_phi_xyz}
\end{align}

\subsubsection{Two entropic conditions suffice to characterise time-covariant transformations in a non-degenerate qubit}
\label{appx:qubit_u1_phi_pure}
We can alternatively characterise each qubit reference frame state as $\eta(r,\theta,\phi)$, where $(r,\theta,\phi)^T$ is state's Bloch vector in spherical polar co-ordinates (radial, polar and azimuthal respectively). Using the standard conversion between spherical and Cartesian co-ordinates $x=r\sin(\theta)\cos(\phi)$, $y=r\sin(\theta)\sin(\phi)$ and $z=r\cos(\theta)$, we can rewrite Equation \ref{eq:qubit_u1_phi_xyz} as
\begin{align}
\Phi_\tau(x,y,z) &\equiv \Phi_\tau(r\sin(\theta)\cos(\phi),r\sin(\theta)\sin(\phi),r\cos(\theta))\\
&=
\begin{cases}
r\frac{\abs{c_\tau}^2}{1-p_\tau} \frac{\tan(\theta)\sin(\theta)}{4} + r \frac{\cos(\theta)}{2} + \frac{1}{2}& \text{ for } 0 \le \frac{\tan(\theta)}{2} \le \frac{1-p_\tau}{\abs{c_\tau}}\\
r\left(p_\tau - \frac{1}{2}\right)\cos(\theta) + r\abs{c_\tau}\sin(\theta) + \frac{1}{2}& \text{ for } \frac{\tan(\theta)}{2} \ge \frac{1-p_\tau}{\abs{c_\tau}} \text{ and } \frac{\tan(\theta)}{2} \le -\frac{p_\tau}{\abs{c_\tau}}\\
-r\frac{\abs{c_\tau}^2}{p_\tau} \frac{\tan(\theta)\sin(\theta)}{4} - r\frac{\cos(\theta)}{2} + \frac{1}{2}& \text{ for } 0 \ge \frac{\tan(\theta)}{2} \ge -\frac{p_\tau}{\abs{c_\tau}}.
\end{cases} 
\label{eq:qubit_u1_phi_pure}
\end{align}

\begin{lemma}
	There exists a time-covariant transformation $\rho$ to $\sigma$ in a qubit with Hamiltonian $\sigma_z$ if and only if
	\begin{align}
		\partial^2_\theta (\Delta \Phi_\eta)|_{\theta = 0} \ge 0 \text{ and } \partial^2_\theta (\Delta \Phi_\eta)|_{\theta = \pi} \ge 0
	\end{align}
	for all $0 < r \leq 1$ and $0 \leq \phi < 2\pi$, where we recall $\Delta \Phi_\eta \coloneqq \Phi_{\eta(r,\theta,\phi)}(\rho) - \Phi_{\eta(r,\theta,\phi)}(\sigma)$ and $\partial_\theta \coloneqq \frac{\partial}{\partial \theta}$. Since $\Phi_\eta(\tau)$ is monotonically decreasing in $H_\eta(\tau)$, this is equivalent to
	\begin{align}
		\partial^2_\theta (\Delta H_\eta)|_{\theta = 0} \ge 0 \text{ and } \partial^2_\theta (\Delta H_\eta)|_{\theta = \pi} \ge 0,
	\end{align}
	where we recall $\Delta H_\eta \coloneqq H_{\eta(r,\theta,\phi)}(\sigma) - H_{\eta(r,\theta,\phi)}(\rho)$. 
\end{lemma}  
\begin{proof}
	From \eqref{eq:qubit_u1_phi_pure}, we can straightforwardly evaluate:
	\begin{align}
		\partial^2_\theta \Phi_\eta(\tau)|_{\theta = 0}  = \frac{r}{2}\left(\frac{\abs{c_\tau}^2}{1-p_\tau} - 1\right) \text{ and } \partial^2_\theta \Phi_\eta(\tau)|_{\theta = \pi} = \frac{r}{2}\left(\frac{\abs{c_\tau}^2}{p_\tau} - 1\right),
	\end{align}
	from which it immediately follows that
	\begin{align}
		\partial^2_\theta (\Delta \Phi_\eta)|_{\theta = 0} \ge 0  \Longleftrightarrow c_\sigma \le c_\rho \sqrt{\frac{1-p_\sigma}{1-p_\rho}} \text{ and }
		\partial^2_\theta (\Delta \Phi_\eta)|_{\theta = \pi} \ge 0  \Longleftrightarrow c_\sigma \le c_\rho \sqrt{\frac{p_\sigma}{p_\rho}}, 
	\end{align}
	which are known necessary and sufficient conditions for a time-covariant transfromation from $\rho$ to $\sigma$ in a qubit with the Hamiltonian $\sigma_z$ \cite{matteo_kamil_bound}. 
\end{proof}
We note that $\partial_\theta (\Delta \Phi_\eta)|_{\theta = 0} = 0$ and $\partial_\theta(\Delta \Phi_\eta)|_{\theta = \pi} = 0$. At $r=1$, this Lemma can be loosely interpreted as asserting that two pure reference states, infinitesimally close to the energy eigenstates $\ketbra{0}$ and $\ketbra{1}$ respectively, are sufficient for determining whether time-covariant interconversion is possible between \emph{any} input and output states in a non-degenerate qubit. 

\subsection{Unitarily covariant transformations ($SU(2)$)}
\label{appx:qubit_su2_phi_derivation}

The only channels that are covariant with all unitary transformations on the qubit are~\cite{Cirstoiu2020_Noether}:
\begin{align}
\E_\lambda(\rho) = \frac{1}{2}(\id + \lambda \bm{r} \cdot \bm{\sigma}),
\end{align}
where $\lambda \in [-1/3, 1]$, and $\bm{r}$ is the Bloch vector of $\rho$. Therefore, writing $\bar{\bm{x}} = (x,-y,z)$ given the Bloch vector $\bm{x} = (x,y,z)$ of $\eta$, we have that
\begin{align}
\tr[\eta^T  \E_\lambda(\rho)] = \frac{1}{2}(1  + \lambda \bar{\bm{x}} \cdot \bm{r}) .
\end{align}
We now want to consider the quantity 
\begin{align}
\Phi_\rho(\bm{x}) = \max_{\E \in \mathcal{O}_{\rm cov}} \tr[\eta^T  \E(\rho)] = \max_{ \lambda \in \left[-\frac{1}{3} , 1\right]}\tr[\eta^T  \E_\lambda(\rho)] =\max_{ \lambda \in \left[-\frac{1}{3} , 1\right]}\frac{1}{2}(1  + \lambda \bar{\bm{x}} \cdot \bm{r}) .
\end{align}
By convexity, we can restrict the set of $\lambda$ that we must optimize over to the extremal values, which gives:
\begin{align}
\Phi_\rho(\bm{x}) &= \begin{cases}
\frac{1}{2}(1  +  \bar{\bm{x}} \cdot \bm{r})  & \text{if } \bar{\bm{x}} \cdot \bm{r} \ge 0,\\
\frac{1}{2}(1 - \frac{1}{3} \bar{\bm{x}} \cdot \bm{r}) & \mbox{ otherwise}.
\end{cases}
\end{align}
In words, if the Bloch vectors of $\eta^T$ and $\rho$ are both located in the hemisphere agove the plane perpendicular to the Bloch vector of $\eta^T$, then the identity channel $\E_1=\mathcal{I}$ optimizes the objective function. Otherwise, the channel $\E_{-\frac{1}{3}}$ which achieves the best inversion of $\rho$ that can be done covariantly, is optimal.

\section{Depolarization conditions}
\label{appx:sufficient_conditions}

\subsection{Depolarization conditions via the modes of asymmetry}

\begin{lemma}
\label{lemma:general_SC_hilbert_schmidt}
 Let $\rho_A$, $\sigma_B$, and $\eta_R$ be quantum states on systems $A$, $B$, and $R$, respectively, where $d_R=d_B$. We have $\rho \xrightarrow{G} \sigma$ if for all quantum states $\eta$ on $\H_R$
\begin{align}
\langle \eta , \E^\eta(\rho) \rangle  \ge \langle \eta , \sigma \rangle ,
\end{align}
for any family of covariant channels $\E^\eta :  \S(\H_A) \rightarrow \S(\H_B)$ parameterised by $\eta$, where $\langle A, B \rangle \coloneqq \tr[A^\dagger B]$.
\end{lemma}
\begin{proof}
Lemma 3 of the supplemental material of \cite{gour2018quantum} tells us that $\rho \xrightarrow{G} \sigma$ if and only if
\begin{align}
    \Phi_\eta(\rho) \ge \tr [\eta^T \sigma], \forall \eta.
\label{blahh1}
\end{align}

From the dual expression for $\Phi_\eta (\rho)$ in \lemref{lemma:Phi_and_dual}, we have the following lower bound on $ \Phi_\eta(\rho)$ for any choice of covariant channel $\E^\eta$:
\begin{equation}
     \Phi_\eta(\rho) \ge \tr [\eta^T \E^\eta(\rho) ].
\label{blahh2}
\end{equation}
Eqs. (\ref{blahh1}) and (\ref{blahh2}) allow us to generate the following \textit{sufficient} condition, such that if we can find any faimily of covariant maps $\E^\eta$ such that
\begin{align}
    \tr[ \eta^T \E^\eta (\rho) ] \ge \tr[ \eta^T  \sigma ] , \forall \eta,
    \label{eq:ur_sufficient_condition}
\end{align}
or, equivalently 
\begin{align}
     \tr[ \eta \E^\eta(\rho) ] \ge \tr[\eta \sigma] , \forall \eta,
\end{align}
then we are guaranteed $\rho \xrightarrow{G} \sigma$, which completes the proof.
\end{proof}

Lemmata \ref{lemma:general_SC_hilbert_schmidt} and \ref{lemma:mode_decomp_trace_product} together produce the following corollary: 
\begin{corollary}
\label{corollary:mode_decomp_sc}
     We have that $\rho \xrightarrow{G} \sigma$ if for all $\eta \in \S(\H_R)$ ,
     \begin{align}
 \sum_{\lambda,j}  \langle \eta^{\lambda}_j, \E^\eta (\rho^\lambda_j) - \sigma_j^\lambda \rangle \ge  0.
\label{eq:sum_l_j_SC}
\end{align}
    for some family of covariant channels $\E^\eta$.
\end{corollary}

\subsection{PGM measure-and-prepare channel}

For a general group $G$ we have the Pretty Good Measurement~\cite{hausladen1994pretty} POVM:
\begin{align}
M_{\mathrm{pgm}}(g) \coloneqq \G(\rho)^{-\frac{1}{2}} \rho(g)  \G(\rho)^{-\frac{1}{2}},
\end{align}
where $\rho(g) \coloneqq \U_g(\rho)$.

Recall $\sigma^\lambda_j = \E(\rho^\lambda_j)$ \cite{marvian2014modes} and $\tr [AB] = \sum_{\lambda,j} \tr[A^{\lambda}_j B^{\lambda^*}_j]$, we can evaluate the trace product for the PGM measure-and-prepare channel, which prepares $\U_g(\tau)$ for an \emph{arbitrary} state $\tau$ when the $g$th outcome is obtained:
\begin{equation}
\E_{\rm pgm} (\rho) := \int dg\ \tau(g) \tr[M_{\rm pgm}(g) \rho],
\label{eq:pgm_channel}
\end{equation}
which, from its form, is manifestly covariant.

We then have that
\begin{align}
\tr[\eta\E_{\mathrm{pgm}}(\rho)] = \sum_{\lambda,j} \tr[\eta_j^{\lambda^*} \E_{\mathrm{pgm}}(\rho_{j}^\lambda)] = \sum_{\lambda,j} \int dg \tr[M_{\mathrm{pgm}}(g) \rho_{j}^\lambda ] \tr[\eta^{\lambda^*}_j \tau(g) ] .
\end{align}
We can write each term in the above equation more compactly in terms of the Hilbert-Schmidt inner product: 
\begin{align}
\langle \eta^\lambda_j, \E_{\mathrm{pgm}}(\rho_{j}^\lambda) \rangle = \int dg \, \langle \eta_j^\lambda, \tau(g) \rangle \langle \rho^{\lambda^*}_j, \overline{\rho}(g) \rangle,
\end{align}
where we have defined the scaled variant of $\rho$, $\overline{\rho}$ as
\begin{align}
\overline{\rho} \coloneqq \G(\rho)^{-\frac{1}{2}} \rho \G(\rho)^{-\frac{1}{2}} .
\end{align}

We now have the following lemma:
	\begin{lemma} For any irrep $\lambda$ and irrep component $j$ we have
		\begin{equation}
		\langle \eta^\lambda_j, \E_{\mathrm{pgm}}(\rho_{j}^\lambda) \rangle = f^\lambda_j(\rho)\langle \eta_j^\lambda, \tau_j^\lambda \rangle,
		\end{equation}
	where we have introduced the functions $f^\lambda_j (\rho)\coloneqq \langle \rho^{\lambda^*}_j,  \overline{\rho}_j^{\lambda^*} \rangle$.
		\label{lemma:eta_pgm_prod}
	\end{lemma}
\begin{proof}
We have:
	\begin{align}
	\langle \eta^\lambda_j, \E_{\mathrm{pgm}}(\rho_{j}^\lambda) \rangle &= \int dg \, \langle \eta_j^\lambda, \tau(g) \rangle \langle \rho^{\lambda^*}_j, \overline{\rho}(g) \rangle \\
	&= \int dg \, \langle \eta^\lambda_j \otimes \rho^{\lambda^*}_j, \tau(g) \otimes \overline{\rho}(g) \rangle \\
	&=  \langle \eta^\lambda_j \otimes \rho^{\lambda^*}_j, \int dg \, \U_g(\tau \otimes \overline{\rho}) \rangle \\
	&=  \langle \eta^\lambda_j \otimes \rho^{\lambda^*}_j, \G(\tau \otimes \overline{\rho}) \rangle \label{eq:ggg} \\
	&= \sum_{i,\mu} \langle \eta^\lambda_j \otimes \rho^{\lambda^*}_j , \tau_i^\mu \otimes \overline{\rho}_i^{\mu^*} \rangle \label{eq:hhh} \\
	&=\langle \eta^\lambda_j \otimes \rho^{\lambda^*}_j, \tau_j^\lambda \otimes \overline{\rho}_j^{\lambda^*} \rangle \\
	&=\langle \eta^\lambda_j , {\tau}_j^\lambda  \rangle \langle \rho^{\lambda^*}_j,  \overline{\rho}_j^{\lambda^*} \rangle\\
	&= \langle \eta^\lambda_j , f^\lambda_j(\rho) {\tau}_j^\lambda  \rangle,
	\end{align}
	where in going from line \ref{eq:ggg} to line \ref{eq:hhh} we have made use of lemma \ref{lemma:G_twirl_modal_match}. This completes the proof.
\end{proof}
We now note the following properties of the functions $f^\lambda_j(\rho)$, which hold for any input state $\rho$ and for all modes labelled by $(\lambda,j)$:

\begin{enumerate}[label=\normalfont \textbf{(F\arabic*)}]
\item \label{propertyf:real_valued} \textit{(Real-valued).} $ f^\lambda_j (\rho) = f^{\lambda^*}_j (\rho) = [f^\lambda_j (\rho)]^*$.
\item \label{propertyf:non-negative} \textit{(Non-negative).} $f^\lambda_j (\rho) \ge 0$.
\item \label{propertyf:zero_mode} \textit{(Trivial irrep).} $f^0_0 (\rho) = 1$.
\end{enumerate}

\begin{proof}[Proof of \ref{propertyf:real_valued}]
We note that 
\begin{align}
f_j^\lambda (\rho) = \tr[\rho^\lambda_j \G(\rho)^{-\frac{1}{2}}\rho^{\lambda^*}_j \G(\rho)^{-\frac{1}{2}}] &=\tr[\rho^{\lambda^*}_j \G(\rho)^{-\frac{1}{2}}\rho^\lambda_j \G(\rho)^{-\frac{1}{2}}] = f_j^{\lambda^*} (\rho)\\
&= \tr[(\rho^\lambda_j \G(\rho)^{-\frac{1}{2}}\rho^{\lambda^*}_j \G(\rho)^{-\frac{1}{2}})^\dagger] = \tr[\rho^\lambda_j \G(\rho)^{-\frac{1}{2}}\rho^{\lambda^*}_j \G(\rho)^{-\frac{1}{2}}]^* = [f^\lambda_j(\rho)]^*
\end{align}
\end{proof}

\begin{proof}[Proof of \ref{propertyf:non-negative}]
   We begin by noting that 
\begin{align} 
\rho^\lambda_j\mathcal{G}(\rho)^{-\frac{1}{2}}\rho^{\lambda^*}_j = \left( \rho^\lambda_j \G(\rho)^{-\frac{1}{4}}\right) \left( \rho^\lambda_j \G(\rho)^{-\frac{1}{4}}\right)^\dagger \eqqcolon AA^\dagger,
\end{align}  
where we have defined $A \coloneqq\rho^\lambda_j \G(\rho)^{-\frac{1}{4}} $. Regardless of what $A$ is, any operator of the form $AA^\dagger$ is positive. Since $\mathcal{G}(\rho)^{-\frac{1}{2}}$ is a positive operator as well, this means $\rho^\lambda_j\mathcal{G}(\rho)^{-\frac{1}{2}}\rho^{\lambda^*}_j\mathcal{G}(\rho)^{-\frac{1}{2}}$ must be positive, so we are guaranteed $f^\lambda_j(\rho) \geq 0, \forall \lambda,j,\rho$.

\end{proof}  

\begin{proof}[Proof of \ref{propertyf:zero_mode}]
 We first note  
\begin{align}
\mathcal{G}(\rho) &= \sum_{\lambda,j} \int dg\ \mathcal{U}_g\left(\sum_\alpha \tr({X_j^{(\lambda,\alpha)}}^\dagger\rho)X_j^{(\lambda,\alpha)}\right) \\
& = \sum_{\lambda,j,j'} \int dg v^{(\lambda)}_{j'j}(g) \sum_\alpha \tr({X_j^{(\lambda,\alpha)}}^\dagger\rho) X_{j'}^{(\lambda,\alpha)}\\
&= \sum_{\lambda,j,j'} \left[\int dg v^{(\lambda)}_{j'j}(g)v^0_{0,0}(g)\right] \sum_\alpha \tr({X_j^{(\lambda,\alpha)}}^\dagger\rho) X_{j'}^{(\lambda,\alpha)} \label{eq:iii}\\
&= \sum_{\lambda,j,j'} \delta_{\lambda,0}\delta_{j',0,}\delta_{j,0} \sum_\alpha \tr({X_j^{(\lambda,\alpha)}}^\dagger\rho) X_{j'}^{(\lambda,\alpha)}\label{eq:jjj}\\
&= \sum_\alpha \tr({X_0^{(0,\alpha)}}^\dagger\rho) X_{0}^{(0,\alpha)} = \rho^0_0.
\end{align}
where in going from line \ref{eq:iii} to line \ref{eq:jjj} we have made use of Equation \ref{eq:marvian_mode_ortho}. It then follows that:
\begin{align}
f_0^0(\rho) = \tr[\mathcal{G}(\rho)\mathcal{G}(\rho)^{-\frac{1}{2}}\mathcal{G}(\rho)\mathcal{G}(\rho)^{-\frac{1}{2}}] = \tr[\mathcal{G}(\rho)] = 1,
\end{align}
which completes the proof.

\end{proof}

\subsection{Derivation of general conditions}

\begin{lemma}
    
Let us define components $t^{(\lambda,\alpha)}_j \coloneqq \langle X^{(\lambda,\alpha)}_j, \tau\rangle$ and $s^{(\lambda,\alpha)}_j \coloneqq \langle X^{(\lambda,\alpha)}_j, \sigma\rangle$.    The state $\rho$ can be transformed into $\sigma$ covariantly with respect to a symmetry group $G$ (i.e. $\rho \xrightarrow{G} \sigma$) if there exists a family of valid quantum states 
\begin{align}
\left\{\forall \theta^{(\lambda,\alpha)}_j \in [0,2\pi) :  \tau_\Theta = \sum_{\lambda,\alpha, j} t^{(\lambda,\alpha)}_j X^{(\lambda,\alpha)}_j  \right\},    
\end{align}
where each $(\lambda,\alpha, j)$-component of $\tau_\Theta$ satisfies either:
   \begin{align}
       f^\lambda_j(\rho)t^{(\lambda,\alpha)}_j &= s^{(\lambda,\alpha)}_j , \text{ or,} \\
   f^\lambda_j(\rho)\abs{t^{(\lambda,\alpha)}_j}  - \abs{s^{(\lambda,\alpha)}_j} &\ge 0, \quad  t^{(\lambda,\alpha)}_j = e^{i\theta^{(\lambda,\alpha)}_j} \abs{t^{(\lambda,\alpha)}_j}.
\end{align}
\label{lemma:sufficient_conditions_tau_undecided}
\end{lemma}

\begin{proof}
We begin with Corollary~\ref{corollary:mode_decomp_sc} and make the explicit choice for each member in our family of covariant channels $\{\E^\eta\}$ to be a PGM-and-prepare channel $\E^\eta \coloneqq \E^\eta_{\mathrm{pgm}}$, as defined in \eqref{eq:pgm_channel}. Combined with Lemma~\ref{lemma:eta_pgm_prod} this gives rise to the following sufficient condition on the transition $\rho \xrightarrow{G} \sigma$: 
\begin{align}
\sum_{\lambda,j} \left\{ \langle \eta^\lambda_j , f_j^\lambda (\rho) \tau_j^\lambda - \sigma_j^\lambda \rangle + \langle \eta^{\lambda^*}_j , f_j^{\lambda} (\rho) \tau_j^{\lambda^*} - \sigma_j^{\lambda^*} \rangle \right\} \ge 0, \quad \forall  \eta.
\label{eq:bldsfhsdl}
\end{align}
We emphasise that the choice of preparation state $\tau$ \emph{can} vary with $\eta$, though this dependence has been suppressed in our notation for clarity. We now introduce the simplifying notation $\bm{\mu} \coloneqq \{\lambda, \alpha, j \}$ and define the following coefficients of the states $\eta, \tau,$ and $\sigma$ in the ITO basis $\{ X^{\bm{\mu}} \} $:
\begin{align}
n^{\bm{\mu}} &\coloneqq \langle X^{\bm{\mu}}, \eta \rangle = \abs{n^{\bm{\mu}}} e^{i \varphi_\eta^{\bm{\mu}}}, \\
t^{\bm{\mu}} &\coloneqq \langle X^{\bm{\mu}}, \tau \rangle  = \abs{t^{\bm{\mu}}} e^{i \varphi_\tau^{\bm{\mu}}}, \\
s^{\bm{\mu}} &\coloneqq \langle X^{\bm{\mu}}, \sigma \rangle  = \abs{s^{\bm{\mu}}} e^{i \varphi_\sigma^{\bm{\mu}}}, 
\end{align}
such that $\eta_j^\lambda = \sum_\alpha n^{\bm{\mu}}X^{\bm{\mu}}$ etc. 
By the hermiticity of $\eta$, we have $n^{\bm{\mu}^*} \coloneqq n^{(\lambda^*,\alpha)}_j = \left(n^{(\lambda,\alpha)}_j\right)^* = (n^{\bm{\mu}})^*$ etc. We further introduce the notation
\begin{align}
    f^{\bm{\mu}} (\rho) \coloneqq f^\lambda_j (\rho), \forall \alpha.  
\end{align}
Substituting these definitions into Eq.~(\ref{eq:bldsfhsdl}) and using the orthonormality of the ITO basis, $\langle X^{\bm{\mu}}, X^{\bm{\nu}} \rangle = \delta_{\bm{\mu},\bm{\nu}}$, gives the following sufficient condition:
\begin{align}
\sum_{\bm{\mu}} \left[n^{\bm{\mu}^*} \left(f^{\bm{\mu}}(\rho) t^{\bm{\mu}} - s^{\bm{\mu}}\right) + n^{\bm{\mu}}\left(f^{\bm{\mu}}(\rho) t^{\bm{\mu}^*} - s^{\bm{\mu}^*}\right)\right] \ge 0  , \quad \forall  \eta.
\end{align}
This can instead be written as
\begin{align}
\sum_{\bm{\mu}}  \Re \{ n^{\bm{\mu^*}} \left(f^{\bm{\mu}}(\rho) t^{\bm{\mu}} - s^{\bm{\mu}}\right)  \} = \sum_{\bm{\mu}} \abs{n^{\bm{\mu^*}}} \Re \left\{  e^{-i \varphi_\eta^{\bm{\mu}}}\left[f^{\bm{\mu}}(\rho) t^{\bm{\mu}} - s^{\bm{\mu}}\right]  \right\}  \ge 0  , \quad \forall  \eta.
\label{eq:blahblahblah}
\end{align}

We can always choose $X^{\bm{0}} \coloneqq \frac{\mathbbm{1}}{\sqrt{d}}$, which means all other ITO basis elements must be traceless. Since we must choose $\tau$ to be a valid quantum state, this means we must assign
\begin{align}
	\tr(\tau) = \sum_{\bm{\mu}} t^{\bm{\mu}} \tr(X^{\bm{\mu}}) = \sqrt{d} t^{\bm{0}} = 1.
\end{align}
As $\sigma$ is a valid quantum state, by similar logic we conclude that $\sqrt{d} s^{\bm{0}} = 1$. Due to \ref{propertyf:zero_mode}, this means \eqref{eq:blahblahblah} reduces to
\begin{align}
\sum_{\bm{\mu} \neq \bm{0}} \abs{n^{\bm{\mu^*}}} \Re \left\{  e^{-i \varphi_\eta^{\bm{\mu}}}\left[f^{\bm{\mu}}(\rho) t^{\bm{\mu}} - s^{\bm{\mu}}\right]  \right\}  \ge 0  , \quad \forall  \eta.
\label{eq:blahblahblah2}
\end{align} 

Given any $\bm{\nu} \neq \bm{0}$, there exists a valid reference state $\eta = \sum_{\bm{\mu}} n^{\bm{\mu}} X^{\bm{\mu}}$ where
\begin{align}
n^{\bm{\mu}}=
	\begin{cases}
		\frac{1}{\sqrt{d}}& \text{ for } \bm{\mu} = \bm{0}\\
		n \text{ such that } \frac{1}{2d} \ge n > 0& \text{ for } \bm{\mu} = \bm{\nu}, \bm{\nu^*}\\
		0 & \text{ otherwise,} 
	\end{cases}
	\label{eq:ref_state_single_mode_pair}
\end{align}
since these assignments evidently lead to $\eta$ being Hermitian and trace 1, and we can further verify, for any pure state $\ket{\psi}$ of the reference system, that
\begin{align}
	\bra{\psi}\eta\ket{\psi} = \sum_{\bm{\mu}}n^{\bm{\mu}} \bra{\psi} X^{\bm{\mu}} \ket{\psi} = \frac{1}{d} - n(\bra{\psi} X^{\bm{\nu}} \ket{\psi} + \bra{\psi} X^{\bm{\nu^*}} \ket{\psi}) \ge \frac{1}{d} - 2n \ge \frac{1}{d} - \frac{2}{2d} = 0 \Rightarrow \eta \ge 0.
\end{align}

The only way that \eqref{eq:blahblahblah2} can be satisfied for reference states with components assigned according to \eqref{eq:ref_state_single_mode_pair} is if for each $\bm{\mu}$-component, we have
\begin{align}
   \Re \left\{  e^{-i \varphi_\eta^{\bm{\mu}}}\left[f^{\bm{\mu}}(\rho) t^{\bm{\mu}} - s^{\bm{\mu}}\right]  \right\}  \ge 0  , \quad \forall \eta.
    \label{eq:ur_conditions_2}
\end{align}

Recall that we are free to choose $\tau$ as we like for \emph{every} $\eta$, and our aim here is to derive a set of conditions that are independent of $\eta$. For a given component $\bm{\mu}$, one possible way of satisfying \eqref{eq:ur_conditions_2} independently of $\eta$ is if we can choose a valid quantum state $\tau= \sum_{\bm{\mu}} t^{\bm{\mu}} X^{\bm{\mu}}$ such that
\begin{align}
    f^{\bm{\mu}}(\rho) t^{\bm{\mu}} = s^{\bm{\mu}} .
\end{align}

Alternatively, note that we can rewrite \eqref{eq:ur_conditions_2} as follows:
\begin{align}
    f^{\bm{\mu}}(\rho) \abs{t^{\bm{\mu}} } \cos{( \varphi_\tau^{\bm{\mu}}-\varphi_\eta^{\bm{\mu}})} - \abs{s^{\bm{\mu}}} \cos{( \varphi_\sigma^{\bm{\mu}}-\varphi_\eta^{\bm{\mu}})}    \ge 0  , \quad \forall \eta.
    \label{eq:phases_suff_cond}
\end{align}

Therefore, an alternative way of satisfying the $\bm{\mu}$-component of the full set of sufficient conditions, is to set the phases of $\tau$ such that they cancel those of $\eta$
\begin{align}
  \varphi_\tau^{\bm{\mu}}=\varphi_\eta^{\bm{\mu}},
  \label{appeq:phases}
\end{align}
for all $\eta$.
Substituting this phase choice on $\tau$ into the sufficient condition in Eq.~(\ref{eq:phases_suff_cond}) gives 
\begin{align}
f^{\bm{\mu}}(\rho) \abs{t^{\bm{\mu}} }  - \abs{s^{\bm{\mu}}} \cos{( \varphi_\sigma^{\bm{\mu}}-\varphi_\eta^{\bm{\mu}})}    \ge 0  , \quad \forall \eta.
\end{align}
Noting that $\max_x \cos{x} = 1$, we see that the above equation holds for all $\eta$ if
\begin{align}
   f^{\bm{\mu}}(\rho) \abs{t^{\bm{\mu}}} \ge \abs{s^{\bm{\mu}} } ,
\end{align}
which together with \eqref{appeq:phases} gives us a second $\eta$-independent way of satisfying the $\bm{\mu}$-component of \eqref{eq:ur_conditions_2}.
Putting back in our explicit labels $\bm{\mu} = \{ \lambda, \alpha,j\}$, we have that, if the $\tau$ we chose at every $\eta$ was a valid quantum state such that
\begin{align}
    f^\lambda_j(\rho)t^{(\lambda,\alpha)}_j  = s^{(\lambda,\alpha)}_j  \, \text{  or  } \,
    f^\lambda_j(\rho)\abs{t^{(\lambda,\alpha)}_j}  \ge \abs{s^{(\lambda,\alpha)}_j} ,\ \quad \forall  \lambda, j, \alpha,
    \label{eq:sufficient_conditions}
\end{align}
then $\rho \xrightarrow{G} \sigma$, and thus Eq.~(\ref{eq:sufficient_conditions}) constitutes a sufficient condition on the transition $\rho \xrightarrow{G} \sigma$, as claimed. \end{proof}

We now present the following theorem from which \thmref{thrm:suff_cond_depolarized} in the main text follows as a corollary.

\begin{theorem} \label{thm:suff_cond_general}
Assume without loss of generality (\lemref{lemma:H_S_construction}) that $\G(\sigma)$ is full-rank. Then there exists a $G$--covariant operation transforming $\rho$ into $\sigma$ if
\begin{align}
 \lambda_{\mathrm{min}} n^{-1}f^\lambda_j(\rho)     \ge g^\lambda_j(\sigma) , \quad \forall \lambda \neq 0, j,
\end{align}
where we have defined
$f^\lambda_j(\rho) \coloneqq \tr[\rho^\lambda_j \G(\rho)^{-\frac{1}{2}}\rho^{\lambda^*}_j \G(\rho)^{-\frac{1}{2}}]$, $g^{\lambda}_j(\sigma) \coloneqq \sum_\alpha \abs{\tr[X^{(\lambda,\alpha)\dagger}_j \sigma] }$, $\lambda_{\mathrm{ min}}$ as the smallest eigenvalue of $\G(\sigma)$, and $n$ as the sum of the dimensions of all distinct non-trivial irreps appearing in the representation of $G$ on $\B(\H_B)$, with $\H_B$ being the Hilbert space of the output system.
\end{theorem}

\begin{proof} 
We begin by making an explicit choice for $\{ \tau_\Theta \}$ in Lemma~\ref{lemma:sufficient_conditions_tau_undecided} and proceed by showing that each $\tau_\Theta$ corresponds to a valid quantum state.

For the $\lambda=0$ irrep we always choose 
\begin{align}
t^{(0,\alpha)}_0 \coloneqq s^{(0,\alpha)}_0,\ \forall \alpha,
\label{eq:tau_zero_mode_sym_access}
\end{align}
which corresponds to setting $\G(\tau_\Theta) = \G(\sigma)$. Due to \ref{propertyf:zero_mode}, this choice guarantees that all zero mode conditions from Lemma~\ref{lemma:sufficient_conditions_tau_undecided} are satisfied for any given $\sigma$. 

For $\lambda \neq 0$, let us choose 
\begin{align}
   t^{(\lambda,\alpha)}_j \coloneqq e^{i \theta^{(\lambda,\alpha)}_j}  c^\lambda_j \abs{s^{(\lambda,\alpha)}_j}, \quad c^\lambda_j \coloneqq \frac{\lambda_{\min}}{n\left(\sum_{\alpha} \abs{s^{(\lambda,\alpha)}_j}\right)}, \quad \theta^{(\lambda^*,\alpha)}_j = -\theta^{(\lambda,\alpha)}_j
\end{align} where $n$ is the sum of the dimensions of all distinct non-trivial irreps appearing in the representation of $G$ on $\B(\H_B)$.

If $\tau_\Theta$ corresponds to a valid quantum state for all values of $\theta^{(\lambda,\alpha)}_j$, from Lemma~\ref{lemma:sufficient_conditions_tau_undecided} we immediately have the sufficient condition
\begin{align}
    f^\lambda_j(\rho) c^\lambda_j \ge 1, \forall \lambda \neq 0, j, 
\end{align}
which, with the identification $g^\lambda_j(\sigma) \coloneqq \sum_{\alpha} \abs{s^{(\lambda,\alpha)}_j} =\sum_{\alpha}  \abs{\tr[X^{(\lambda,\alpha)\dagger}_j \sigma] }$, gives the statement of the lemma.

All that is left to do then is demonstrate that these choices of coefficients $t^{(\lambda,\alpha)}_j$ indeed correspond to a valid quantum state $\tau_\Theta$. We first note that the choice $t^{(0,0)}_0 = s^{(0,0)}_0= \frac{1}{\sqrt{d}}$ ensures that $\tau_\Theta$ has trace 1. Furthermore, since $\tau_\Theta = \G(\sigma) + \sum_{\substack{\lambda,\alpha, j: \\ \lambda \neq 0}} t^{(\lambda,\alpha)}_j X^{(\lambda,\alpha)}_j$, where the $t^{(\lambda,\alpha)}_j$ in the sum have been so chosen that $t^{(\lambda^*,\alpha)}_j = \left(t^{(\lambda,\alpha)}_j\right)^*$, $\tau_\Theta$ is the sum of two Hermitian operators and so is also Hermitian. All that remains is to verify that $\tau_\theta$ is positive semidefinite, i.e.,
\begin{align}
\bra{\psi}\tau_\Theta \ket{\psi} \geq 0,
\end{align}
for any pure state $\ket{\psi}$ of the output system.
The left hand side of this expression can be lower bounded in the following way for any $\theta^{(\lambda,\alpha)}_j \in[0,2\pi)$:
\begin{align}
\bra{\psi}\tau_\Theta \ket{\psi} &= \sum_{\lambda,\alpha,j} t^{(\lambda,\alpha)}_j \bra{\psi} X_j^{(\lambda,\alpha)} \ket{\psi} \\
&= \sum_{\alpha: \lambda=0} s^{0,\alpha}_0 \bra{\psi} X_0^{(0,\alpha)} \ket{\psi} +  \sum_{\substack{\lambda,\alpha, j: \\ \lambda \neq 0} } e^{i \theta^{(\lambda,\alpha)}_j}c^\lambda_j \abs{s^{(\lambda,\alpha)}_j} \bra{\psi} X_j^{(\lambda,\alpha)} \ket{\psi} \\
&\ge \bra{\psi}\G(\sigma)\ket{\psi}-  \sum_{\substack{\lambda,\alpha, j: \\ \lambda \neq 0} } c^\lambda_j \abs{s^{(\lambda,\alpha)}_j} \abs{\bra{\psi} X_j^{(\lambda,\alpha)} \ket{\psi} } \\
&\ge \lambda_{\min}- \sum_{\substack{\lambda,\alpha, j: \\ \lambda \neq 0} } c^\lambda_j \abs{s^{(\lambda,\alpha)}_j}\\
&= \lambda_{\min} - \sum_{\substack{\lambda,\alpha, j: \\ \lambda \neq 0} } \frac{\lambda_{\min}}{n\left(\sum_{\alpha'} \abs{s^{(\lambda,\alpha')}_j}\right)} \abs{s^{(\lambda,{\alpha})}_j}\\
&= \lambda_{\min} - \sum_{\lambda \neq 0,j} \frac{\lambda_{\min}}{n}\\
&= \lambda_{\min} - \lambda_{\min} = 0,
\end{align}
where in the second inequality we have used the fact that the operators $\left\{X^{(\lambda,\alpha)}_j\right\}$ are normalized. This confirms that $\tau_\Theta \ge 0$ for all $\theta^{(\lambda,\alpha)}_j \in [0,2\pi)$, which completes the proof. \end{proof}

The Sandwiched $\alpha$--R\'{e}nyi divergence $D_\alpha(\rho ||\sigma)$ for two states $\rho, \sigma$ of a quantum system is defined as \cite{muller2013quantum,wilde2014strong}
\begin{equation}
    D_\alpha(\rho||\sigma):= \frac{1}{\alpha-1} \log \tr \left [\sigma^{\frac{1-\alpha}{2\alpha}}\rho \sigma^{\frac{1-\alpha}{2\alpha}} \right ]^\alpha,
\end{equation}
whenever the support of $\rho$ lies in the support of $\sigma$, and is infinite otherwise.

In the case of $\alpha = 2$, we extend the domain of the first argument to the set of all linear operators in the support of $\sigma$ in the following manner:
\begin{equation}
\label{eq:alpha_renyi_div}
D_2(X||\sigma):= \log \tr \left( \left [\sigma^{-\frac{1}{4}} X \sigma^{-\frac{1}{4}} \right ]^\dagger \left [\sigma^{-\frac{1}{4}} X \sigma^{-\frac{1}{4}} \right ]^\dagger \right).
\end{equation}
When $X$ is Hermitian, which is the case whenever $X$ is a valid quantum state, this extension reduces to the standard definition of the sandwiched $\alpha$--R\'{e}yni divergence for $\alpha = 2$. We can then summarize the theorem above as the following statement, reproduced from the main text:

\SuffCond*
\begin{proof}
We consider the transition $\rho \xrightarrow{G} \sigma_p$, where  $\sigma_p \coloneqq (1-p) \sigma +p \frac{\id}{d}$. 
We begin by noting that $\lambda_{\mathrm{min}}[\G(\sigma_p)] = (1-p)\lambda_{\mathrm{min}} + \frac{p}{d}$ and $g^\lambda_j(\sigma_p)=(1-p)g^\lambda_j(\sigma)$ when $(\lambda,j) \neq (0,0)$. 
	
Therefore, substituting $\sigma_p \coloneqq (1-p) \sigma +p \frac{\id}{d}$ into \thmref{thm:suff_cond_general} gives the sufficient condition on the transition $\rho \xrightarrow{G} \sigma_p$: 
\begin{align}
 n^{-1}\left (\lambda_{\rm min} + \frac{p}{d(1-p)} \right) f^\lambda_j (\rho) \ge g^\lambda_j (\sigma), \forall (\lambda,j)\neq 0.
\label{apeq:iiii}
\end{align}
We now express this in terms of the extended sandwiched $\alpha$--R\'{e}nyi divergence $D_\alpha(\rho ||\sigma)$ for $\alpha = 2$ defined in \eqref{eq:alpha_renyi_div}, and it is readily checked that $D_2(\rho^\lambda_j || \G(\rho)) = \log f^\lambda_j (\rho)$ for all $\lambda,j$. Since $\log(x)$ is monotonic in $x$, we can re-express \eqref{apeq:iiii} as in the statement of the theorem, completing the proof. \end{proof}

\subsection{Conditions for identical input and output systems} 
When the input system $A$ and output system $B$ are identical, one can, without loss of generality, replace $\E^\eta$ in Lemma~\ref{lemma:general_SC_hilbert_schmidt} with $(1-q) \I + q \E^\eta$, where $q$ is a probability and $\I$ is the identity channel. The conditions of Lemma~\ref{lemma:general_SC_hilbert_schmidt} are then rewritten as:
\begin{align}
(1-q) \langle \eta,\rho\rangle +  q \langle \eta, \E^\eta(\rho) \rangle \ge \langle \eta, \sigma\rangle , \mbox{ for all }\eta \mbox{ and any }q\in[0,1],
\end{align}
for an arbitrary family of covariant channels $\{\E^\eta\}$ parameterised by $\eta$, which can be rearranged as
\begin{align}
q \langle \eta,\E^\eta(\rho) \rangle \ge \langle \eta, (\sigma - (1-q)\rho)\rangle , \mbox{ for all }\eta \mbox{ and any }q \in (0,1],
\end{align}
and so we define 
\begin{equation}
\sigma(q) := \sigma - (1-q)\rho, 
\end{equation}
for any $q \in [0,1]$ for which we require
\begin{align}
q \langle \eta, \E^\eta(\rho) \rangle \ge \langle \eta, \sigma(q) \rangle, \mbox{ for all }\eta \mbox{ and any }q \in [0,1].
\label{eq:dkdkdkdkd}
\end{align}

Therefore, when the input and output systems are the same, Lemma~\ref{lemma:sufficient_conditions_tau_undecided} can be upgraded to
\begin{lemma}	
	Let us define components $t^{(\lambda,\alpha)}_j \coloneqq \langle X^{(\lambda,\alpha)}_j, \tau\rangle$ and $s(q)^{(\lambda,\alpha)}_j \coloneqq \langle X^{(\lambda,\alpha)}_j, \sigma(q) \rangle$.  When the input system $A$ and the output system $B$ are identical, the state $\rho$ can be transformed into $\sigma$ covariantly with respect to a symmetry group $G$ (i.e. $\rho \xrightarrow{G} \sigma$) if there exists a family of valid quantum states 
	\begin{align}
	\left\{\forall \theta^{(\lambda,\alpha)}_j \in [0,2\pi) :  \tau_\Theta = \sum_{\lambda,\alpha, j} t^{(\lambda,\alpha)}_j X^{(\lambda,\alpha)}_j  \right\},    
	\end{align}
	where each $(\lambda,\alpha, j)$-component of $\tau_\Theta$ satisfies either:
	\begin{align}
	qf^\lambda_j(\rho)t^{(\lambda,\alpha)}_j &= s(q)^{(\lambda,\alpha)}_j , \text{ or,} \\
	qf^\lambda_j(\rho)\abs{t^{(\lambda,\alpha)}_j}  - \abs{s(q)^{(\lambda,\alpha)}_j} &\ge 0, \quad  t^{(\lambda,\alpha)}_j = e^{i\theta^{(\lambda,\alpha)}_j} \abs{t^{(\lambda,\alpha)}_j}.
	\end{align}
	for some probability $q \in [0,1]$.
	\label{lemma:sufficient_conditions_tau_undecided_q}
\end{lemma}
\begin{proof}
	The proof follows that of \lemref{lemma:sufficient_conditions_tau_undecided} with $\E^\eta \rightarrow q \E^\eta$ and $\sigma \rightarrow \sigma(q)$.
\end{proof}

These conditions are identical to those in \lemref{lemma:sufficient_conditions_tau_undecided} under $f^\lambda_j(\rho) \rightarrow q f^\lambda_j(\rho)$ and $\sigma \rightarrow \sigma(q)$. Making use of this lemma, we find that, when the input and output systems are the same, Theorem \ref{thm:suff_cond_general} can be similarly upgraded to
\begin{theorem} \label{thm:suff_cond_general_q}
	Assume that $\G(\sigma)$ is full-rank. When the input and output systems are identical, there exists a $G$--covariant operation transforming $\rho$ into $\sigma$ if $\rho = \sigma$ or if there exists a probability $q \in (0,1]$ such that
	\begin{align}
	\lambda_{\mathrm{min}}(\G[\sigma(q)]) n^{-1} f^\lambda_j(\rho)     \ge g^\lambda_j(\sigma(q)) , \quad \forall \lambda \neq 0, j,
	\end{align}
	and $\G(\sigma(q)) \ge 0$, where we have defined
	$f^\lambda_j(\rho) \coloneqq \tr[\rho^\lambda_j \G(\rho)^{-\frac{1}{2}}\rho^{\lambda^*}_j \G(\rho)^{-\frac{1}{2}}]$, $g^{\lambda}_j(\sigma(q)) \coloneqq \sum_\alpha \abs{\tr[X^{(\lambda,\alpha)\dagger}_j \sigma(q)] }$, $n$ as the sum of the dimensions of all distinct non-trivial irreps appearing in the representation of $G$ on $\B(\H_B)$, with $\H_B$ being the Hilbert space of the output system, and $\lambda_{\mathrm{min}}(\G[\sigma(q)])$ as the smallest eigenvalue of $\G(\sigma(q))$. 
\end{theorem}
\begin{proof}
	The only covariantly accessible state at $q=0$ from $\rho$ is $\rho$ itself. At all other values of $q$, the proof follows that of Theorem \ref{thm:suff_cond_general} with the substitutions $f^\lambda_j(\rho) \rightarrow q f^\lambda_j(\rho)$ and $\sigma \rightarrow \sigma(q)$, except for the following caveat. The part of that proof demonstrating why $\tau_\Theta$ is positive relies on $\G(\tau_\Theta)$ being positive, which may not be true for all $q \neq 0$, since we now assign $\G(\tau_\Theta) \coloneqq q^{-1} \G(\sigma(q))$, but because $q^{-1} \sigma(q)$ or equivalently $\G(\sigma_q)$ may not be positive for all $q \neq 0$, $q^{-1} \G(\sigma(q))$ may not be either. We therefore additionally impose $\G(\sigma(q)) \ge 0$.
\end{proof}
 
By applying Theorem \ref{thm:suff_cond_general_q} to a partially depolarised version of $\sigma$, $\sigma_p \coloneqq (1-p)\sigma + p\frac{\id}{d}$, and noting once again that $ D_2 ( \rho^\lambda_j || \G(\rho)) = \log f^\lambda_j(\rho)$, we obtain \thmref{thrm:suff_cond_q} of the main text as an immediate corollary. We note that the assumption of $\G(\sigma)$ being full-rank can be dropped by using Lemma \ref{lemma:H_S_construction} to justify truncating the output Hilbert space to $\supp(\G(\sigma))$, and then redefining $\sigma_p \rightarrow (1-p)\sigma + p\frac{\id}{d_S}$ where $d_S$ is the dimension of $\supp[\G(\sigma)]$, $\lambda_{\mathrm{ \min}}[\cdot]$ becomes $\min_{\ket{\psi} \in \H_S} \bra{\psi}(\cdot)\ket{\psi}$, and $n$ becomes the sum of the dimensions of all distinct non-trivial irreps appearing in the representation of $\G$ on $\B(\supp[\G(\sigma)])$.
\end{appendices}

\end{document}